\documentclass[draftcls, onecolumn]{IEEEtran}

\usepackage[cmex10]{amsmath}  
\usepackage{amssymb}
\usepackage{amsfonts}
\usepackage{bm}

\usepackage{dsfont}
\usepackage{times}

\usepackage[noadjust]{cite}
\usepackage{graphicx}

\newtheorem{thm}{Theorem}
\newtheorem{rem}{Remark}
\newtheorem{cor}{Corollary}
\newtheorem{lem}{Lemma}
\newtheorem{dfn}{Definition}

\def\mb{\mathbf}
\def\opn{\operatorname*}
\def\bs{\boldsymbol}

\def\ds{\mathds}
\def\mc{\mathcal}
\def\ss#1{{\sf #1}}

\def\esc#1{\textrm{#1}}
\def\resc#1{\uppercase{#1}}

\def\vec#1{\mb{#1}}
\def\rvec#1{\bs{\uppercase{#1}}} 
\def\rvecr#1{\bs{\lowercase{#1}}} 
\def\vecs#1{\bs{#1}}

\def\rvecIndc#1{\rvec{#1}^{\textrm{ind cond.}}}

\def\mat#1{\mb{\uppercase{#1}}}
\def\mats#1{\bs{#1}}

\def\R{\ds{R}}

\def\SymM#1{\ds{S}^{#1}}
\def\PSD#1{\SymM{#1}_+}

\def\N#1#2{\mc{N}\left(#1,#2\right)}
\def\NZ#1{\N{\vecs{0}}{#1}}

\def\pdffun{P}

\def\pdf#1#2{\pdffun_{\rvec{#1}}(#2)}
\def\Parg#1{\pdffun_{\rvec{#1}}}
\def\EspOp{\ss{E}}

\newcommand{\Esp}[2][5]{%
  \ifcase#1
     \EspOp\{ #2 \}
     \or \EspOp \bigl\{ #2 \bigr\}
     \or \EspOp \Bigl\{ #2 \Bigr\}
     \or \EspOp \biggl\{ #2 \biggr\}
     \or \EspOp \Biggl\{ #2 \Biggr\}
  \else
     \EspOp \left\{ #2  \right\}
\fi}

\newcommand{\Earg}[3][5]{%
  \ifcase#1
     \EspOp_{#3} \{ #2 \}
     \or \EspOp_{#3} \bigl\{ #2 \bigr\}
     \or \EspOp_{#3} \Bigl\{ #2 \Bigr\}
     \or \EspOp_{#3} \biggl\{ #2 \biggr\}
     \or \EspOp_{#3} \Biggl\{ #2 \Biggr\}
  \else
     \EspOp_{#3} \left\{ #2  \right\}
\fi}

\newcommand{\CEsp}[3][5]{%
  \ifcase#1
     \EspOp\{ #2 \mid #3 \}
     \or \EspOp \bigl\{ #2 \bigm\vert #3 \bigr\}
     \or \EspOp \Bigl\{ #2 \Bigm\vert #3 \Bigr\}
     \or \EspOp \biggl\{ #2 \biggm\vert #3 \biggr\}
     \or \EspOp \Biggl\{ #2 \Biggm\vert #3 \Biggr\}
  \else
     \EspOp \left\{ #2  \,\middle\vert\, #3 \right\}
\fi}

\newcommand{\Diag}[2][5]{%
  \ifcase#1
     \mb{Diag}( #2 )
     \or \mb{Diag} \bigl( #2 \bigr)
     \or \mb{Diag} \Bigl( #2 \Bigr)
     \or \mb{Diag} \biggl( #2 \biggr)
     \or \mb{Diag} \Biggl( #2 \Biggr)
  \else
     \mb{Diag} \left( #2  \right)
\fi}

\newcommand{\diag}[2][5]{%
  \ifcase#1
     \mb{diag}( #2 )
     \or \mb{diag} \bigl( #2 \bigr)
     \or \mb{diag} \Bigl( #2 \Bigr)
     \or \mb{diag} \biggl( #2 \biggr)
     \or \mb{diag} \Biggl( #2 \Biggr)
  \else
     \mb{diag} \left( #2  \right)
\fi}

\def\Jacob{{\ss D}}

\def\Tr{\ss{Tr}}

\def\exp{\ss{exp}}
\def\log{\ss{log}}

\def\vecop{\ss{vec}}

\def\pinv{\textrm{\footnotesize{+}}}

\def\d{\opn{d}\!}

\def\T{\ss{T}}

\def\nume{\esc{e}}

\def\ones{\mats{1}}

\def\ie{i.e.}

\def\iid{i.i.d. }
\def\etal{et al. }

\def\EM{\mat{E}}
\def\MSE#1{\EM_{\rvec{#1}}}

\def\Sym#1{\mat{N}_{#1}}
\def\Dup#1{\mat{D}_{#1}}

\def\dim{n}
\def\I{I} 
\def\Ent{h} 

\def\mmse{\ss{mmse}}
\def\scalart{\gamma}
\def\f{f}
\def\ft{\f(\scalart)}
\def\q{q}

\def\CM{\mats{\Phi}}
\def\EM{\mat{E}}
\def\MSE#1{\EM_{\rvec{#1}}}

\def\CMSE#1#2{\CM_{\rvec{#1}}(\rvec{#2})}
\def\CMSEr#1#2{\CM_{\rvec{#1}}(\rvecr{#2})}

\def\JM{\mat{J}}

\def\J#1{\JM_{\rvec{#1}}}

\def\vectort{t}
\def\vectortEq{t_e}
\def\vectorT{\vecs{\phi}}
\def\Q{\mat{Q}}
\def \pQ{ \widetilde{\Q} }
\def\Qt{\Q(\vectort)}
\def\QT{\Q(\vectorT)}

\def\BQ{\ss{d}_i}
\def\BQsum{\ss{d}}
\def\Igen#1#2{ \I \left(  #1;#2 \right)}
\def\Icond#1#2#3{ \I \left( #1;#2 | #3 \right)}
\def\timeIn{m}
\def\Power{P}
\def\Identity{\mat{I}}
\def\rate{\ss{R}}
\def\MSElin#1{\MSE{#1}^{\ss{lin}}}
\def\Fisher{ \mat{J}}
\def\Gradient{ \nabla }
\def\QFisher{ \mat{W} }

\def\HdC#1#2{ \ss{h} \left( #1  | #2 \right)}

\def\CovMat{\mat{R}}
\def\Cov#1{\CovMat_{\rvec{#1}}}

\def\egvaliCov#1#2{\lambda_{\rvec{#1}, #2}}
\def\Chan{\mat{H}}
\def\ChanDChan{\mat{B}}
\def\snr{\ss{snr}}

\begin{document}

\title{On MMSE Properties and I-MMSE Implications in Parallel MIMO Gaussian
Channels}
\author{Ronit Bustin, Miquel Payar\'o, Daniel P.~Palomar, and Shlomo Shamai
(Shitz)}

\maketitle

\vspace{-1cm}

\begin{abstract}
The scalar additive Gaussian noise channel has the ``single crossing point'' property between the minimum-mean
square error (MMSE) in the estimation of the input given the channel output, assuming a Gaussian input to the
channel, and the MMSE assuming an arbitrary input. This paper extends the result to the parallel MIMO additive Gaussian
channel in three phases: i) The channel matrix is the identity matrix, and we limit the Gaussian input to a vector of Gaussian \iid elements. The ``single crossing point'' property is with respect to
the $\snr$ (as in the scalar case). ii) The channel matrix is arbitrary, the Gaussian input is limited to an independent Gaussian input. A ``single crossing point'' property is derived for each diagonal element of the MMSE
matrix. iii) The Gaussian input is allowed to be an arbitrary Gaussian random vector. A ``single crossing point'' property is derived for each eigenvalue of the MMSE matrix.

These three extensions are then translated to new information
theoretic properties on the mutual information, using the fundamental relationship between estimation theory and
information theory. The results of the last phase are also translated to a new
property of Fisher's information. Finally, the applicability of all three extensions on information theoretic problems
is demonstrated through: a proof of a special case of Shannon's vector EPI, a converse proof of the capacity region
of the parallel \emph{degraded} MIMO broadcast channel (BC) under per-antenna power constrains and under covariance
constraints, and a converse proof of the capacity region of the compound parallel \emph{degraded} MIMO BC under
covariance constraint.
\end{abstract}

%
%

\section{Introduction} \label{sec:introduction}

This paper considers parallel multiple-input multiple-output (MIMO) channels, with an arbitrary input distribution and additive standard Gaussian noise. These channels are a subset of the important family of
MIMO additive Gaussian noise channels, which have been extensively investigated in the literature.
For most Gaussian channel models studied in information theory, Gaussian signaling happens to be optimal, from point-to-point channels, to multiple-access channels (MAC), and broadcast channels (BC) \cite[Ch.~9 and 15]{CoverThomas} \cite{Hanan}.
The methods used to prove this optimality were not easy to come across, even when considering scalar Gaussian channels. For example, in order to prove that Gaussian inputs are optimal for the scalar Gaussian BC, Bergmans employed Shannon's entropy power inequality (EPI) \cite{Bergmans74}. The solution for the MIMO Gaussian BC came only 30 years later in \cite{Hanan}, using a new enhancement approach. Since then, several other proofs were derived, using different tools, such as, the extremal inequality in \cite{External}, the de Bruijn identity in coordination with Dembo's inequality in \cite{EkremUlukusMIMO},
and the ``single crossing point'' property presented by Guo \etal in \cite{PROP_full}.
The ``single crossing point'' stemmed from the I-MMSE relationship, a fundamental relationship between estimation theory and information theory revealed by Guo, Shamai, and Verd$\acute{\textrm{u}}$ in \cite{IMMSE}.

The relationship between estimation theory and information theory goes back to the late 1950's, when Stam \cite{STAM} used the de Bruijn's identity to prove Shannon's EPI, and then in the early 1970's when the mutual information was represented as a function of the causal filtering error by Duncan \cite{Duncan} and Kadota, Zakai and Ziv \cite{KadotaZivZakai}. The I-MMSE relationship, given for discrete-time and continuous-time, scalar and vector additive Gaussian noise channels, deepens the connection between these two fields. Specifically, for a scalar additive Gaussian noise channel,
\begin{gather} \label{eq:scalarModel}
 \resc{y} = \sqrt{\snr} \resc{x} + \resc{n}
\end{gather}
where $\resc{n}$ is standard Gaussian additive noise, then, regardless of the input distribution of $\resc{x}$, the mutual information, $\Igen{\resc{x}}{\resc{y}}$, and minimum-mean-square error (MMSE) in the estimation of $\resc{X}$ given the observation $\resc{Y}$, $\mmse(\resc{x}, \snr)$, are related (assuming
real-valued inputs/outputs) by
\begin{align}
\Jacob_{\snr} \Igen{\resc{x}}{\resc{y}} = \frac{1}{2} \mmse(\resc{x}, \snr)
\end{align}
where $\Jacob_{\snr}$ is the derivative with respect to $\snr$, and
\begin{gather} \label{eq:scalarMMSE}
 \mmse(\resc{x}, \snr) = \Esp[1]{ (\resc{x} - \CEsp[0]{\resc{x}}{\sqrt{\snr}
\resc{x} + \resc{n}})^2} .
\end{gather}
The work in \cite{IMMSE} has been extended in several directions, among which we have: the additive Poisson noise channel \cite{IMMSE_Poisson,WeissmanAtar_IMMSEPoisson}, the general additive noise channel \cite{IMMSE_nonGaussian}, arbitrary channels \cite{Palomar2}, representation of the relative entropy as a function of the difference between the mismatched MMSE and the matched MMSE in \cite{VerduIT,WeissmanMismatched}, and others. One important extension, on which we heavily rely, is the one done by Palomar and Verd$\acute{\textrm{u}}$ in \cite{Palomar}, where they obtain the gradient of the mutual information with respect to different parameters of the MIMO channel.

Going back to the ``single crossing point'' property, one of the goals in \cite{PROP_full} was to show the applicability of the I-MMSE relationship as a tool to solve information-theoretic problems. Specifically, the authors of \cite{PROP_full} examined the scalar Gaussian BC and gave an alternative proof for the optimality of Gaussian inputs.
In order to show this, Guo \etal defined the following function in \cite{PROP_full}:
\begin{gather} \label{eq:scalarF}
\f(\resc{x}, \scalart) = (1 + \scalart)^{-1} - \mmse(\resc{x}, \scalart)
\end{gather}
where the simplified notation $\ft$ will be used when there is no confusion about the distribution of $\resc{x}$. It was shown that $\ft$ has \emph{at most} a single crossing point of the horizontal axis. In other words, the first term, which is the MMSE assuming a standard Gaussian input, may be smaller than the second term in some range of $\snr$ values (note that the parameter $\scalart$ is the $\snr$); however, once the two terms are equal, at some $\gamma_0$, the MMSE of the standard Gaussian input remains greater than the MMSE of the arbitrary input for all $\gamma > \gamma_0$, and the function remains nonnegative.
This property together with the I-MMSE relationship, provides the missing link to derive a simple and elegant converse proof of the capacity region of the scalar Gaussian BC.

The ``single crossing point'' was derived only for the scalar additive Gaussian channel, as can be seen from the definition of the function $\ft$. The motivation of this work is to extend this property to the vector Gaussian channel.
This extension is done in three phases. First, we consider random input vectors instead of random input scalars, but keep the dependence on a scalar quantity -- the channel $\snr$. In this setting we also limit the Gaussian input to a vector of Gaussian \iid elements. In the second and third phases we consider dependence on a vector quantity -- a parallel channel matrix. The difference between the second and third phases is that in the second phase we limit the Gaussian input to an independent random vector, whereas in the third phase the Gaussian input is arbitrary. In all three phases, the general channel model considered is the following:
\begin{gather} \label{eq:generalModel}
 \rvec{y} = \Chan \rvec{x} + \rvec{n}
\end{gather}
where $\rvec{n}$ is a standard Gaussian random vector, and $\Chan$ is a square and diagonal channel matrix known to the receiver(s). In the vector case, the scalar MMSE does not capture all the needed information, and we need to resort to the matrix extension, the MMSE matrix defined as
\begin{gather} \label{eq:MMSEmatrix}
\MSE{x} = \Esp[1]{ (\rvec{x} - \CEsp{\rvec{x}}{\Chan \rvec{x} + \rvec{n}})
(\rvec{x} - \CEsp{\rvec{x}}{\Chan \rvec{x} + \rvec{n}})^\T}
\end{gather}
from which we can see that, in general, the MMSE matrix $\MSE{x}$
depends on the channel $\MSE{x} = \MSE{x}(\Chan)$, but whenever the channel
coefficients depend on other parameters $\Chan = \Chan(\vectorT)$, we will
write $\MSE{x}(\vectorT)$.
Observe that the standard scalar MMSE value in the vector case can be easily recovered from the MMSE matrix as follows:
\begin{gather} \label{eq:scalarMMSE_vector_arg}
 \mmse(\rvec{x}, \snr) = \Esp[1]{ \Vert\rvec{x} - \CEsp[0]{\rvec{x}}{\sqrt{\snr}
\rvec{x} + \rvec{n}}\Vert^2} = \Tr(\MSE{x}(\sqrt{\snr} \mat{I}_\dim)) .
\end{gather}
For the important case when the input distribution of $\rvec{x}$ is Gaussian with covariance matrix $\Cov{x_G}$ we will use the following notation:
\begin{gather} \label{eq:MMSEmatrix_Gaussin}
\EM_{G}(\Cov{x_G}, \Chan) = (\Cov{x_G}^{-1} + \Chan^\T \Chan)^{-1}
\end{gather}
where we assumed that $\Cov{x_G}$ is of full rank. As in the case of $\MSE{x}$, whenever the channel coefficients depend on other parameters $\Chan = \Chan(\vectorT)$, we will write $\EM_{G}(\Cov{x_G}, \vectorT)$.
Another important quantity is the MMSE given for a specific output, $\rvec{Y} = \rvecr{y}$, defined as:
\begin{gather} \label{eq:CMMSE_matrix}
\CMSEr{x}{y} = \Esp[1]{ (\rvec{x} - \CEsp{\rvec{x}}{\rvecr{y}})
(\rvec{x} - \CEsp{\rvec{x}}{\rvecr{y} })^\T | \rvecr{y}}.
\end{gather}
Although not specified explicitly, $\CMSEr{x}{y}$ depends on the channel matrix/parameters.
Note that $\MSE{x}(\Chan) = \Esp[1]{ \CMSE{x}{y}}$. Interestingly, when the input distribution of $\rvec{x}$ is Gaussian, $\CMSEr{x}{y}$ is independent of $\rvecr{y}$ and the following equality holds for all $\rvecr{y}$: $\EM_{G}(\Cov{x_G}, \Chan) = \CMSEr{x}{y}$.
Finally, given all these quantities we can define the main player in this work: the MMSE matrix difference (analog to $\ft$ in the scalar case)
\begin{align}
\Q(\rvec{x}, \Cov{x_G}, \vectorT)
  &= \EM_{G}(\Cov{x_G}, \vectorT) - \MSE{x}(\vectorT) \\ 
  &= \MSE{x_G}(\vectorT) - \MSE{x}(\vectorT) \label{eq:vectorQ_general} \\
  &= (\Cov{x_G}^{-1} + (\Chan(\vectorT))^\T \Chan(\vectorT))^{-1} -
    \MSE{x}(\vectorT) \label{eq:vectorQ}
\end{align}
where, similarly to the scalar case in (\ref{eq:scalarF}), we will use the simplified notation $\QT$
when the distribution of $\rvec{x}$ and the covariance matrix of the Gaussian
distribution $\Cov{x_G}$ are clear from the context. Note that there is \emph{no} requirement that the covariance of the random vector $\rvec{x}$ be equal to $\Cov{x_G}$.

As it has already been pointed out, the extension from scalar-to-vector is done in three phases. In the first step the dependence remains on a scalar parameter -- the $\snr$. This is obtained by setting $\Chan = \Identity$ in the general MIMO model in (\ref{eq:generalModel}). We further limit our observation to the comparison of an arbitrary input distribution with the subset of Gaussian random vectors with \iid elements. For this case, we show that the ``single crossing point'' property extends smoothly to any linear combination of $\Q(\rvec{x}, \Cov{x_G}, \vectorT)$ with a positive semidefinite matrix. Although this is the simplest scalar-to-vector extension, the proof is not straightforward.
In order to demonstrate the applicability of this result we extend the proof of a special case of Shannon's EPI, done in \cite{PROP_full}, to the vector case.

Proceeding with the scalar-to-vector extension, we assume that the channel matrix, $\Chan$, is parallel, thus our dependence is now on a vector parameter. In this setting we have two distinguishable results, given in phases two and three, that cannot be trivially deduced from each other. In phase two, we limit the Gaussian distribution, to which we compare, to any independent Gaussian distribution characterized by its diagonal covariance matrix, $\mat{\Lambda}_{\rvec{x_G}}$. Under this assumption we show that a ``single crossing point'' property exists for each and every diagonal element of $\Q(\rvec{x}, \mat{\Lambda}_{\rvec{x_G}}, \vectorT)$. Together with the I-MMSE relationship, this result provides some interesting properties of the mutual information, and its applicability is demonstrated by providing a simple converse proof for the parallel Gaussian BC capacity region under per-antenna power constraints.

The third phase, which is the main result of this work, does not require any further assumptions (apart from the diagonal channel matrix). That is, we compare an arbitrary input distribution with any general Gaussian input distribution, with covariance $\Cov{x_G}$. In this setting we show that a ``single crossing point'' property exists for each and every eigenvalue of the matrix $\Q(\rvec{x}, \Cov{x_G}, \vectorT)$. The applicability of this result is demonstrated with two information-theoretic problems: the converse proof of the parallel Gaussian BC capacity region under covariance constraint and the converse proof of the compound parallel Gaussian BC capacity region under covariance constraint.


Much of this work regards the behavior of functions around zeros, the existence and amount of actual crossings of the horizontal axis. Thus, before proceeding with the technical content of the paper and, in order to make these observations rigorous, we require the next definitions which will be used throughout the paper.
\begin{dfn} \label{dfn:negative2nonegative}
Given a function $h(t)$ continuous within the neighborhood of $t_0$, we say that a negative-to-nonnegative zero crossing occurs at $t=t_0$ if, and only if, $h(t_0)=0$ and there exists a positive value $\epsilon$ such that $h(t) < 0$ for $t \in (t_0 - \epsilon, t_0)$ and $h(t) \geq 0$ for $t \in (t_0, t_0 + \epsilon)$.
\end{dfn}
\begin{dfn} \label{dfn:nonnegative2negative}
Given a function $h(t)$ continuous within the neighborhood of $t_0$, we say that a nonnegative-to-negative zero crossing occurs at $t=t_0$ if, and only if, $h(t_0)=0$ and there exists a positive value $\epsilon$ such that $h(t) \geq 0$ for $t \in (t_0 - \epsilon, t_0)$ and $h(t) < 0$ for $t \in (t_0, t_0 + \epsilon)$.
\end{dfn}
Similar definitions can be given for positive-to-nonpositive and nonpositive-to-positive zero crossings. Another required definition is the following:
\begin{dfn} \label{dfn:negative-zero-positive}
Given a function $h(t)$ continuous within the neighborhood of $t_0$, we say that a negative-zero-positive crossing occurs at $t=t_0$ if, and only if, a negative-to-nonegative zero crossing occurs as $t=t_0$ and there exists a positive $\delta$ such that $h(t)=0$ for $t \in (t_0, t_0 + \delta)$ and a nonpositive-to-positive zero crossing occurs as $t_0 + \delta$.
\end{dfn}
Similarly we can define a positive-zero-negative crossing.

The remaining of this paper is organized as follows: Section \ref{sec:scalarMIMO} considers the first phase of our extension from scalar-to-vector, in which case the dependence is on the scalar parameter, $\snr$. In Section \ref{sec:scalartovector} we provide the framework in which we handle the assumption of a parallel channel matrix, $\Chan$. This framework is relevant for phases two and three of our scalar-to-vector extension. In Section \ref{sec:independentGaussian} we consider phase two of our extension, where we limit our observations to an independent Gaussian input distribution. Section \ref{sec:generalGaussian} considers phase three, where we compare the arbitrary input to any general Gaussian input distribution.


\emph{Notation:} Straight boldface denotes multivariate quantities such as vectors (lowercase) and matrices (uppercase). Uppercase italics denotes random variables (boldface if we consider random vectors rather then random variables), and their realizations are represented by lowercase italics. The set of $\dim$-dimensional positive semidefinite matrices is denoted by $\PSD{\dim}$. The elements of a matrix $\mat{A}$ are represented by $[ \mat{A} ]_{ij}$. The operator $\diag{ \mat{A}}$ represents a column vector with the diagonal entries of matrix $\mat{A}$, and $\Diag{ \bf{a} }$ represents a diagonal matrix whose non-zero elements are given by the elements of vector $\bf{a}$. The superscript $( \cdot )^\T$ denotes the transpose. The operator $\Tr( \cdot )$ denotes the trace function, and $| \cdot |$ denotes the determinant function. The operator $\Jacob_{\bs{\scalart}} \mat{A}$ denotes the Jacobian matrix of $\mat{A}$ with respect to $\bs{\scalart}$ \cite{Magnus}.

Note that we also consider the conditioned version of the above defined quantities. That is, when the random vector $\rvec{x}$ depends on the random vector $\rvec{u}$, we require, for example, a conditioned version for the MMSE and the matrix $\Q$ given for a specific value of $\rvec{u} = \rvecr{u}$. In this case both quantities depend on an additional parameter $\rvecr{u}$, \ie, $\MSE{x|u}(\vectorT, \rvecr{u})$ and $\Q(\rvec{x}|\rvec{u} = \rvecr{u}, \Cov{x_G}, \vectorT)$ (the precise definitions given in Section \ref{ssec:independentConditioned}).

\section{The Scalar MIMO Channel} \label{sec:scalarMIMO}

As pointed out in the introduction, 
we begin our study with the simplest multivariate extension of the result in
\cite[Prp.~16]{PROP_full}, that is, we consider that the scalar random variables
involved in the model in (\ref{eq:scalarModel}) become random vectors. In other words, in
this section, we consider the following model:
\begin{gather} \label{eq:modelSimplest}
 \rvec{y} = \sqrt{\snr} \rvec{x} + \rvec{n}
\end{gather}
where the input random vector $\rvec{x} \in \R^{\dim}$ is arbitrarily
distributed and $\rvec{n} \in \R^{\dim}$ follows a standard Gaussian
distribution. Observe that (\ref{eq:modelSimplest}) is obtained by setting
$\Chan = \sqrt{\snr}\mat{I}_\dim$ in the vector model in
(\ref{eq:generalModel}).

Moreover, we further limit our discussion in this section to the comparison with a Gaussian input with \iid elements,
\ie, we assume that $\Cov{x_G} = \sigma^2 \mat{I}_\dim$.

Thus, for the settings in this section, the general MMSE matrix difference
function in (\ref{eq:vectorQ}) simplifies to
\begin{gather} \label{eq:Fsnr}
\Q(\rvec{x}, \sigma^2 \mat{I}_\dim, \scalart) =  \frac{\sigma^2}{1 + \sigma^2
\scalart} \mat{I}_\dim - \MSE{x}(\scalart)
\end{gather}
where $\scalart$ plays the role of the estimation $\snr$.

\subsection{A Single Crossing Point} \label{ssec:singleCrossingSimplest}

Motivated by the ``single crossing property'' of $\f(\resc{x},
\scalart)$ presented in \cite[Prp.~16]{PROP_full}, an immediate question that comes to mind is ``does this property extend to the MIMO scenario?'' Our hypothesis was that for the setting in (\ref{eq:modelSimplest}) this property will have a simple extension.
Thus, we
examine the simplest scalar function of the MMSE matrix difference function of
(\ref{eq:Fsnr}), that is, we consider some linear combination of it.
Accordingly, we define
\begin{IEEEeqnarray}{rCl}
\q_{\mat{A}}(\rvec{x}, \sigma^2, \scalart) & = & \Tr \left( \mat{A}
\Q(\rvec{x}, \sigma^2 \mat{I}_\dim, \scalart) \right) \\ & = & \frac{\sigma^2}{1
+ \sigma^2 \scalart} \Tr \left( \mat{A} \right) - \Tr \left( \mat{A}
\MSE{x}(\scalart) \right) \label{eq:defqA}
\end{IEEEeqnarray}
where $\mat{A}$ is a weighting matrix. 

The ``single crossing point'' property of $\ft$ extends naturally to the function $\q_{\mat{A}}(\rvec{x}, \sigma^2, \scalart)$, for a specific subset of matrices $\mat{A}$. This result is given in the next theorem.
\begin{thm} \label{thm:ScalarUniqueCrossingPoint}
Let $\mat{A} \in \PSD{\dim}$ be a positive semidefinite matrix. Then, the
function $\scalart \mapsto \q_{\mat{A}}(\rvec{x}, \sigma^2, \scalart)$, defined
in (\ref{eq:defqA}), has no nonnegative-to-negative zero crossings and, at most,
a single negative-to-nonnegative zero crossing in the range $\scalart \in [0,
\infty)$.

Moreover, assume $\snr_0 \in [0,  \infty)$ is a
negative-to-nonnegative crossing point. Then,
\begin{enumerate}
\item $\q_{\mat{A}}(\rvec{x}, \sigma^2, 0) \leq 0$.
\item $\q_{\mat{A}}(\rvec{x}, \sigma^2, \scalart)$ is a strictly increasing
function in the range $\scalart \in [0, \snr_0)$.
\item $\q_{\mat{A}}(\rvec{x}, \sigma^2, \scalart) \geq 0$ for all $\scalart \in
[\snr_0, \infty)$.
\item $\lim_{\scalart \to \infty} \q_{\mat{A}}(\rvec{x}, \sigma^2, \scalart) =
0$.
\end{enumerate}
\end{thm}
\begin{IEEEproof}
We start with the following three lemmas that are instrumental for this proof.
\begin{lem} \label{lem:AequivI}
 Let $\mat{A} \in \PSD{\dim}$ be a positive semidefinite matrix and let the
random vector $\rvec{x} \in \R^{\dim}$ be arbitrarily distributed. Then, we can
always find a random vector $\hat{\rvec{x}} \in \R^{\dim}$ such that the number
of nonnegative-to-negative and negative-to-nonnegative zero crossings of
$\q_{\mat{A}}(\rvec{x}, \sigma^2, \scalart)$ is the same as those of
$\q_{\mat{I}_\dim}(\hat{\rvec{x}}, \sigma^2, \scalart)$.
\end{lem}
\begin{IEEEproof}
See Appendix \ref{app:AequivI}.
\end{IEEEproof}

\begin{lem} \label{lem:GaussianOnTop}
Let $\rvec{x} \in \R^{\dim}$ be a random vector such
that $\Tr(\Cov{x})/\dim \leq \sigma^2$. Then, for every $\scalart \geq 0$, we
have
\begin{eqnarray} \label{eq:GaussianOnTop}
\frac{\Tr(\MSE{x}(\scalart))}{\dim} \leq \frac{\sigma^2}{1 +
\sigma^2 \scalart}
\end{eqnarray}
with equality if and only if $\rvec{X}$ is a Gaussian vector with \iid elements of
variance $\sigma^2$.
\end{lem}
\begin{IEEEproof}
See Appendix \ref{app:GaussianOnTop}.
\end{IEEEproof}

\begin{lem} \label{lem:dqAdgamma}
 Let $\mat{A} \in \R^{\dim\times \dim}$ be a square matrix. The derivative of
the function $\q_{\mat{A}}(\rvec{x}, \sigma^2, \scalart)$ with respect to
$\scalart$ is given by
\begin{gather} \label{eq:dqAdgamma}
\Jacob_{\scalart} \q_{\mat{A}}(\rvec{x}, \sigma^2, \scalart) = \Tr \left(
\mat{A} \Esp{\CMSE{x}{y}^2} \right)- \frac{\sigma^4}{(1 + \sigma^2
\scalart)^2} \Tr \left( \mat{A}
\right).
\end{gather}
\end{lem}
\begin{IEEEproof}
See Appendix \ref{app:dqAdgamma}.
\end{IEEEproof}

With these three lemmas at hand, we are now ready to continue with the proof of Theorem \ref{thm:ScalarUniqueCrossingPoint}.

Since we are assuming that the matrix $\mat{A}$ is positive semidefinite and
the distribution of $\rvec{x}$ is arbitrary, from Lemma \ref{lem:AequivI}, we
see that we can restrict our study of $\q_{\mat{A}}(\rvec{x}, \sigma^2,
\scalart)$ to that of $\q_{\mat{I}_\dim}(\hat{\rvec{x}}, \sigma^2, \scalart)$.
For the sake of simplicity, throughout this proof we will use
$\q(\sigma^2, \scalart) = \q_{\mat{I}_\dim} (\hat{\rvec{x}}, \sigma^2,
\scalart)$.

Now, according to Lemma \ref{lem:GaussianOnTop}, for the case where
$\Tr(\Cov{x})/\dim < \sigma^2$, the function $\q(\sigma^2, \scalart)$ has no
zeros and the statement in Theorem \ref{thm:ScalarUniqueCrossingPoint} is true.
In addition, if $\rvec{x}$ is Gaussian distributed with covariance matrix equal to
$\sigma^2 \mat{I}_\dim$, then $\q(\sigma^2, \scalart) = 0$, $\forall \scalart$,
which also fulfills Theorem \ref{thm:ScalarUniqueCrossingPoint}.

Thus, from this point, we can assume that $\Tr(\Cov{x})/\dim \geq
\sigma^2$ and that $\rvec{x}$ is not a Gaussian vector with covariance matrix $\sigma^2
\mat{I}_\dim$. Now, for $\scalart = 0$ we have $\q(\sigma^2, 0) = \sigma^2 -
\Tr(\Cov{x})/\dim \leq 0$ as required.

From the smoothness of $\q(\sigma^2, \scalart)$ as a function of $\gamma$, and done in \cite[Prp.~16]{PROP_full}, in order to
prove that no nonnegative-to-negative and at most one negative-to-nonnegative
zero crossings of $\q(\sigma^2, \scalart)$ can occur, we only need to show that
the derivative of $\q(\sigma^2, \scalart)$ is positive for all values of
$\scalart$ for which $\q(\sigma^2, \scalart) < 0$. Observe that
$\q(\sigma^2, \scalart) < 0$ implies that
\begin{gather}
\dim \frac{\sigma^2}{1 + \sigma^2\scalart} < \Tr ( \MSE{x}(\scalart) ) =
\Tr \left( \Esp{ \CMSE{x}{y} } \right). \label{eq:cond} 
\end{gather}

Now, particularizing Lemma \ref{lem:dqAdgamma} for $\mat{A} = \mat{I}_\dim$, we
have that
\begin{eqnarray}
\Jacob_\scalart \q(\sigma^2, \scalart)
& = & \Tr \left( \Esp{ \CMSE{x}{y}^2 } \right)- \dim \frac{\sigma^4}{(1 +
\sigma^2 \scalart)^2} \label{eq:derivative_f_nonnegative_1} \\
& > & \Tr \left( \Esp{ \CMSE{x}{y}^2 } \right) - \frac{\left(\Tr
\left( \Esp{ \CMSE{x}{y} } \right) \right)^2}{\dim}
\label{eq:derivative_f_nonnegative_2} \\
& = &  \ones^\T \Esp{ \CMSE{x}{y} \circ \CMSE{x}{y} } \ones -
\ones^\T \frac{\Esp[1]{ \diag{\CMSE{x}{y}} } \Esp[1]{
\diag{\CMSE{x}{y}}^\T}}{\dim}
\ones \label{eq:derivative_f_nonnegative_3} \\
& \geq & \Esp{ \ones^\T \left( \CMSE{x}{y} \circ \CMSE{x}{y}
- \frac{ \diag{\CMSE{x}{y}} \diag{\CMSE{x}{y}}^\T}{\dim} \right) \ones}
\label{eq:derivative_f_nonnegative_3_1} \\
& \geq & 0 \label{eq:derivative_f_nonnegative_4}
\end{eqnarray}
where (\ref{eq:derivative_f_nonnegative_2}) follows directly
from (\ref{eq:cond}); in (\ref{eq:derivative_f_nonnegative_3}) we have defined
$\mat{1}$ as the column vector whose entries are all ones, and we used $\circ$
to denote the Schur product; and (\ref{eq:derivative_f_nonnegative_3_1}) follows
from Jensen's inequality; 
finally, (\ref{eq:derivative_f_nonnegative_4}) follows from
\cite[Prp.~H.9]{Palomar2}.

Observe that the inequality in (\ref{eq:derivative_f_nonnegative_4}), which
holds for values of $\scalart$ such that $\q(\sigma^2, \scalart) < 0$, also
proves the second item in Theorem \ref{thm:ScalarUniqueCrossingPoint} and the
third one follows directly from the inexistence of nonnegative-to-negative zero
crossings. Furthermore, regarding the fourth item, it is clear that
$\lim_{\scalart \to \infty} \q(\sigma^2, \scalart) = 0$, as both terms in
$\q(\sigma^2, \scalart)$ tend to zero.
\end{IEEEproof}

\begin{rem} \label{rem:averageMMSE}
Note that the above theorem also holds for the normalized function, $\frac{1}{n}\q(\sigma^2, \scalart)$. Specifically, for the case of $\mat{A} = \mat{I}_\dim$, this is simply the difference between the MMSE of a general Gaussian random variable, with variance $\sigma^2$, and the \emph{average} MMSE of the $n$ elements of the random vector $\rvec{x}$.
\end{rem}

\begin{rem} \label{rem:negativeA}
For negative semidefinite $\mat{A}$ it can easily be seen from the proof of Lemma \ref{lem:AequivI} that $\q_{\mat{A}}(\rvec{x}, \sigma^2, \scalart)$ has the inverse properties, since it is a mirroring of some $\q_{\mat{I}_\dim}(\hat{\rvec{x}}, \sigma^2, \scalart)$ over the x-axis. This is to say, that it has at most a single positive-to-nonpositive zero crossing and, if such crossing exists, $\q_{\mat{A}}(\rvec{x}, \sigma^2, \scalart)$ will be nonnegative at $\gamma = 0$, strictly decreasing up to the crossing, nonpositive after the crossing, and will tend to zero as $\gamma \to \infty$.
\end{rem}

\begin{rem} \label{rem:indefiniteA}
For indefinite $\mat{A}$, ``single crossing point'' properties, such as those shown in Theorem \ref{thm:ScalarUniqueCrossingPoint}, do not hold in general.
\end{rem}

\subsection{Application: A Proof of a Special Case of Shannon's Vector EPI}
\label{ssec:applicationVectorEPI}

We now show that Theorem \ref{thm:ScalarUniqueCrossingPoint} can be used to
prove a special case of Shannon's EPI \cite[Th.~17.7.3]{CoverThomas}, similarly as it was done
in \cite{PROP_full} for the scalar case. Precisely, we will show that
\begin{eqnarray} \label{eq:specialEPI}
\exp\left(\frac{2}{\dim} \Ent\left( \rvec{x} + \rvec{n} \right) \right) \geq
\exp\left(\frac{2}{\dim} \Ent\left( \rvec{x} \right) \right) + 2 \pi \nume \vert
\Cov{n} \vert ^{\frac{1}{\dim}}
\end{eqnarray}
for any independent $\dim$-dimensional vectors $\rvec{x}$ and $\rvec{n}$ as long
as the differential entropy of $\rvec{x}$ is well-defined and $\rvec{n}$ is
Gaussian distributed with a positive definite covariance matrix
$\Cov{n}$.

We define $\rvec{z}$ to be an $\dim$-dimensional Gaussian vector with
covariance $\Cov{z} = \Cov{n}$ and independent of both $\rvec{x}$ and
$\rvec{n}$. Thus, without making any assumptions on the covariance
matrix of $\rvec{x}$, we can find an $\alpha \in [0, \infty)$ such that
the following equality holds:
\begin{eqnarray} \label{eq:alpha}
\Ent\left( \rvec{x} \right) = \Ent\left( \alpha \rvec{z} \right) =
\frac{1}{2} \log \left( (2 \pi \nume)^\dim \alpha^{2\dim} \vert
\Cov{n} \vert \right).
\end{eqnarray}

Since $\Cov{n}$ is positive definite there exists an
invertible matrix $\mat{V}$ such that $\Cov{n} = \mat{V}\mat{V}^\T$. Defining
$\rvec{\tilde{x}} = \mat{V}^{-1} \rvec{x}$, $\rvec{\tilde{z}} = \mat{V}^{-1} \rvec{z}$ and $\rvec{\tilde{n}} = \mat{V}^{-1} \rvec{n}$ we have the following chain of equalities:
\begin{align}
\Delta \I(\snr) &=
\I
\left( \alpha \rvec{z}; \sqrt{\snr} \alpha \rvec{z} + \rvec{n} \right) - \I
\left( \rvec{x}; \sqrt{\snr}\rvec{x} + \rvec{n} \right) \\
&= \I \left( \alpha \tilde{\rvec{z}}; \sqrt{\snr} \alpha \tilde{\rvec{z}} +
\tilde{\rvec{n}} \right) - \I \left( \tilde{\rvec{x}};
\sqrt{\snr}\tilde{\rvec{x}} + \tilde{\rvec{n}} \right) \\
&= \Ent \left( \sqrt{\snr} \alpha \tilde{\rvec{z}} + \tilde{\rvec{n}}
\right) - \Ent \left( \sqrt{\snr}\tilde{\rvec{x}} + \tilde{\rvec{n}} \right)
\label{eq:snr2infty} \\
&= \frac{1}{2} \int_{0}^{\snr} ( \mmse( \alpha \tilde{ \rvec{z}}, \scalart) -
\mmse( \tilde{\rvec{x}}, \scalart) ) \d \scalart
\label{eq:transfromationEquality} \\
&= \frac{1}{2} \int_{0}^{\snr} \Tr( \EM_{\alpha \tilde{\rvec{z}}}(\scalart) -
\EM_{\tilde{\rvec{x}}}(\scalart) ) \d \scalart \label{eq:EPIScalartoVector} \\
&= \frac{1}{2} \int_{0}^{\snr} \q_{\mat{I}_\dim}(\tilde{\rvec{x}}, \alpha^2, \scalart) \d \scalart \label{eq:DeltaIintegral}
\end{align}
where we have used the $\mmse$ function defined in
(\ref{eq:scalarMMSE_vector_arg}) and
the integral expression for the entropy function in \cite{IMMSE}.


Now, from (\ref{eq:snr2infty}) together with (\ref{eq:alpha}), it follows that
\begin{gather} \label{eq:limDeltaI}
\lim_{\snr \rightarrow \infty} \Delta \I(\snr) = 0
\end{gather}
which, from the integral expression in (\ref{eq:DeltaIintegral}), further
implies that the (smooth) integrand must have, at least, one zero crossing.
However, from Theorem \ref{thm:ScalarUniqueCrossingPoint}, we know that
$\q_{\mat{I}_\dim}(\tilde{\rvec{x}}, \alpha^2, \scalart)$ can
have, at most, one zero crossing. Consequently, in this case,
$\q_{\mat{I}_\dim}(\tilde{\rvec{x}}, \alpha^2, \scalart)$ must
have exactly one zero crossing. Also, from Theorem
\ref{thm:ScalarUniqueCrossingPoint} and (\ref{eq:limDeltaI}), we can infer that
there exists some $\snr_0 \in (0,\infty)$ such that
$\q_{\mat{I}_\dim}(\tilde{\rvec{x}}, \alpha^2, \scalart) < 0$,
$\forall \scalart \in [0, \snr_0)$, and $\q_{\mat{I}_\dim}(\tilde{\rvec{x}},
\alpha^2, \scalart) \geq 0$, $\forall \scalart \in [\snr_0,
\infty)$. Thus, it immediately follows that for finite $\snr$, $\Delta\I(\snr) \leq 0$ and
\begin{align}
\Delta \I (\snr)
&= \I \left( \alpha \rvec{z}; \sqrt{\snr} \alpha \rvec{z} + \rvec{n}
\right) - \I\left( \rvec{x}; \sqrt{\snr}\rvec{x} + \rvec{n} \right)
\\
&= \Ent \left( \sqrt{\snr} \alpha \rvec{z} + \rvec{n} \right) - \Ent
\left(\sqrt{\snr}\rvec{x} + \rvec{n} \right) \leq 0.
\label{eq:EPI_largZero}
\end{align}
It is now straightforward to see that
\begin{align}
\exp \left( \frac{2}{\dim} \Ent \left(\sqrt{\snr}\rvec{x} +
\rvec{n} \right) \right)
&\geq \exp\left( \frac{2}{\dim} \Ent \left(
\sqrt{\snr} \alpha \rvec{z} + \rvec{n} \right) \right) \\
&= \exp \left( \frac{1}{\dim} \log \left( (2 \pi \nume)^\dim ( \snr \alpha^2
+ 1)^\dim \vert \Cov{n} \vert \right) \right) \\
&= (2 \pi \nume) ( \snr \alpha^2 + 1) \vert \Cov{n}
\vert^{\frac{1}{\dim}} \\
&= \exp \left( \frac{2}{\dim} \Ent \left( \sqrt{\snr} \alpha \rvec{z}
\right) \right) + (2 \pi \nume) \vert \Cov{n} \vert^{\frac{1}{\dim}} \\
&= \exp \left( \frac{2}{\dim} \Ent \left( \sqrt{\snr} \rvec{x} \right) \right) +
(2 \pi \nume) \vert \Cov{n} \vert^{\frac{1}{\dim}} \label{eq:EPI_final}
\end{align}
which is exactly (\ref{eq:specialEPI}) up to scaling in $\sqrt{\snr}$, which we
can always take equal to $1$. We note here that the I-MMSE relationship was
used in \cite{EPI} to prove Shannon's EPI, Costa's EPI and also the generalized EPI for linear
transformations of a random vector.

\section{From Scalar to Vector Channels: Definitions and Preliminaries}
\label{sec:scalartovector}

In the previous section, we discussed the simple model presented in (\ref{eq:modelSimplest}). We have shown that the ``single crossing point'' property initially proved for the scalar channel in \cite{PROP_full} extends very smoothly and intuitively on to this model. The reason for the smooth transition is that, even though we are considering a multivariate scenario, all elements of the input vector undergo the same effect in the channel. They are all amplified by $\snr$ and distorted by additive standard Gaussian noise. From a more technical viewpoint, when one wants to search for a ``single crossing point'' property, one must define some scalar function of some scalar parameter, for which the property holds. In the model of (\ref{eq:modelSimplest}) the intuitive choice is simply to take the trace of the MMSE as a function of $\snr$. And indeed, this is just one possible linear combination included in Theorem \ref{thm:ScalarUniqueCrossingPoint}, for which we have shown that the property can be extended.

Taking the next step, from this initial extension to the general model of (\ref{eq:generalModel}), is a harder task.
Moreover, there is no single method of doing so. In fact there are two degrees of freedom in this transition. First of all there is a need for some scalar parameter that will define $\Chan$. This parameter will be equivalent to the $\snr$ parameter in the scalar case or the simple model of (\ref{eq:modelSimplest}). Secondly, there is a need for some scalar function of the matrix $\Q$. In the simple model of (\ref{eq:modelSimplest}) we defined the function $\q_{\mat{A}}(\rvec{x}, \sigma^2, \scalart)$ which was simply taking some linear (positive semidefinite) combination of the elements of the matrix. The trace function is one example of such a combination, which is also the most intuitive extension; however, in the general model (or even the parallel model, which we will discuss shortly) the ``single crossing point'' property does not hold, in general, for the trace function. Thus, our goal is to find a ``single crossing point'' property that will be both elegant and, more importantly, useful and applicable.

As such, in this work we narrowed our investigation to the subset of parallel channels or diagonal matrices $\Chan$, for which we have the following result.

\begin{lem} \label{lem:path}
For any two diagonal channel matrices, $\Chan_1$ and $\Chan_2$, such that $\mat{0} \preceq \Chan_1 \preceq \Chan_2$, there exists a path $\Chan( \vectort )$ such that the following holds:
\begin{itemize}

\item For all $\vectort$, $\Chan( \vectort) \succeq \mat{0}$ and is a diagonal matrix.

\item For all $\vectort$, $\Jacob_{\vectort} \Chan( \vectort) \succeq \mat{0}$ and is a diagonal matrix.

\item $\Chan( 0 ) = \mat{0}$.

\item $\Chan( \vectort_1 ) = \Chan_1$ and $\Chan( \vectort_2 ) = \Chan_2$ where $0 \leq \vectort_1 \leq \vectort_2$.

\item The diagonal elements of $\Chan( \vectort)$ go to $\infty$ in a linear rate.
\end{itemize}
\end{lem}
\begin{IEEEproof}
We need to define a function, $g_i(\vectort)$, for each diagonal element of the matrix $\Chan( \vectort )$. It suffices to choose any non-negative function $h_i(\vectort)$ such that the area from $0$ to $\vectort_1$ will equal $[\Chan_1]_{ii}$ and the area from $\vectort_1$ to $\vectort_2$ will equal $[\Chan_2]_{ii} - [\Chan_1]_{ii}$. Given that, we can set the function to be $g_i(\vectort) = \int_0^\vectort h_i(\tau) \d \tau$. The entire path, $\Chan(\vectort)$, will be given by:
\begin{eqnarray} \label{eq:path}
\Chan(\vectort) = \Diag{ \{g_i(\vectort)\} } .
\end{eqnarray}
As required, this path passes between the zero matrix at $\vectort=0$,
$\Chan_1$ at $\vectort_1$ and $\Chan_2$ at $\vectort_2$. Since $h_i(\vectort)$ are chosen nonnegative for all $i$ we
have a nonnegative and monotonically nondecreasing path for all $\vectort$. The above construction guarantees that both $\Chan(\vectort)$ and $\Jacob_{\vectort} \Chan(\vectort)$ will be diagonal matrices for all $\vectort$. Moreover, we may also assume that the functions $h_i(\vectort)$ plateau after complying with all other requirements, that is, from $\vectort_2$ onwards. This assures that $g_i(\vectort)$ goes to $\infty$ in a linear rate.
\end{IEEEproof}
Note that the above lemma can be extended to $M$ matrices $\Chan_j \preceq \Chan_{j+1}$ for $j=1,...,M-1$, using a similar construction.

Under the above detailed limitation, of restricting ourselves to parallel channels, we examine two different cases: phases two and three of our extension. In phase two, detailed in Section \ref{sec:independentGaussian}, we assume that the Gaussian covariance matrix defining the matrix $\Q$ in (\ref{eq:Fsnr}) is that of a Gaussian distribution with independent elements, that is $\Cov{x_G} = \mat{\Lambda}_{\rvec{x_G}}$ is a diagonal matrix. In this case we will see that the ``single crossing point'' property occurs for each and every diagonal element of $\Q$. This is not a straightforward extension of the scalar property, since the elements of the random input vector $\rvec{x}$ are, in general, not independent. In Section \ref{sec:generalGaussian} we proceed to phase three where we allow any Gaussian distribution in the definition of $\Q$. In this phase we will see that the ``single crossing point'' property occurs for each and every eigenvalue of the matrix $\Q$. Surely, this is not a straightforward extension of any of the previous results. Moreover, the results of Section \ref{sec:independentGaussian} cannot be trivially deduced from the results of phase three, since restricting only the Gaussian covariance to be diagonal does not guarantee that the eigenvalues of $\Q$ will be on its diagonal. Thus, we have two distinctive results. All results (including those of the previous section), fall back to the scalar ``single crossing point'' property result \cite{PROP,PROP_full} when both the arbitrary input vector $\rvec{x}$ and the Gaussian input random vector are restricted to have independent elements.

Before proceeding to examine these two cases we require a preliminary result.
The basis for the applicability of the ``single crossing point'' property in the scalar case and in the simple model of (\ref{eq:generalModel}) is the I-MMSE relationship \cite{IMMSE}. This is still the case in the extensions we are considering next, however, we require also an extension of the I-MMSE result which was derived by Palomar and Verd$\acute{\textrm{u}}$ in \cite{Palomar}:
\begin{gather} \label{eq:palomarVerdu}
\nabla_{\Chan} \Igen{\rvec{x}}{\Chan \rvec{x}+\rvec{n}} = \Chan \EM .
\end{gather}
This relationship was derived for complex-valued variables, however it holds verbatim for real-valued variables. Assuming the channel coefficients can be written as a function of a single parameter, $\vectort$, we can rewrite the above relationship as an integral over this parameter, which results with the following expression:
\begin{align} \label{eq:lineIntegral}
\Igen{\rvec{x}}{\rvec{y}(\vectort)} &=  \Igen{\rvec{x}}{\Chan(\vectort) \rvec{x} + \rvec{n}} \nonumber \\
 &= \int_{\tau=0}^{\vectort} \ones^\T \left( \Chan(\tau) \MSE{\rvec{x}}(\tau) \circ \Jacob_{\tau}\Chan(\tau) \right) \ones \d \tau \nonumber \\
 &=  \int_{\tau=0}^{\vectort} \Tr \left( \left(
\Chan(\tau) \MSE{\rvec{x}}(\tau) \right)^\T  \Jacob_{\tau}\Chan(\tau)
\right) \d \tau \\ \label{eq:lineIntegral2}
&=  \int_{\tau=0}^{\vectort} \Tr \left(
\ChanDChan(\tau) \MSE{\rvec{x}}(\tau)
\right) \d \tau
\end{align}
where we have used the following definition:
\begin{gather} \label{eq:definitionB}
\ChanDChan(\vectort) \equiv \Chan(\vectort) \left( \Jacob_{\vectort}\Chan(\vectort) \right)^\T.
\end{gather}
This also carries over to the conditioned case as follows:
\begin{align} \label{eq:lineIntegral_cond}
\Icond{\rvec{x}}{\rvec{y}(\vectort)}{\rvec{u}} &=  \Icond{\rvec{x}}{\Chan(\vectort) \rvec{x} + \rvec{n}}{\rvec{u}} \nonumber \\
 &=  \int_{\tau=0}^{\vectort} \Tr \left( \left(
\Chan(\tau) \MSE{\rvec{x}|\rvec{u}}(\tau) \right)^\T  \Jacob_{\tau}\Chan(\tau)
\right) \d \tau \\ \label{eq:lineIntegral_cond2}
&=  \int_{\tau=0}^{\vectort} \Tr \left(
\ChanDChan(\tau) \MSE{\rvec{x}|\rvec{u}}(\tau)
\right) \d \tau.
\end{align}


\section{Vector Channel: Comparing with an Independent Gaussian Distribution}
\label{sec:independentGaussian}

We begin our analysis of the extended model (\ref{eq:generalModel}), limited to parallel channel matrices, by assuming that the Gaussian covariance matrix, defining the matrix $\Q$, is that of an independent distribution, that is, $\Cov{x_G} = \mat{\Lambda}_{\rvec{x_G}}$, throughout this section. Recall, nonetheless, that $\rvec{x}$ remains completely arbitrary. More precisely, we consider the following matrix:
\begin{align}
\Q(\rvec{x}, \mat{\Lambda}_{\rvec{x_G}}, \vectort)
  &= \EM_{G}(\mat{\Lambda}_{\rvec{x_G}}, \vectort) - \MSE{x}(\vectort)
\label{eq:defQind} \\
 &= \Diag{\left\{ \frac{[\mat{\Lambda}_{\rvec{x_G}}]_{ii}}{1 +
[\Chan(\vectort)]_{ii}^2 [\mat{\Lambda}_{\rvec{x_G}}]_{ii}} \right\}} -
\MSE{x}(\vectort). \label{eq:Qindexpr}
\end{align}
Under these assumptions we will see, in Section \ref{ssec:singleCrossingPointInd}, that a ``single crossing point'' property occurs for each and every diagonal element of the matrix $\Q$. After extending this result to the conditioned case, in Section \ref{ssec:independentConditioned}, we will use the I-MMSE relationship, in Section \ref{ssec:independentMIproperties}, to show the effect of this property on information-theoretic quantities, and more specifically on the mutual information. Finally, in Section \ref{ssec:applicationInd}, we will put these results to use on a variant of the \emph{degraded} BC, in order to show their applicability to information theory problems.

\subsection{A Single Crossing Point Property on the Diagonal Elements of $\Q$}
\label{ssec:singleCrossingPointInd}

As pointed out above, our main result, in this section, is an extension of the ``single crossing point'' property. Precisely, we show that the property extends on each and every diagonal element of the matrix $\Q$. This result is given in the next theorem.

\begin{thm} \label{thm:DiagonalsUniqueCrossingPoint}
The diagonal entries of the matrix-valued function $\vectort \mapsto
\Q(\rvec{x}, \mat{\Lambda}_{\rvec{x_G}}, \vectort)$, defined
in (\ref{eq:defQind}), have no nonnegative-to-negative zero crossings and, at
most, a single negative-to-nonnegative zero crossing in the range $\vectort \in
[0, \infty)$. Moreover, let $\vectort_0 \in [0,  \infty)$ be the
negative-to-nonnegative crossing point for $[\Q(\rvec{x}, \mat{\Lambda}_{\rvec{x_G}}, \vectort)]_{ii}$. Then,
\begin{enumerate}

\item $[\Q(\rvec{x}, \mat{\Lambda}_{\rvec{x_G}}, 0)]_{ii} \leq 0$.
\item $[\Q(\rvec{x}, \mat{\Lambda}_{\rvec{x_G}}, \vectort)]_{ii}$ is a strictly
increasing function in the range $\vectort \in [0, \vectort_0)$.
\item $[\Q(\rvec{x}, \mat{\Lambda}_{\rvec{x_G}}, \vectort)]_{ii} \geq 0$ for
all $\vectort \in [\snr_0, \infty)$.
\item Assuming $\lim_{\vectort \rightarrow \infty} [\Chan(\vectort)]_{ii} = \infty$, we have that $\lim_{\vectort \to \infty} [\Q(\rvec{x}, \mat{\Lambda}_{\rvec{x_G}},
\vectort)]_{ii} = 0$.

\item $[\Q(\rvec{x}, \mat{\Lambda}_{\rvec{x_G}}, \vectort)]_{ii}$ is a continuous and monotonically increasing function in $[\mat{\Lambda}_{\rvec{x_G}}]_{ii}$.

\end{enumerate}
\end{thm}
\begin{IEEEproof}
Before giving the actual proof, let us first present an intermediate result.
\begin{lem} \label{lem:indGaussianOnTop}
Let $\rvec{x} \in \R^{\dim}$ be a random vector such
that $[\Cov{x}]_{ii} \leq [\mat{\Lambda}_{\rvec{x_G}}]_{ii}$,
where $i\in[1, \dim]$. Then, for
every $\vectort \geq 0$, we have
\begin{eqnarray} \label{eq:indGaussianOnTop}
[\MSE{x}(\vectort)]_{ii} \leq [\EM_{G}(\mat{\Lambda}_{\rvec{x_G}},
\vectort)]_{ii} = \frac{[\mat{\Lambda}_{\rvec{x_G}}]_{ii}}{1 +
[\Chan(\vectort)]_{ii}^2 [\mat{\Lambda}_{\rvec{x_G}}]_{ii}}
\end{eqnarray}
with equality if and only if $[\rvec{X}]_i$ is Gaussian distributed,
independent of the other entries of $\rvec{x}$ and such that $[\Cov{x}]_{ii} =
[\mat{\Lambda}_{\rvec{x_G}}]_{ii}$.
\end{lem}
\begin{IEEEproof}
See Appendix \ref{app:indGaussianOnTop}.
\end{IEEEproof}

Now, according to Lemma \ref{lem:indGaussianOnTop}, for the case where
$[\Cov{x}]_{ii} < [\mat{\Lambda}_{\rvec{x_G}}]_{ii}$, the function
$[\Q(\rvec{x}, \mat{\Lambda}_{\rvec{x_G}}, \vectort)]_{ii}$ has no zeros and
the statement in Theorem \ref{thm:ScalarUniqueCrossingPoint} is true. In
addition, if $[\rvec{x}]_i$ is Gaussian distributed (and independent of the other
entries of the vector $\rvec{x}$) with variance equal to $[\Cov{x}]_{ii} =
[\mat{\Lambda}_{\rvec{x_G}}]_{ii}$, then $[\Q(\rvec{x},
\mat{\Lambda}_{\rvec{x_G}}, \vectort)]_{ii} = 0$, $\forall \vectort$, which
also fulfills Theorem \ref{thm:ScalarUniqueCrossingPoint}.

Thus, from this point, we can assume that $[\Cov{x}]_{ii} \geq
[\mat{\Lambda}_{\rvec{x_G}}]_{ii}$ and that $[\rvec{x}]_i$ is not: Gaussian
distributed, independent of the other entries of $\rvec{x}$, and with
$[\Cov{x}]_{ii} = [\mat{\Lambda}_{\rvec{x_G}}]_{ii}$. Now, for
$\vectort = 0$ we have $[\Q(\rvec{x}, \mat{\Lambda}_{\rvec{x_G}},
0)]_{ii} = [\mat{\Lambda}_{\rvec{x_G}}]_{ii} - [\Cov{x}]_{ii} \leq 0$ as
required.

Similarly as it was done in the proof of Theorem
\ref{thm:ScalarUniqueCrossingPoint}, in order to prove that no
nonnegative-to-negative and at most one negative-to-nonnegative zero crossings
of $[\Q(\rvec{x}, \mat{\Lambda}_{\rvec{x_G}}, \vectort)]_{ii}$ can occur, we
only need to show that the derivative of
$[\Q(\rvec{x}, \mat{\Lambda}_{\rvec{x_G}}, \vectort)]_{ii}$ with respect to
$\vectort$ is positive for all values of $\vectort$ for which $[\Q(\rvec{x},
\mat{\Lambda}_{\rvec{x_G}}, \vectort)]_{ii} < 0$. Observe that $[\Q(\rvec{x},
\mat{\Lambda}_{\rvec{x_G}}, \vectort)]_{ii} < 0$ implies
\begin{gather}
\frac{[\mat{\Lambda}_{\rvec{x_G}}]_{ii}}{1 + [\Chan(\vectort)]_{ii}^2
[\mat{\Lambda}_{\rvec{x_G}}]_{ii}} = [\EM_{G}(\mat{\Lambda}_{\rvec{x_G}},
\vectort)]_{ii} <
[\MSE{x}(\vectort)]_{ii} .
\end{gather}


Now, from (\ref{eq:defQind}), it is clear that, in order to compute the
derivative of $[\Q(\rvec{x}, \mat{\Lambda}_{\rvec{x_G}}, \vectort)]_{ii}$, we
first need the derivative of $[\MSE{x}(\vectort)]_{ii}$:
\begin{align} \label{eq:proofThmSingleCrossingInd}
 \Jacob_\vectort [\MSE{x}(\vectort)]_{ii}
&= \Jacob_{\Chan(\vectort)} [\MSE{x}(\vectort)]_{ii} \Jacob_\vectort
\Chan(\vectort) \\
&= \sum_{j=1}^{\dim} \Jacob_{[\Chan(\vectort)]_{jj}} [\MSE{x}(\vectort)]_{ii}
\Jacob_\vectort [\Chan(\vectort)]_{jj}
\end{align}
where, in the last step, we have used the assumption that $\Chan(\vectort)$ is
a diagonal matrix for all $\vectort$. From \cite[Eq.~(131)]{Palomar2}, we have
\begin{align}
 \Jacob_{[\Chan(\vectort)]_{jj}} [\MSE{x}(\vectort)]_{ii}
&= -2 \Esp[1]{[\CMSE{x}{y}]_{ij}[\CMSE{x}{y} \Chan(\vectort)^\T]_{ij}} \\
&= -2 [\Chan(\vectort)]_{jj} \Esp[1]{[\CMSE{x}{y}]_{ij}^2}.
\end{align}

Recalling the definition $[\ChanDChan(\vectort)]_{ii} = [\Chan(\vectort)]_{ii}
\Jacob_\vectort [\Chan(\vectort)]_{ii}$ in (\ref{eq:definitionB}), we are now ready to compute
the derivative of $[\Q(\rvec{x}, \mat{\Lambda}_{\rvec{x_G}}, \vectort)]_{ii}$,
which reads as
\begin{IEEEeqnarray}{rCl} 
\IEEEeqnarraymulticol{3}{l}{
\Jacob_\vectort [\Q(\rvec{x}, \mat{\Lambda}_{\rvec{x_G}}, \vectort)]_{ii}} \\
\quad & = & 2 \sum_{j=1}^{\dim} [ \ChanDChan(\vectort) ]_{jj} \left(
\Esp[1]{[\CMSE{x}{y}]_{ij}^2} - [ \EM_{G}(\mat{\Lambda}_{\rvec{x_G}}, \vectort)
]_{ij}^2 \right) \label{eq:derivativeHi_1} \\
& = & 2 [\ChanDChan(\vectort)]_{ii} \left( \Esp[1]{[\CMSE{x}{y}]_{ii}^2} -
[\EM_{G}(\mat{\Lambda}_{\rvec{x_G}}, \vectort)]_{ii}^2
\right)
+ 2 \sum_{j \neq i} [ \ChanDChan(\vectort) ]_{jj}
\Esp[1]{[\CMSE{x}{y}]_{ij}^2}
\label{eq:derivativeHi_2} \\
& \geq & 2 [\ChanDChan(\vectort)]_{ii} \left( \Esp[1]{[\CMSE{x}{y}]_{ii}^2} -
[\EM_{G}(\mat{\Lambda}_{\rvec{x_G}}, \vectort)]_{ii}^2
\right) \label{eq:derivativeHi_3} \\
& > & 2 [\ChanDChan(\vectort)]_{ii} \left( \Esp[1]{[\CMSE{x}{y}]_{ii}^2} -
\big(\Esp[1]{[\CMSE{x}{y}]_{ii}} \big)^2
\right) \label{eq:derivativeHi_4}  \\
& \geq & 0 \label{eq:derivativeHi_5}
\end{IEEEeqnarray}
where (\ref{eq:derivativeHi_1}) follows from the fact that for
Gaussian input distributions (not necessarily \iid), the conditional MMSE
matrix $\CMSEr{\rvec{x}_G}{y}$ does not depend on the observation
$\rvecr{y}$, \ie, $\EM_G (\CovMat_{\rvec{x}_G}, \vectort) = \CM_{\rvec{x_G}}$.
Equation (\ref{eq:derivativeHi_2}) is due to the
fact that the entries of the Gaussian input distribution $\rvec{x_G}$ are
independent and, thus, its MMSE matrix is diagonal; (\ref{eq:derivativeHi_3}) is
due to the fact that $[\ChanDChan(\vectort)]_{ii} \geq 0$, as shown in Lemma \ref{lem:path};
(\ref{eq:derivativeHi_4}) follows from the assumption $[\Q(\rvec{x},
\mat{\Lambda}_{\rvec{x_G}}, \vectort)]_{ii} < 0$ and
(\ref{eq:derivativeHi_5}) can be derived from Jensen's inequality.

Observe that the inequality in (\ref{eq:derivativeHi_5}), which
holds for values of $\vectort$ such that $[\Q(\rvec{x},
\mat{\Lambda}_{\rvec{x_G}}, \vectort)]_{ii} < 0$, also
proves the second item in Theorem \ref{thm:DiagonalsUniqueCrossingPoint} and the
third one follows directly from the inexistence of nonnegative-to-negative zero
crossings. Regarding the fourth item, it is clear that
$\lim_{\vectort \to \infty} [\Q(\rvec{x}, \mat{\Lambda}_{\rvec{x_G}},
\vectort)]_{ii} = 0$, as both terms in the expression of $[\Q(\rvec{x},
\mat{\Lambda}_{\rvec{x_G}}, \vectort)]_{ii}$ in (\ref{eq:Qindexpr}) tend to
zero, when $\lim_{\vectort \rightarrow \infty} [\Chan(\vectort)]_{ii} = \infty$.
Finally, the last property is a direct consequence of the definition of the function $\Q(\rvec{x}, \mat{\Lambda}_{\rvec{x_G}},
\vectort)$ (\ref{eq:Qindexpr}).
\end{IEEEproof}

We now define the following function:
\begin{gather} \label{eq:BQ}
\BQ(\rvec{x},\mat{\Lambda}_{\rvec{x_G}},\vectort) =  [\ChanDChan(\vectort)]_{ii} [\Q(\rvec{x},\mat{\Lambda}_{\rvec{x_G}}, \vectort)]_{ii}
\end{gather}
and also,
\begin{gather} \label{eq:BQsum}
\BQsum(\rvec{x},\mat{\Lambda}_{\rvec{x_G}},\vectort) = \sum_{i=1}^{\dim} \BQ(\rvec{x},\mat{\Lambda}_{\rvec{x_G}},\vectort)
\end{gather}
For which we can give the following two corollaries,
\begin{cor} \label{cor:BQ}
Let $\rvec{x} \in \R^{\dim}$ be any random vector. The function $\BQ(\rvec{x},\mat{\Lambda}_{\rvec{x_G}},\vectort)$ has the following properties:
\begin{enumerate}

\item $\BQ(\rvec{x},\mat{\Lambda}_{\rvec{x_G}},0) = 0$.

\item It has at most a single negative-zero-positive crossing in the range $\vectort \in
(0, \infty)$.

\item When $\lim_{\vectort \rightarrow \infty} [\Chan(\vectort)]_{ii} = \infty$ we have that, $\lim_{t \to \infty} \BQ(\rvec{x},\mat{\Lambda}_{\rvec{x_G}},\vectort) = 0$.

\item If $[\mat{\Lambda}_{\rvec{x_G}}]_{ii} = [\Cov{\rvec{x}}]_{ii}$, then $\BQ(\rvec{x},\mat{\Lambda}_{\rvec{x_G}},\vectort) \geq 0$ for all $\vectort$. Furthermore, $\BQ(\rvec{x},\mat{\Lambda}_{\rvec{x_G}},\vectort)$ is a continuous and monotonically increasing function in $[\mat{\Lambda}_{\rvec{x_G}}]_{ii}$.
\end{enumerate}
\end{cor}
\begin{IEEEproof} The first three properties follow from Theorem
\ref{thm:DiagonalsUniqueCrossingPoint} and the fact that $[\ChanDChan(\vectort)]_{ii}$ is zero at $\vectort=0$, non-negative for all other values of $\vectort \in (0, \infty)$ and $[\ChanDChan(\vectort)]_{ii}$ goes to $\infty$ in a linear rate, as shown in Lemma \ref{lem:path}. The fourth property is a direct result of Lemma \ref{lem:indGaussianOnTop} and the fifth item of Theorem \ref{thm:DiagonalsUniqueCrossingPoint}.
\end{IEEEproof}
Figure \ref{f:f_i_example} illustrates this property, in which the negative-zero-positive crossing of $\BQ(\rvec{x},\mat{\Lambda}_{\rvec{x_G}},\vectort)$ is simply a negative-to-nonnegative zero crossing and, thus, agrees with the negative-to-nonegative zero crossing of $[\Q(\rvec{x}, \mat{\Lambda}_{\rvec{x_G}},\vectort)]_{ii}$.

\begin{figure}
\begin{center}
\setlength{\unitlength}{.1cm}
    \includegraphics[width=1\textwidth]{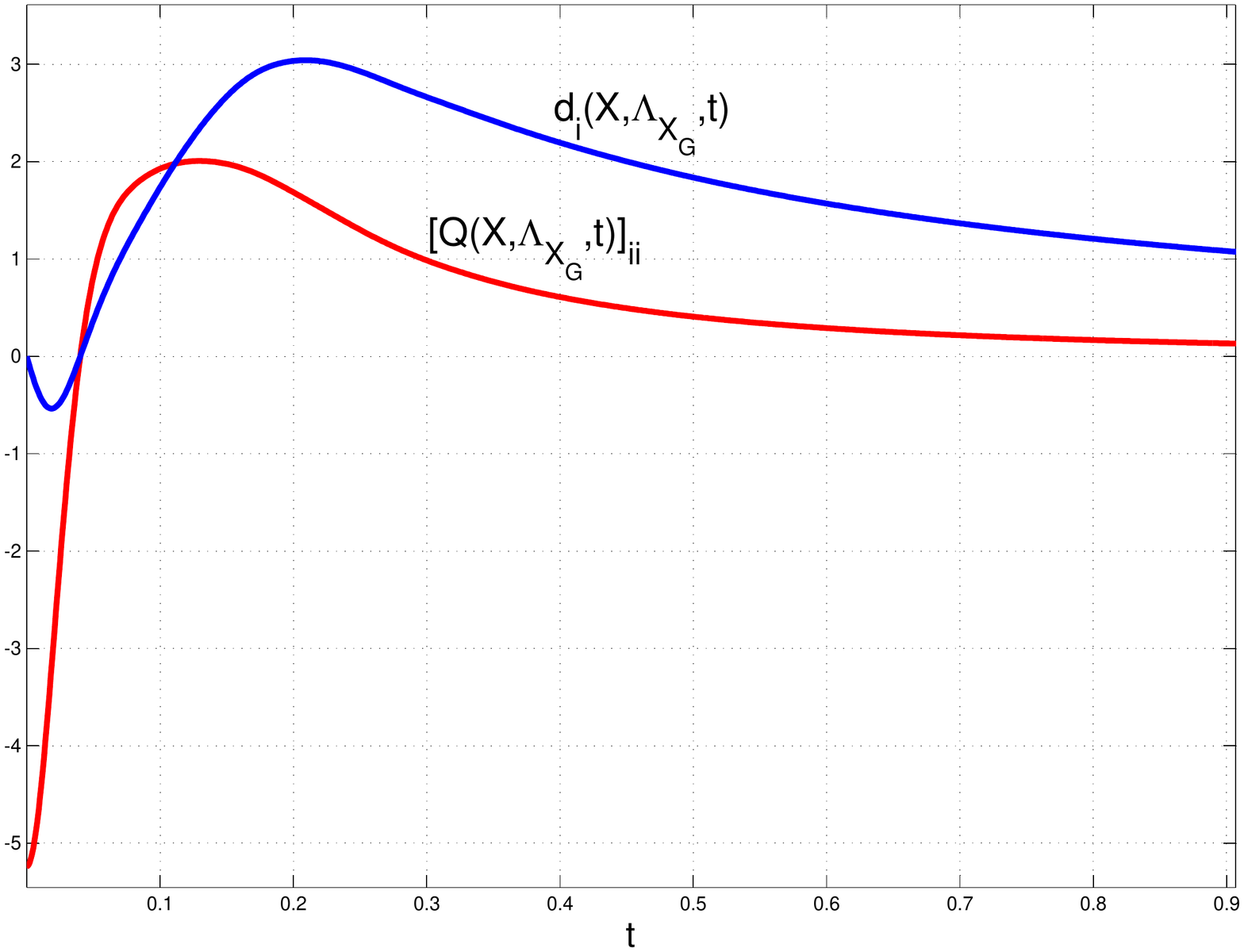}
    \caption{An example of the function
$[\Q(\rvec{x}, \mat{\Lambda}_{\rvec{x_G}},\vectort)]_{ii}$ (in red) and the matching
function $\BQ(\rvec{x},\mat{\Lambda}_{\rvec{x_G}},\vectort)$ (in blue). Both have the same single negative-to-nonnegative zero
crossing in the range $\vectort \in (0, \infty)$.} \label{f:f_i_example}
\end{center}
\end{figure}


\begin{cor} \label{cor:BQsum}
Let $\rvec{x} \in \R^{\dim}$ be any random vector. The function $\BQsum(\rvec{x},\mat{\Lambda}_{\rvec{x_G}},\vectort)$ is either negative for all $\vectort$, or there exists $\vectort' \in [0,\infty)$ such that for all $\vectort > \vectort'$ the function $\BQsum(\rvec{x},\mat{\Lambda}_{\rvec{x_G}},\vectort)$ is nonnegative. Moreover, when $\lim_{\vectort \rightarrow \infty} [\Chan(\vectort)]_{ii} = \infty$ we have that, $\lim_{\vectort \to \infty} \BQsum(\rvec{x},\mat{\Lambda}_{\rvec{x_G}},\vectort) = 0$, and if $[\mat{\Lambda}_{\rvec{x_G}}]_{ii} = [\Cov{\rvec{x}}]_{ii}$ for all $i$, then $\BQsum(\rvec{x},\mat{\Lambda}_{\rvec{x_G}},\vectort) \geq 0$ for all $\vectort$.
\end{cor}


\subsection{The Conditioned Case} \label{ssec:independentConditioned}

Before proceeding to understanding the implications of the above results on information-theoretic quantities, we would like to extend
these results to the conditioned case.

Let us begin with the conditioned MMSE matrix. We first consider the
following matrix quantity:
\begin{align} \label{eq:defineE_X_u}
\EM_{\rvec{x}|\rvec{u}}(\vectort, \rvecr{u})
&= \CEsp[1]{ (\rvec{x} - \CEsp{\rvec{x}}{\Chan(\vectort) \rvec{x} + \rvec{n},
\rvec{u} = \rvecr{u}}) (\rvec{x} - \CEsp{\rvec{x}}{\Chan(\vectort) \rvec{x} +
\rvec{n}, \rvec{u} = \rvecr{u}})^\T}{\rvec{u} = \rvecr{u}} \\
&= \Esp[1]{ (\rvec{x}_{\rvecr{u}} - \CEsp{\rvec{x}_{\rvecr{u}}}{\Chan(\vectort)
\rvec{x}_{\rvecr{u}} + \rvec{n}}) (\rvec{x}_{\rvecr{u}} -
\CEsp{\rvec{x}_{\rvecr{u}}}{\Chan(\vectort) \rvec{x}_{\rvecr{u}} +
\rvec{n}})^\T} \label{eq:defineE_X_u_2}
\end{align}
where $\rvec{x}_{\rvecr{u}} $ is a random vector distributed according to
$\pdffun_{\rvec{x} | \rvec{u} = \rvecr{u}}$. The conditioned MMSE matrix is
simply the expectation of (\ref{eq:defineE_X_u}) according to the distribution
of the
random vector $\rvec{u}$:
\begin{gather} \label{eq:EuFromE_X_u}
\EM_{\rvec{x}|\rvec{u}}(\vectort) = \Esp{\EM_{\rvec{x}|\rvec{u}}(\vectort,
\rvec{u})} = \Esp[1]{ (\rvec{x} - \CEsp{\rvec{x}}{\Chan(\vectort) \rvec{x} +
\rvec{n}, \rvec{u}}) (\rvec{x} - \CEsp{\rvec{x}}{\Chan(\vectort)
\rvec{x} + \rvec{n}, \rvec{u} })^\T} .
\end{gather}
Another important quantity that needs to be extended to the conditioned case is:
\begin{align} \label{eq:CMMSE_matrix_u}
\CMSEr{x_{\rvecr{u}}}{y} &= \Esp[1]{ (\rvec{x}_{\rvecr{u}} - \CEsp{\rvec{x}_{\rvecr{u}}}{\rvecr{y}})
(\rvec{x}_{\rvecr{u}} - \CEsp{\rvec{x}_{\rvecr{u}}}{\rvecr{y} })^\T | \rvecr{y}} \\
&= \Esp[1]{ (\rvec{x} - \CEsp{\rvec{x}}{\rvecr{y},\rvec{u} = \rvecr{u}})
(\rvec{x} - \CEsp{\rvec{x}}{\rvecr{y},\rvec{u} = \rvecr{u} })^\T | \rvecr{y},\rvec{u} = \rvecr{u}} \\
& = \CMSEr{x}{y,\rvec{u}=\rvecr{u}}
\end{align}
where, as in the unconditioned case, this function, in general, depends on both $\rvecr{u}$ and $\rvecr{y}$, thus, we have $\EM_{\rvec{x}|\rvec{u}}(\vectort, \rvecr{u}) = \Esp[1]{ \CMSEr{x}{\rvec{y},\rvec{u}=\rvecr{u}} }$, where the expectation is over $\rvec{y}$. However, when the input distribution of $\rvec{x}_{\rvecr{u}} $ is Gaussian $\CMSEr{x}{y,\rvec{u}=\rvecr{u}}$ is independent of $\rvecr{y}$.
In a similar manner, we have the following:
\begin{gather} \label{eq:assistHi_u}
\Q(\rvec{x}|\rvec{u} = \rvecr{u}, \mat{\Lambda}_{\rvec{x_G}}, \vectort) = \EM_{G}(\mat{\Lambda}_{\rvec{x_G}}, \vectort) -
\EM_{\rvec{x}|\rvec{u}}(\vectort, \rvecr{u})
\end{gather}
and, thus, we also have:
\begin{gather} \label{eq:QmatrixConditioned}
\Q(\rvec{x}|\rvec{u}, \mat{\Lambda}_{\rvec{x_G}}, \vectort) = \EspOp_{\rvec{u}}\left\{
\Q(\rvec{x}|\rvec{u} = \rvecr{u}, \mat{\Lambda}_{\rvec{x_G}}, \vectort) \right\} =
\EM_{G}(\mat{\Lambda}_{\rvec{x_G}}, \vectort) -
\EM_{\rvec{x}|\rvec{u}}(\vectort) .
\end{gather}
Using these definitions we can now extend the results of Theorem \ref{thm:DiagonalsUniqueCrossingPoint} to the conditioned case in the following theorem.
\begin{thm} \label{thm:DiagonalsUniqueCrossingPoint_Cond}
Let $\rvec{u} - \rvec{x} - \rvec{y}$ form a Markov chain. Then, the diagonal
entries of the matrix-valued function $\vectort \mapsto
\Q(\rvec{x}|\rvec{u}, \mat{\Lambda}_{\rvec{x_G}}, \vectort)$, defined
in (\ref{eq:QmatrixConditioned}), have no nonnegative-to-negative zero crossings
and, at most, a single negative-to-nonnegative zero crossing in the range
$\vectort \in [0, \infty)$. Moreover, let $\vectort_0 \in [0,  \infty)$ be the
negative-to-nonnegative crossing point for $[\Q(\rvec{x}|\rvec{u}, \mat{\Lambda}_{\rvec{x_G}}, \vectort)]_{ii}$. Then,
\begin{enumerate}
\item $[\Q(\rvec{x}|\rvec{u}, \mat{\Lambda}_{\rvec{x_G}}, 0)]_{ii} \leq 0$.
\item $[\Q(\rvec{x}|\rvec{u}, \mat{\Lambda}_{\rvec{x_G}}, \vectort)]_{ii}$ is a
strictly increasing function in the range $\vectort \in [0, \vectort_0)$.
\item $[\Q(\rvec{x}|\rvec{u}, \mat{\Lambda}_{\rvec{x_G}}, \vectort)]_{ii} \geq
0$ for all $\vectort \in [\snr_0, \infty)$.
\item When $\lim_{\vectort \rightarrow \infty} [\Chan(\vectort)]_{ii} = \infty$ we have that $\lim_{\vectort \to \infty} [\Q(\rvec{x}|\rvec{u},
\mat{\Lambda}_{\rvec{x_G}}, \vectort)]_{ii} = 0$.
\item $[\Q(\rvec{x}|\rvec{u}, \mat{\Lambda}_{\rvec{x_G}}, \vectort)]_{ii}$ is a continuous and monotonically increasing function in $[\mat{\Lambda}_{\rvec{x_G}}]_{ii}$.
\end{enumerate}
\end{thm}
\begin{IEEEproof}
If $[\rvec{x}]_i$ is Gaussian distributed (independent of $\rvec{u}$ and independent of the other
entries of the vector $\rvec{x}$) with variance equal to $[\Cov{x}]_{ii} =
[\mat{\Lambda}_{\rvec{x_G}}]_{ii}$, then $[\Q(\rvec{x}|\rvec{u},
\mat{\Lambda}_{\rvec{x_G}}, \vectort)]_{ii} = 0$, $\forall \vectort$, which
also fulfills Theorem \ref{thm:DiagonalsUniqueCrossingPoint_Cond}.
Thus, from this point, we can assume that 
$[\rvec{x}]_i$ is not ``Gaussian
distributed, independent of $\rvec{u}$ and independent of the other entries of $\rvec{x}$, and such that
$[\Cov{x}]_{ii} = [\mat{\Lambda}_{\rvec{x_G}}]_{ii}$''.

In this conditioned case, it is harder to determine, up front, all cases in which the function $[\Q(\rvec{x}|\rvec{u}, \mat{\Lambda}_{\rvec{x_G}},\vectort)]_{ii}$ has no zeros. Thus, contrary to the approach used in the proof of
Theorem \ref{thm:DiagonalsUniqueCrossingPoint}, we first prove that no nonnegative-to-negative and at most one negative-to-nonnegative zero crossings of $[\Q(\rvec{x}|\rvec{u}, \mat{\Lambda}_{\rvec{x_G}},
\vectort)]_{ii}$ can occur. The first property is a direct consequence of this, and there is no need to determine the exact conditions under which the function has no zeros. This approach could have also been used in proving Theorem \ref{thm:DiagonalsUniqueCrossingPoint}, however in the unconditioned case we can easily determine the set of cases in which $[\Q(\rvec{x}, \mat{\Lambda}_{\rvec{x_G}},\vectort)]_{ii}$ has no zeros.

Similarly to the proof of Theorem \ref{thm:DiagonalsUniqueCrossingPoint}, in order to prove that no nonnegative-to-negative and at most one negative-to-nonnegative zero crossings of $[\Q(\rvec{x}|\rvec{u}, \mat{\Lambda}_{\rvec{x_G}},
\vectort)]_{ii}$ can occur, we
only need to show that the derivative of
$[\Q(\rvec{x}|\rvec{u}, \mat{\Lambda}_{\rvec{x_G}}, \vectort)]_{ii}$ with respect to
$\vectort$ is positive for all values of $\vectort$ for which $[\Q(\rvec{x}|\rvec{u},
\mat{\Lambda}_{\rvec{x_G}}, \vectort)]_{ii} < 0$.
According to equations (\ref{eq:derivativeHi_1}) and (\ref{eq:derivativeHi_3}) we have the following lower bound:
\begin{align}
\Jacob_\vectort [\Q(\rvec{x}|\rvec{u} = \rvecr{u}, \mat{\Lambda}_{\rvec{x_G}}, \vectort)]_{ii} &= 2 \sum_{j=1}^{\dim} [ \ChanDChan(\vectort) ]_{jj} \left(
\Esp[1]{[\CMSE{x_{\rvecr{u}}}{y}]_{ij}^2} - [ \EM_{G}(\mat{\Lambda}_{\rvec{x_G}}, \vectort)
]_{ij}^2 \right) \label{eq:derivativeHi_1_Cond} \\
& \geq  2 [\ChanDChan(\vectort)]_{ii} \left( \Esp[1]{[\CMSE{x_{\rvecr{u}}}{y}]_{ii}^2} -
[\EM_{G}(\mat{\Lambda}_{\rvec{x_G}}, \vectort)]_{ii}^2
\right) . \label{eq:derivativeHi_3_Cond}
\end{align}
Now we can take expectation over $\rvec{u}$ on both sides and attain the following:
\begin{align}
\Jacob_\vectort [\Q(\rvec{x}|\rvec{u}, \mat{\Lambda}_{\rvec{x_G}}, \vectort)]_{ii} &= \Earg[1]{\Jacob_\vectort [\Q(\rvec{x}|\rvec{u}, \mat{\Lambda}_{\rvec{x_G}}, \vectort)]_{ii}}{\rvec{u}} \label{eq:derivativeHi_4_Cond} \\
& \geq \Earg[1]{2 [\ChanDChan(\vectort)]_{ii} \left( \Esp[1]{[\CMSE{x_{\rvecr{u}}}{y}]_{ii}^2} -
[\EM_{G}(\mat{\Lambda}_{\rvec{x_G}}, \vectort)]_{ii}^2
\right)}{\rvec{u}} \label{eq:derivativeHi_5_Cond} \\
&= 2 [\ChanDChan(\vectort)]_{ii} \left( \Esp[1]{[\CMSE{x}{y,\rvec{u}}]_{ii}^2} -
[\EM_{G}(\mat{\Lambda}_{\rvec{x_G}}, \vectort)]_{ii}^2  \right) \label{eq:derivativeHi_6_Cond} \\
& > 2 [\ChanDChan(\vectort)]_{ii} \left( \Esp[1]{[\CMSE{x}{y,\rvec{u}}]_{ii}^2} -
\left(\Esp[1]{[\CMSE{x}{y,\rvec{u}}]_{ii}}\right)^2  \right) \label{eq:derivativeHi_7_Cond} \\
& \geq 0 \label{eq:derivativeHi_8_Cond}
\end{align}
where (\ref{eq:derivativeHi_5_Cond}) is due to (\ref{eq:derivativeHi_3_Cond}), (\ref{eq:derivativeHi_7_Cond}) follows from the assumption $[\Q(\rvec{x}|\rvec{u}, \mat{\Lambda}_{\rvec{x_G}}, \vectort)]_{ii} < 0$ and (\ref{eq:derivativeHi_8_Cond}) can be derived from Jensen's inequality.

Observe that the inequality in (\ref{eq:derivativeHi_7_Cond}), which
holds for values of $\vectort$ such that $[\Q(\rvec{x}|\rvec{u}, \mat{\Lambda}_{\rvec{x_G}}, \vectort)]_{ii} < 0$, also
proves the second item in Theorem \ref{thm:DiagonalsUniqueCrossingPoint_Cond} and the
third one follows directly from the inexistence of nonnegative-to-negative zero
crossings. Regarding the fourth item, it is clear that
$\lim_{\vectort \to \infty} [\Q(\rvec{x}|\rvec{u}, \mat{\Lambda}_{\rvec{x_G}}, \vectort)]_{ii} = 0$, as both terms in the expression of $[\Q(\rvec{x}|\rvec{u}, \mat{\Lambda}_{\rvec{x_G}}, \vectort)]_{ii}$ in (\ref{eq:QmatrixConditioned}) tend to zero, when $\lim_{\vectort \rightarrow \infty} [\Chan(\vectort)]_{ii} = \infty$.
Finally, the last property is a direct consequence of the definition of the function $\Q(\rvec{x}|\rvec{u}, \mat{\Lambda}_{\rvec{x_G}},
\vectort)$ in (\ref{eq:QmatrixConditioned}).
\end{IEEEproof}



We now extend the definition of the function $\BQ(\rvec{x},\mat{\Lambda}_{\rvec{x_G}},\vectort)$ (\ref{eq:BQ}) and the function $\BQsum(\rvec{x},\mat{\Lambda}_{\rvec{x_G}},\vectort)$ (\ref{eq:BQsum}) to the conditioned case:
\begin{align} \label{eq:BQ_cond}
\BQ(\rvec{x}|\rvec{u} = \rvecr{u},\mat{\Lambda}_{\rvec{x_G}},\vectort) &=  [\ChanDChan(\vectort)]_{ii} [\Q(\rvec{x}|\rvec{u} = \rvecr{u},\mat{\Lambda}_{\rvec{x_G}}, \vectort)]_{ii} \\
\BQ(\rvec{x}|\rvec{u},\mat{\Lambda}_{\rvec{x_G}},\vectort) &=  [\ChanDChan(\vectort)]_{ii} [\Q(\rvec{x}|\rvec{u},\mat{\Lambda}_{\rvec{x_G}}, \vectort)]_{ii}
\end{align}
and also,
\begin{align} \label{eq:BQsum_cond}
\BQsum(\rvec{x}|\rvec{u} = \rvecr{u},\mat{\Lambda}_{\rvec{x_G}},\vectort) &= \sum_{i=1}^{\dim} \BQ(\rvec{x}|\rvec{u} = \rvecr{u},\mat{\Lambda}_{\rvec{x_G}},\vectort) \\
\BQsum(\rvec{x}|\rvec{u} ,\mat{\Lambda}_{\rvec{x_G}},\vectort) &= \sum_{i=1}^{\dim} \BQ(\rvec{x}|\rvec{u},\mat{\Lambda}_{\rvec{x_G}},\vectort)
\end{align}
For which we can extend corollaries \ref{cor:BQ} and \ref{cor:BQsum} as follows,
\begin{cor} \label{cor:BQ_cond}
Let $\rvec{u} - \rvec{x} - \rvec{y}$ form a Markov chain such that the random
vector $\rvec{x}|\rvec{u}=\rvecr{u} \in \R^{\dim}$ has covariance matrix $\CovMat_{\rvec{x}|\rvec{u}=\rvecr{u}}$. The function $\BQ(\rvec{x}|\rvec{u},\mat{\Lambda}_{\rvec{x_G}},\vectort)$ has the following properties:
\begin{enumerate}

\item $\BQ(\rvec{x}|\rvec{u},\mat{\Lambda}_{\rvec{x_G}},0) = 0$

\item It has at most a single negative-zero-positive crossing in the range $\vectort \in
(0, \infty)$.

\item When $\lim_{\vectort \rightarrow \infty} [\Chan(\vectort)]_{ii} = \infty$ we have that, $\lim_{\vectort \to \infty} \BQ(\rvec{x}|\rvec{u},\mat{\Lambda}_{\rvec{x_G}},\vectort) = 0$.

\item If $[\mat{\Lambda}_{\rvec{x_G}}]_{ii} = [\Cov{\rvec{x}}]_{ii}$, then $\BQ(\rvec{x}|\rvec{u},\mat{\Lambda}_{\rvec{x_G}},\vectort) \geq 0$ for all $\vectort$. Furthermore, $\BQ(\rvec{x}|\rvec{u},\mat{\Lambda}_{\rvec{x_G}},\vectort)$ is a continuous and monotonically increasing function in $[\mat{\Lambda}_{\rvec{x_G}}]_{ii}$.

\end{enumerate}
\end{cor}
\begin{IEEEproof} The first three properties follow directly from Theorem
\ref{thm:DiagonalsUniqueCrossingPoint_Cond} and the fact that $[\ChanDChan(\vectort)]_{ii}$ is zero at $\vectort=0$ and nonnegative for all other values of $\vectort \in (0, \infty)$ and $[\ChanDChan(\vectort)]_{ii}$ goes to $\infty$ in a linear rate, as shown in Lemma \ref{lem:path}. The fourth property is a direct result of Lemma \ref{lem:indGaussianOnTop}, and the fifth property in Theorem
\ref{thm:DiagonalsUniqueCrossingPoint_Cond}.
\end{IEEEproof}
\begin{cor} \label{cor:BQsum_cond}
Let $\rvec{u} - \rvec{x} - \rvec{y}$ form a Markov chain. The function $\BQsum(\rvec{x}|\rvec{u},\mat{\Lambda}_{\rvec{x_G}},\vectort)$ is either negative for all $\vectort$, or there exists $\vectort' \in [0,\infty)$ such that for all $\vectort > \vectort'$ the function $\BQsum(\rvec{x}|\rvec{u},\mat{\Lambda}_{\rvec{x_G}},\vectort)$ is nonnegative. Moreover, when $\lim_{\vectort \rightarrow \infty} [\Chan(\vectort)]_{ii} = \infty$ we have that, $\lim_{\vectort \to \infty} \BQ(\rvec{x}|\rvec{u},\mat{\Lambda}_{\rvec{x_G}},\vectort) = 0$, and if $[\mat{\Lambda}_{\rvec{x_G}}]_{ii} = [\Cov{\rvec{x}}]_{ii}$ for all $i$, then $\BQsum(\rvec{x}|\rvec{u},\mat{\Lambda}_{\rvec{x_G}},\vectort) \geq 0$ for all $\vectort$.
\end{cor}


\subsection{Properties of the Mutual Information} \label{ssec:independentMIproperties}
\label{ssec:mutualInformationInd}
So far, we have seen properties of the matrix $\Q(\rvec{x}, \mat{\Lambda}_{\rvec{x_G}},
\vectort)$ or, more precisely, of its diagonal elements. We have seen that these properties extend naturally to the conditioned case, and also to the function $\BQ(\rvec{x},\mat{\Lambda}_{\rvec{x_G}},\vectort)$ and its conditioned version. In this section, our goal is to use these results to derive new properties on the mutual information between the input and the output of parallel Gaussian channels.
In order to derive these results we put to use the I-MMSE relationship, as given in equations (\ref{eq:lineIntegral})-(\ref{eq:lineIntegral2}) and (\ref{eq:lineIntegral_cond})-(\ref{eq:lineIntegral_cond2}).

For the sake of compactness we will write the properties in this section only for the more general, conditioned case, from which one can easily derive the respective unconditioned theorems.

\begin{thm} \label{thm:CrossingPointUpperBoundUsingIndpDist}
Let $\rvec{u} - \rvec{x} - \rvec{y}$ form a Markov chain. Assume an independent Gaussian input, $\rvec{x_G}$, with covariance $\mat{\Lambda}_{\rvec{x_G}}$, such that for all $i$,
\begin{gather} \label{eq::CrossingPointUpperBoundUsingIndpDist}
\Icond{ [\rvec{x}]_i}{ [\rvec{y}(\vectortEq)]_i}{ \rvec{u}} = \Igen{ [\rvec{x_G}]_i}{ [\rvec{y_G}(\vectortEq)]_i}
\end{gather}
where
\begin{align} \label{eq::CrossingPointUpperBoundUsingIndpDist_2}
\rvec{y}(\vectortEq) &= \Chan(\vectortEq) \rvec{x} + \rvec{n} \quad \textrm{and}\\
\rvec{y_G}(\vectortEq) &= \Chan(\vectortEq) \rvec{x_G} + \rvec{n}. \label{eq::CrossingPointUpperBoundUsingIndpDist_3}
\end{align}
Then $\BQsum(\rvec{x}|\rvec{u},\mat{\Lambda}_{\rvec{x_G}},\vectort) \geq 0$ for all $\vectort \geq \vectortEq$.
\end{thm}
\begin{IEEEproof}
Let us define $\rvecIndc{x} \in \R^{\dim}$ as a random vector with independent elements when conditioned on $\rvec{u}$, and with distribution of each pair $\left( [\rvecIndc{x}]_i, \rvec{u}  \right)$ being the same as the marginal distribution of the corresponding pair $\left( [\rvec{x}]_i, \rvec{u} \right)$.
Thus, $\left[ \MSE{\rvecIndc{x}|\rvec{u}} \right]_{ii}$ is basically the MMSE of $[ \rvecIndc{x}]_i$ from $\rvec{u}$ and $[\rvec{y}(\vectort)]_i$, which is:
\begin{eqnarray}\label{eq:Y_i}
[\rvec{y}(\vectort)]_i = \left[ \Chan(\vectort) \rvecIndc{x} +\rvec{n} \right]_i =
\left[ \Chan(\vectort) \right]_{ii} [\rvecIndc{x}]_i+[\rvec{n}]_i
\end{eqnarray}
where the equality holds due to the fact that the channel matrix
$\Chan(\vectort)$ is diagonal for all $\vectort$ and $\rvec{n}$ is
standard Gaussian.
Using these definitions we can give the following special case of (\ref{eq:lineIntegral_cond2}):
\begin{align} \label{eq:lineIntegralIndependent}
\Icond{[\rvec{x}]_i}{[\rvec{y}(\vectort)]_i}{\rvec{u}} &= \Icond{[\rvec{x}]_i}{\left[ \Chan(\vectort) \right]_{ii} [\rvecIndc{x}]_i+[\rvec{n}]_i}{\rvec{u}} \\
&= \int_{\tau=0}^{\vectort} \left[ \Chan(\tau) \right]_{ii}  \left[ \MSE{\rvecIndc{x}|\rvec{u}}(\tau) \right]_{ii} \left[ \Jacob_{\tau}\Chan(\tau) \right]_{ii} \d \tau \\
&=  \int_{\tau=0}^{\vectort} \left[ \ChanDChan(\tau) \right]_{ii}  \left[ \MSE{\rvecIndc{x}|\rvec{u}}(\tau) \right]_{ii} \d \tau .
\end{align}
Putting this together with the assumption, we have,
\begin{align}
0 = \Igen{ [\rvec{x_g}]_i}{ [\rvec{y_g}(\vectortEq)]_i} - \Icond{ [\rvec{x}]_i}{ [\rvec{y}(\vectortEq)]_i}{ \rvec{u}} = \int_{\tau=0}^{\vectortEq} \BQ(\rvecIndc{x}|\rvec{u},\mat{\Lambda}_{\rvec{x_G}},\tau) \d \tau .
\end{align}
Now, due to Corollary \ref{cor:BQsum_cond} we can conclude that there exists a $\vectort_0 \in [0,\vectortEq]$ such that $\BQ(\rvecIndc{x}|\rvec{u},\mat{\Lambda}_{\rvec{x_G}},\vectort) \geq 0$ for all $\vectort > \vectort_0$ and as a result, $\BQ(\rvecIndc{x}|\rvec{u},\mat{\Lambda}_{\rvec{x_G}},\vectort) \geq 0$ for all $\vectort > \vectortEq$. Now,
for all $\vectort$ we have that $[ \MSE{\rvec{x}|\rvec{u}}(\vectort) ]_{ii} \leq [ \MSE{\rvecIndc{x}|\rvec{u}}(\vectort) ]_{ii}$.
Thus, if the negative-zero-positive crossing of $\BQ(\rvecIndc{x}|\rvec{u},\mat{\Lambda}_{\rvec{x_G}},\vectort)$ is at $\vectort_0$, the
negative-zero-positive crossing of $\BQ(\rvec{x}|\rvec{u},\mat{\Lambda}_{\rvec{x_G}},\vectort)$ is at a $\vectort_0' \leq \vectort_0$.
From this we can conclude that also $\BQ(\rvec{x}|\rvec{u},\mat{\Lambda}_{\rvec{x_G}},\vectort) \geq 0$ for all $\vectort > \vectortEq$. Finally, since this holds for every $i$, it also holds for the summation over $i$, \ie, for the function $\BQsum(\rvec{x}|\rvec{u},\mat{\Lambda}_{\rvec{x_G}},\vectort)$, concluding the proof.
\end{IEEEproof}

We are now ready to give the main theorem of this section.
\begin{thm} \label{thm:MIparallel}
Let $\rvec{u} - \rvec{x} - \rvec{y}$ form a Markov chain. For any $\vectortEq \in [0, \infty)$, there exists an independent Gaussian input, $\rvec{x_G}$, with covariance $\mat{\Lambda}_{\rvec{x_G}}$ such that the following properties hold:
\begin{enumerate}

\item $\BQsum(\rvec{x}|\rvec{u},\mat{\Lambda}_{\rvec{x_G}},\vectort) \geq 0$ for all $\vectort \geq \vectortEq$.

\item $\Icond{\rvec{x}}{\rvec{y}(\vectortEq)}{\rvec{u}} = \Igen{\rvec{x_G}}{\rvec{y_G}(\vectortEq)}$, where $\rvec{y}(\vectortEq)$ and $\rvec{y_G}(\vectortEq)$ are as defined in (\ref{eq::CrossingPointUpperBoundUsingIndpDist_2}) and (\ref{eq::CrossingPointUpperBoundUsingIndpDist_3}) respectively.

\item $[\mat{\Lambda}_{\rvec{x_G}}]_{ii} \leq [\Cov{\rvec{x}} ]_{ii}$ for all $i$.
\end{enumerate}

\end{thm}
\begin{IEEEproof}
We provide a constructive proof, and show how one can build an independent Gaussian input distribution complying with all three requirements. We begin by examining the meaning of the second requirement. First, recall the I-MMSE relationship in the parallel setting, given in equation (\ref{eq:lineIntegral_cond2}),
\begin{gather} \label{eq:thm:MIparallel_proof_1}
\Icond{\rvec{x}}{\rvec{y}(\vectortEq)}{\rvec{u}} =  \int_{\tau=0}^{\vectortEq} \Tr \left(
\ChanDChan(\tau) \MSE{\rvec{x}|\rvec{u}}(\tau)
\right) \d \tau.
\end{gather}
Now, the second requirement is equivalent to the following equality,
\begin{align} \label{eq:thm:MIparallel_proof_2}
0 = \Igen{\rvec{x_G}}{\rvec{y_G}(\vectortEq)} - \Icond{\rvec{x}}{\rvec{y}(\vectortEq)}{\rvec{u}} &= \int_{\tau=0}^{\vectortEq} \Tr \left(
\ChanDChan(\tau) \Q(\rvec{x}|\rvec{u},\mat{\Lambda}_{\rvec{x_G}},\tau) \right) \d \tau \\ \nonumber
& = \int_{\tau=0}^{\vectortEq} \sum_{i=1}^{\dim} \BQ(\rvec{x}|\rvec{u},\mat{\Lambda}_{\rvec{x_G}},\tau)
\d \tau \\ \nonumber
&= \sum_{i=1}^{\dim} \int_{\tau=0}^{\vectortEq} \BQ(\rvec{x}|\rvec{u},\mat{\Lambda}_{\rvec{x_G}},\tau)
\d \tau.
\end{align}
Thus, we wish to show the existence of an independent Gaussian input distribution which complies with requirements 1, 3 and (\ref{eq:thm:MIparallel_proof_2}). There are different ways to attain equality in (\ref{eq:thm:MIparallel_proof_2}), however since we need only to show the existence of a specific independent Gaussian distribution, we follow one possible approach, which is to require the following,
\begin{gather} \label{eq:thm:MIparallel_proof_3}
\int_{\tau=0}^{\vectortEq} \BQ(\rvec{x}|\rvec{u},\mat{\Lambda}_{\rvec{x_G}},\tau)
\d \tau = 0, \quad \forall i.
\end{gather}
Now, according to the fourth property in Corollary \ref{cor:BQ_cond} we know that,
\begin{gather} \label{eq:thm:MIparallel_proof_4}
\BQ(\rvec{x}|\rvec{u},\mat{\Lambda}_{\rvec{x_G}},\vectort) \geq 0
\end{gather}
for all $\vectort$, when $[\mat{\Lambda}_{\rvec{x_G}}]_{ii} = [\Cov{\rvec{x}}]_{ii}$, and that it is continuous and monotonically increasing in the value of $[\mat{\Lambda}_{\rvec{x_G}}]_{ii}$ (and trivially negative, for all $\vectort$, when $[\mat{\Lambda}_{\rvec{x_G}}]_{ii} = 0$). Thus, there exists a number $\eta_i \in [0,1]$ such that setting $[\mat{\Lambda}_{\rvec{x_G}}]_{ii} = \eta_i [\Cov{\rvec{x}}]_{ii}$ results with the equality in (\ref{eq:thm:MIparallel_proof_3}).
Due to the second property in Corollary \ref{cor:BQ_cond}, we know that either  $\BQ(\rvec{x}|\rvec{u},\mat{\Lambda}_{\rvec{x_G}},\vectort) = 0$ for all $\vectort$ or that there exists a single negative-zero-positive crossing in the range $[0,\vectortEq]$. In both cases the setting $[\mat{\Lambda}_{\rvec{x_G}}]_{ii} = \eta_i [\Cov{\rvec{x}}]_{ii}$ results with $\BQ(\rvec{x}|\rvec{u},\mat{\Lambda}_{\rvec{x_G}},\vectort) \geq 0$ for all $\vectort > \vectortEq$. Since there exists such an $\eta_i$ for every $i$ we comply also with requirements 1 and 3, and conclude the proof.
\end{IEEEproof}

\begin{rem}
Note that the above choice of $\mat{\Lambda}_{\rvec{x_G}}$ does not necessarily imply $\mat{\Lambda}_{\rvec{x_G}} \preceq \Cov{\vec{x}}$. However, we can conclude that $\mat{\Lambda}_{\rvec{x_G}} \nsucc \Cov{\vec{x}}$.
\end{rem}

The following is a simple corollary of the above theorem.
\begin{cor} \label{cor:MIparallel_ind}
Given any arbitrary \emph{independent} input distribution over $\rvec{x} \in \R^{\dim}$, with covariance $\mat{\Lambda}_{\rvec{x}}$, and any $\vectortEq$, there exists an independent Gaussian input, $\rvec{x_G}$, with covariance $\mat{\Lambda}_{\rvec{x_G}}$ such that
\begin{align} \label{eq:cor:MIparallel_ind}
\Igen{\rvec{x}}{\Chan(\vectortEq) \rvec{x} + \rvec{n}} &= \Igen{\rvec{x_G}}{\Chan(\vectortEq) \rvec{x_G} + \rvec{n}} \\
\mat{\Lambda}_{\rvec{x_G}} & \preceq \mat{\Lambda}_{\rvec{x}} \\
\textrm{and} \quad \EM_G (\mat{\Lambda}_{\rvec{x_G}}, \vectortEq) & \preceq \MSE{\rvec{x}}(\vectortEq)
\end{align}
\end{cor}

\subsection{Application: The Degraded Parallel Gaussian BC Capacity
Region under Per-antenna Power Constraint} \label{ssec:applicationInd}

We now show that Theorem \ref{thm:MIparallel} can be used in providing a converse proof for the \emph{degraded} parallel Gaussian BC capacity region under a per-antenna power constraint.
We consider the following model,
\begin{align} \label{eq:app:perAntenna_model}
\rvec{y}_1[\timeIn] &= \Chan_{1} \rvec{x}[\timeIn] + \rvec{n}_1[\timeIn] \nonumber \\
\rvec{y}_2[\timeIn] &= \Chan_{2} \rvec{x}[\timeIn] + \rvec{n}_2[\timeIn]
\end{align}
where $\rvec{n}_1[\timeIn]$ and $\rvec{n}_2[\timeIn]$ are standard additive Gaussian noise vectors independent for different time indices $\timeIn$, and $\Chan_{1}$ and $\Chan_{2}$ are diagonal positive semidefinite matrices such that $\Chan_{1} \preceq \Chan_{2}$. $\rvec{x} \in \R^{\dim}$ is the random input vector, and it is assumed independent for different time indices $\timeIn$. Note that $\timeIn$ is the time index and should not be confused with the scalar parameter $\vectort$ which is used as a ``MIMO $\snr$ parameter'', \ie, the parameter $\vectort$ determines the channel matrix $\Chan(\vectort)$.

We consider a per-antenna power constraint:
\begin{eqnarray} \label{eq:app:perAntenna_constraintPerAntenna}
\left[ \Esp[1]{\rvec{x}\rvec{x}^\T } \right]_{ii} \leq \Power_i \quad \forall i, 1 \leq i \leq \dim  \text{.}
\end{eqnarray}
Since we have a \emph{degraded} BC, we can use the single-letter expression given in \cite{Comments},
\begin{align} \label{eq:app:perAntenna_singleLetterExpr}
\rate_1 & \leq \Igen{\rvec{u}}{\rvec{y}_1} \nonumber \\
\rate_2 & \leq \Icond{\rvec{x}}{\rvec{y}_2}{\rvec{u}}
\end{align}
where $\rvec{u}$ is an auxiliary random vector over a certain alphabet
that satisfies the Markov relation $\rvec{u} - \rvec{x} - (\rvec{y}_1,\rvec{y}_2)$. The following proof was originally given for the scalar
Gaussian BC in \cite{PROP,PROP_full} and we now extend it to the \emph{degraded} parallel Gaussian channel. Using Lemma \ref{lem:path} we can construct a path such that:
\begin{align} \label{eq:app:perAntenna_path}
\Chan(\vectort_2) & = \Chan_2 \nonumber \\
\Chan(\vectort_1) & = \Chan_1 \nonumber \\
\Chan(0) & = \mat{0}
\end{align}
where $0 \leq \vectort_1 \leq \vectort_2$ and $\Chan(\vectort)$ is diagonal for all $\vectort \in [0, \vectort_2]$.

Now, assume a pair $(\rvec{u},\rvec{x})$ such that $\rvec{x}$ has covariance $\Cov{\rvec{x}}$. According to Theorem \ref{thm:MIparallel}, there exists an independent Gaussian vector, $\rvec{x}_G$, with covariance matrix $\mat{\Lambda}_{\rvec{x_G}}$ such that the following properties hold:
\begin{align} \label{eq:app:perAntenna_usingTheoremP1}
\Icond{\rvec{x}}{\rvec{y}_1}{\rvec{u}} = \Icond{\rvec{x}}{\Chan(\vectort_1) \rvec{x} + \rvec{n}}{\rvec{u}} &=
\Igen{\rvec{x_g}}{\rvec{y_g}(\vectort_1)} = \Igen{\rvec{x_g}}{\Chan(\vectort_1) \rvec{x_g} + \rvec{n}} \\ \label{eq:app:perAntenna_usingTheoremP2}
\BQsum(\rvec{x}|\rvec{u},\mat{\Lambda}_{\rvec{x_G}},\vectort) &\geq 0, \quad \forall \vectort \geq \vectort_1 \\
\label{eq:app:perAntenna_usingTheoremP3}
[\mat{\Lambda}_{\rvec{x_G}}]_{ii} & \leq [\Cov{\rvec{x}} ]_{ii} \quad \forall i.
\end{align}
Using the I-MMSE relationship (\ref{eq:lineIntegral_cond2}) we can write,
\begin{align} \label{eq:app:perAntenna_IMMSE_1}
\Igen{\rvec{x_g}}{\Chan(\vectort) \rvec{x_g} + \rvec{n}} - \Icond{\rvec{x}}{\Chan(\vectort) \rvec{x} + \rvec{n}}{\rvec{u}} &=
\int_{\tau=0}^{\vectort} \Tr \left(
\ChanDChan(\tau) \Q(\rvec{x}|\rvec{u},\mat{\Lambda}_{\rvec{x_G}},\tau) \right) \d \tau \\
\label{eq:app:perAntenna_IMMSE_2}
& = \int_{\tau=0}^{\vectort} \sum_{i=1}^{\dim} \BQ(\rvec{x}|\rvec{u},\mat{\Lambda}_{\rvec{x_G}},\tau)
\d \tau \\
\label{eq:app:perAntenna_IMMSE_3}
&= \int_{\tau=0}^{\vectort} \BQsum(\rvec{x}|\rvec{u},\mat{\Lambda}_{\rvec{x_G}},\tau)
\d \tau .
\end{align}
Using the above properties on (\ref{eq:app:perAntenna_IMMSE_3}) we have that for any $\vectort' > \vectort_1$,
\begin{align} \label{eq:app:perAntenna_conclusingFromThm}
\Igen{\rvec{x_g}}{\Chan(\vectort') \rvec{x_g} + \rvec{n}} - \Icond{\rvec{x}}{\Chan(\vectort') \rvec{x} + \rvec{n}}{\rvec{u}}
&= \int_{\tau=0}^{\vectort_1} \BQsum(\rvec{x}|\rvec{u},\mat{\Lambda}_{\rvec{x_G}},\tau)
\d \tau + \int_{\tau=\vectort_1}^{\vectort'} \BQsum(\rvec{x}|\rvec{u},\mat{\Lambda}_{\rvec{x_G}},\tau)
\d \tau \\
&= 0 + \int_{\tau=\vectort_1}^{\vectort'} \BQsum(\rvec{x}|\rvec{u},\mat{\Lambda}_{\rvec{x_G}},\tau)
\d \tau \geq 0
\end{align}
where the second transition is due to (\ref{eq:app:perAntenna_usingTheoremP1}) and the inequality is due to (\ref{eq:app:perAntenna_usingTheoremP2}). Thus, we have shown the existence of an independent Gaussian vector, $\rvec{x}_G$, with covariance matrix $\mat{\Lambda}_{\rvec{x_G}}$, with the following properties:
\begin{align}
\Icond{\rvec{x}}{\Chan_1 \rvec{x} + \rvec{n}}{\rvec{u}} &= \frac{1}{2}\log |\Identity + \Chan_1 \mat{\Lambda}_{\rvec{x_G}} \Chan_1^\T |  \\
\Icond{\rvec{x}}{\Chan_2 \rvec{x} + \rvec{n}}{\rvec{u}} &\leq \frac{1}{2}\log |\Identity + \Chan_2 \mat{\Lambda}_{\rvec{x_G}} \Chan_2^\T |  \\
\textrm{and} \quad [\mat{\Lambda}_{\rvec{x_G}}]_{ii} & \leq [\Cov{\rvec{x}} ]_{ii} \quad \forall i.
\end{align}
Using these properties on the single-letter expression (\ref{eq:app:perAntenna_singleLetterExpr}) we obtain the following outer bound:
\begin{align} \label{eq:using Theorem6_u_5}
\rate_1 & \leq \Igen{\rvec{u}}{\rvec{y}_1} = \Igen{\rvec{x}}{\rvec{y}_1} - \Icond{\rvec{x}}{\rvec{y}_1}{\rvec{u}} \nonumber \\
& \leq \frac{1}{2} \log | \Identity + \Chan_1 \mat{P} \Chan_1^\T |
- \frac{1}{2} \log | \Identity + \Chan_1 \mat{\Lambda}_{\rvec{x_G}} \Chan_1^\T |
= \frac{1}{2} \log  \frac{| \Identity + \Chan_1 \mat{P} \Chan_1^\T |}{| \Identity + \Chan_1 \mat{\Lambda}_{\rvec{x_G}} \Chan_1^\T |} \\
\rate_2 & \leq \Icond{\rvec{x}}{\rvec{y}_2}{\rvec{u}} \leq
\frac{1}{2} \log | \Identity + \Chan_2 \mat{\Lambda}_{\rvec{x_G}} \Chan_2^\T |
\end{align}
where $\mat{P}$ is a diagonal matrix with $[\mat{P}]_{ii} = \Power_i$ for all $i$. This outer bound is tight and the
achievability is well-known using superposition coding. 
This approach can be extended to the M-user scenario as shown in Appendix \ref{app:converseProofBCPerAntennaMUsers}.

\section{Vector Channel: Comparing with a General Gaussian Distribution}
\label{sec:generalGaussian}
In this section we extend our analysis of the previous section. We continue looking into the model given in (\ref{eq:generalModel}), limited to parallel channel matrices, however we now allow the Gaussian covariance matrix, defining the matrix $\Q$, to be any proper covariance matrix. In other words, we no longer limit ourselves to independent Gaussian inputs. For this, more general setting, we will see in Section \ref{ssec:singleCrossingPointEig} that a ``single crossing point'' property occurs for each and every eigenvalue of the matrix $\Q$. After extending this result to the conditioned case, in Section \ref{ssec:generalConditioned}, we will use the I-MMSE relationship, in Section \ref{ssec:MutualInformationEig}, to show the effect of this property on information-theoretic quantities, and more specifically on the mutual information. We will relate these results to the Fisher information in Section \ref{ssec:fisher}. Finally, in Sections \ref{ssec:BCcovarianceConstraint} and \ref{ssec:compBCcovarianceConstraint} we will put these results to use in the \emph{degraded} BC capacity converse proof, for both the compound and non-compound scenarios.


\subsection{Single Crossing Point for Each Eigenvalue of $\Qt$}
\label{ssec:singleCrossingPointEig}

In this section we prove the main result of this paper: showing that each eigenvalue of the matrix $\Q$ has at most a single negative-to-nonnegative zero crossing. This is, to our understanding, not an intuitive extension of the ``single crossing point'' property, which emphasizes the importance of the eigenvalues in the analysis of MIMO scenarios.


For the proof of the main theorem, we require the following lemma, which might also be of interest on its own.
\begin{lem} \label{lem:LemmaLowerBound}
The following lower bound holds:
\begin{align} \label{eq:LemmaLowerBound}
\Jacob_{\vectort}\Q(\rvec{x}, \Cov{\rvec{x_G}}, \vectort) \succeq 2 \left( \MSE{x}(\vectort)
\ChanDChan(\vectort) \MSE{x}^\T (\vectort) - \EM_G(\vectort)
\ChanDChan(\vectort) \EM_G^\T(\vectort) \right)
\end{align}
where $\ChanDChan(\vectort)$ was defined in (\ref{eq:definitionB}) and assumed a positive semidefinite diagonal matrix for all $\vectort$ (see Lemma \ref{lem:path}).
\end{lem}
\begin{IEEEproof}
See Appendix \ref{app:LemmaLowerBound}.
\end{IEEEproof}
We are now ready to proceed to the main result of the paper:
\begin{thm} \label{thm:TheoremSingleCrossingEig}
Each eigenvalue of $\Q(\rvec{x}, \Cov{\rvec{x_g}}, \vectort)$ has, \emph{at most}, a single negative-to-nonnegative zero crossing of
the horizontal axis.
\end{thm}
\begin{IEEEproof}
Loosely speaking, the proof is based on proving that, once an eigenvalue has
become (or is) nonnegative, it cannot become negative. Thus, from the (weak)
continuity of the eigenvalues as a function of $\vectort$, that follows from
\cite[App.~D]{Horn}, the eigenvalues can cross the horizontal axis, at most,
once.
Also from continuity arguments, it is easy to see that we must limit our study
of the eigenvalues of $\Q( \rvec{x}, \Cov{x_g}, \vectort)$ to the values of $\vectort$ where the matrix
$\Q( \rvec{x}, \Cov{x_g}, \vectort)$ is singular (i.e., a subset of its eigenvalues are zero) as it is
the only possible situation where a zero crossing can occur.
%
Finally, throughout this proof and for the sake of simplicity we will use the
simplified notation $\Q( \rvec{x}, \Cov{x_g}, \vectort) = \Q(\vectort)
\triangleq \EM_G(\vectort) - \MSE{x}(\vectort)$ because the entire
proof is given for any constant setting of the input random vector $\rvec{x}$
and the Gaussian covariance $\Cov{x_g}$.

We begin by stating a few supporting results and giving some preliminary definitions.
\begin{lem} \label{lem:simultaneouslyD}
Let $\mat{A}$ and $\mat{B}$ be two $\dim$-dimensional positive semidefinite
matrices, \ie, $ \mat{A} \succeq \mat{0}$, $\mat{B} \succeq \mat{0}$. Then,
there exists an invertible matrix $\mat{S}$ such that both $\mat{S} \mat{A}
\mat{S}^T$ and $\mat{S} \mat{B} \mat{S}^T$ are diagonal matrices.
\end{lem}
\begin{IEEEproof} See Appendix \ref{app:simultaneouslyD}.
\end{IEEEproof}

Let us consider the simultaneous decomposition of $\{\EM_G(\vectort),
\MSE{x}(\vectort)\}$ according to Lemma \ref{lem:simultaneouslyD} as:
\begin{gather} \label{eq:gen_decomp}
\begin{split}
\EM_G(\vectort) &= \mat{V}(\vectort)^\T \mats{\Sigma}_G(\vectort)
\mat{V}(\vectort) \\
\MSE{x}(\vectort) &= \mat{V}(\vectort)^\T \mats{\Sigma}_{\rvec{x}}(\vectort)
\mat{V}(\vectort)
\end{split}
\end{gather}
where $\mat{V}(\vectort)$ is an invertible matrix and $\mats{\Sigma}_G(\vectort)$ and
$\mats{\Sigma}_{\rvec{x}}(\vectort)$ are diagonal matrices.
It will be convenient to define $\pQ (\vectort, \tau)$, for $\tau \geq
0$, according to
\begin{gather} \label{eq:qtilde}
\pQ (\vectort, \tau) = \mat{V}(\tau)^{-\T} \Q(\vectort )
\mat{V}(\tau)^{-1}
\end{gather}
where $\mat{V}(\tau)$ is the same as defined in (\ref{eq:gen_decomp}).

The remainder of the proof is split into two parts. In the first part we will
prove that each eigenvalue of $\pQ( \vectort, \tau)$ has at most a single
negative-to-nonnegative zero crossing. In the second part, we will show that
this property transfers to $\Q( \vectort )$, thus completing the proof.
Coincidentally, both parts of the proof will be based on contradiction
arguments, \ie, we assume that the opposite of what we want to prove is true
and, then, end up with an inconsistency.

\subsubsection{Single crossing point for the eigenvalues of $\pQ( \vectort, \tau)$}
Let us start by presenting a result on the differentiability of the eigenvalues
of a symmetric matrix with respect to some scalar parameter $\vectort$, which
was studied by Rellich in \cite[Ch.~1]{rellich:54}\footnote{Rellich studied the
eigenvalue differentiability for Hermitian matrices. We specialized his result
for the real case studied in this paper.}:
\begin{lem}{\cite[Th.~in p.~57]{rellich:54}} \label{lem:rellich}
Suppose that $\mat{A}(\vectort)$ is an $\dim$-dimensional symmetric matrix
defined on some open interval $\vectort \in (\vectort_1, \vectort_2)$. Suppose
that the derivative $\Jacob_{\vectort} \mat{A}(\vectort)$ exists and it is
continuous for each $\vectort \in (\vectort_1, \vectort_2)$. Then,
there exist $\dim$ functions $\lambda_i(\vectort)$, $i=1,\ldots, n$ with
continuous
derivatives in $\vectort \in (\vectort_1, \vectort_2)$, such that
\begin{gather}
 \mat{A}(\vectort) \vec{u}_i(\vectort) = \lambda_i(\vectort)
\vec{u}_i(\vectort), \quad i = 1, \ldots, \dim
\end{gather}
for some properly chosen orthonormal system of vectors $\vec{u}_i(\vectort)$, $i
= 1, \ldots, \dim$.
\end{lem}

Since $\pQ( \vectort, \tau)$ is a symmetric matrix whose derivative
$\Jacob_{\vectort} \pQ( \vectort, \tau)$ exists, Lemma \ref{lem:rellich} ensures
the existence of $\dim$ continuous and differentiable functions such that they
are equal to the eigenvalues of the matrix $\pQ(\vectort, \tau)$, for any
choice of $\tau$. These functions will be denoted from now on by
$\lambda_i(\vectort, \tau)$, for $i=1,\ldots, \dim$.

Now, let us assume that, at $\vectort = \vectort_0$, $k$ of these
eigenvalues (with $k \leq \dim$) are equal to zero, \ie, $\lambda_i(\vectort_0,
\tau) = 0$, for $i=1,\ldots, k$. Furthermore, we also assume that, from these
$k$ eigenvalues that are zero at $\vectort = \vectort_0$, $s$ of them (with
$s\leq k$) have a nonnegative-to-negative zero crossing at $\vectort =
\vectort_0$. To sum up, we assume that the differentiable functions
$\lambda_i(\vectort, \tau)$, with $i=1,\ldots, s$ have a nonnegative-to-negative
zero crossing at $\vectort = \vectort_0$.

Let us now present a property of differentiable functions that contain
nonnegative-to-negative zero crossings:
\begin{lem} \label{lem:ptnzc_diff}
Assume that $f(t)$ has a nonnegative-to-negative zero crossing at $t = t_0$ and
that $f(t)$ is differentiable in a neighborhood of $t_0$. Then, there exists a
positive value $\varepsilon$ such that
\begin{align}
 f(t) &< 0, \quad t \in (t_0, t_0 + \varepsilon), \label{eq:fptnzc_first} \\
 \Jacob_{t} f(t) &< 0, \quad t \in (t_0, t_0 + \varepsilon).
\label{eq:fptnzc_second}
\end{align}
\end{lem}
\begin{proof}
From Definition \ref{dfn:nonnegative2negative}, (\ref{eq:fptnzc_first})
follows immediately for any $\varepsilon \leq \epsilon$. The proof for
(\ref{eq:fptnzc_second}) follows easily from the mean value theorem and
elementary calculus.
\end{proof}

Applying Lemma \ref{lem:ptnzc_diff} to the set of functions $\lambda_i(\vectort,
\tau)$, with $i=1,\ldots, s$, we readily obtain:
\begin{gather} \label{eq:lambdaQt}
 \left.
\begin{array}{r}
\lambda_i(\vectort, \tau) <
0, \quad \vectort \in (\vectort_0, \vectort_0 + \varepsilon_i(\tau)) \\
\Jacob_{\vectort} \lambda_i(\vectort, \tau) < 0, \quad \vectort \in
(\vectort_0, \vectort_0 + \varepsilon_i(\tau))
\end{array}
\right\} \quad i = 1,\ldots, s
\end{gather}
where we have written $\varepsilon_i(\tau)$ to make explicit the dependence of
$\varepsilon_i$ on the specific value of $\tau$. For the sake of convenience, we
want to eliminate the dependence of $\varepsilon_i$ on $\tau$. A possible
method to eliminate this dependence is to define
\begin{gather}
 \varepsilon_i^\star = \inf_{\tau \in [\vectort_0, \vectort_0 + M]}
\varepsilon_i(\tau) = \min_{\tau
\in [\vectort_0, \vectort_0 + M]} \varepsilon_i(\tau) > 0
\end{gather}
where, for the sake of convenience, we have restricted the values of $\tau$ in
the interval $[\vectort_0, \vectort_0 + M]$, with $M$ being an arbitrary fixed positive value
(observe that since $\varepsilon_i(\tau)$ can be made arbitrarily small we can
always guarantee that $M > \varepsilon_i(\tau) \geq \varepsilon_i^\star$), and
where the second equality follows from the fact that the optimization set is a
closed interval and the third one follows from $\varepsilon_i(\tau) > 0$,
$\forall \tau$.

Consequently, after this simplification, we have that, assuming that the
differentiable functions $\lambda_i(\vectort, \tau)$, with $i=1,\ldots, s$ have
a nonnegative-to-negative zero crossing at $\vectort = \vectort_0$, they must
fulfill:
\begin{gather} \label{eq:lambdaQt_simpl}
 \left.
\begin{array}{r}
\lambda_i(\vectort, \tau) <
0, \quad \vectort \in (\vectort_0, \vectort_0 + \varepsilon_i^\star) \\
\Jacob_{\vectort} \lambda_i(\vectort, \tau) < 0, \quad \vectort \in
(\vectort_0, \vectort_0 + \varepsilon_i^\star)
\end{array}
\right\} \quad i = 1,\ldots, s.
\end{gather}
Now, we can particularize the expression above for the case where $\vectort =
\vectort_0 + \varepsilon_i^\star/2 \triangleq \vectort^\star$ and where we
also choose $\tau = \vectort^\star$. We obtain
\begin{gather} \label{eq:lambdaQt_critic_repeat}
 \left.
\begin{array}{r}
\lambda_i(\vectort^\star, \vectort^\star) <
0 \\ \left. \Jacob_{\vectort} {\lambda}_i(\vectort, \vectort^\star)
\right|_{\vectort = \vectort^\star} < 0
\end{array}
\right\} \quad i = 1,\ldots, s.
\end{gather}

From this point our goal is to prove that the two
conditions in (\ref{eq:lambdaQt_critic_repeat}) cannot both hold at the same
time. For that purpose, we need an expression for the derivative of the
eigenvalue function $\Jacob_{\vectort} {\lambda}_i(\vectort, \vectort^\star)$.
Since we have that $\lambda_i(\vectort_0, \tau) = 0$ for $i=1,\ldots,k$ (\ie,
the multiplicity of the zero eigenvalue is $k$) we cannot guarantee that the
multiplicity of the eigenvalue $\lambda_i(\vectort^\star, \vectort^\star)$ is
equal to one. From this point, we assume that the multiplicity of
$\lambda_i(\vectort^\star, \vectort^\star)$ is $l$.

Consequently, we now require the following result by Lancaster in
\cite[Th.~7]{lancaster:64} (it is also reproduced in
\cite[Ch.~8, Sec.~12, Th.~13]{Magnus}), which gives us an
expression for the derivatives of the multiple eigenvalues\footnote{The
assumptions \cite[Th.~7]{lancaster:64} are different than those in Lemma
\ref{lem:rellich}, but, once existence of the derivatives of the
eigenvalues has been established, their expression has to be the same.}:
\begin{lem}{\cite[Th.~7]{lancaster:64}} \label{lem:lancaster}
Under the assumptions in Lemma \ref{lem:rellich}, let's consider the case
where $\mat{A}(\vectort)$ has a repeated eigenvalue $\lambda_0$ with
multiplicity $l$, \ie., $\lambda_1(\vectort) = \lambda_2(\vectort) = \ldots =
\lambda_l(\vectort) = \lambda_0$. Assume further that the $\dim \times l$ matrix
$\mat{U}(\vectort)$ spans the space associated with the repeated eigenvalues
(\ie, $\mat{U}(\vectort)$ contains one particular set of eigenvectors associated
with the $l$ repeated eigenvalue). Then, the $l$ derivatives of the eigenvalues,
which coincide at $\lambda_0$ are the eigenvalues of the matrix
\begin{gather}
\mat{U}(\vectort)^\T \Jacob_{\vectort} \mat{A}(\vectort) \mat{U}(\vectort).
\end{gather}
\end{lem}

Using Lemma \ref{lem:lancaster}, we can write
\begin{align}
\left. \Jacob_{\vectort} {\lambda}_i(\vectort, \tau=\vectort^\star)
\right|_{\vectort = \vectort^\star} &= \mu_i\left( \mat{1}^\T_{\dim, l}  \left.
\Jacob_{\vectort} \pQ(\vectort, \tau=\vectort^\star) \right|_{\vectort = \vectort^\star} \mat{1}_{\dim,l}
\right) \label{eq:mui_first} \\
&= \mu_i\left( \mat{1}^\T_{\dim, l} \mat{V}(\vectort^\star)^{-\T} \left.
\Jacob_{\vectort} \Q(\rvec{x}, \Cov{x_g}, \vectort) \right|_{\vectort = \vectort^\star} \mat{V}(\vectort^\star)^{-1} \mat{1}_{\dim,l}
\right) \\
&\geq \mu_i \left( \left[ \mats{\Sigma}_{\rvec{x}}(\vectort^\star)
\mat{C}(\vectort^\star)
\mats{\Sigma}_{\rvec{x}}(\vectort^\star) - \mats{\Sigma}_G(\vectort^\star) \mat{C}(\vectort^\star) \mats{\Sigma}_G(\vectort^\star)
\right]_{1:l, 1:l} \right) \label{eq:before_last_repeated} \\
&= \mu_i \left( \left[ \mats{\Sigma}_{\rvec{x}}(\vectort^\star) \right]_{1:l,
1:l} \left[\mat{C}(\vectort^\star)\right]_{1:l, 1:l}
\left[\mats{\Sigma}_{\rvec{x}}(\vectort^\star)\right]_{1:l, 1:l} - \left[\mats{\Sigma}_G(\vectort^\star)\right]_{1:l, 1:l} \left[\mat{C}(\vectort^\star)\right]_{1:l, 1:l} \left[ \mats{\Sigma}_G(\vectort^\star) \right]_{1:l, 1:l} \right) \label{eq:last_repeated}
\end{align}
where $\mu_i(\mat{A})$ denotes the eigenvalue function of a generic matrix
$\mat{A}$. Observe that, thanks to the fact that $\pQ(\vectort^\star, \vectort^\star) =
\mats{\Sigma}_G(\vectort^\star) - \mats{\Sigma}_{\rvec{x}}(\vectort^\star)$ is a
diagonal matrix, in (\ref{eq:mui_first}) we have chosen
\begin{gather}
 \mat{U} = \mat{1}_{\dim, l} \triangleq \left(
\begin{array}{c}
 \mat{I}_l \\ \mat{0}_{\dim-l, l}
\end{array}
\right)
\end{gather}
with $\mat{I}_l$ being the $l\times l$ identity matrix and $\mat{0}_{\dim-l, l}$
being the $(\dim-l)\times l$ zero matrix. Moreover, in (\ref{eq:before_last_repeated}),
we have used the fact that $\mat{A} \succeq \mat{B}$ implies both that $\mat{C}^\T \mat{A}
\mat{C} \succeq \mat{C}^\T \mat{B} \mat{C}$ and that $\mu_i(\mat{A}) \geq
\mu_i(\mat{B})$ \cite[Cor.~7.7.4(c)]{Horn} and the lower bound on the
derivative of the matrix $\Q(\vectort)$ given in Lemma
\ref{lem:LemmaLowerBound}. We further used the definition:
 \begin{gather}
\mat{C}(\vectort^\star) = \mat{V}(\vectort^\star) \ChanDChan(\vectort^\star) \mat{V}(\vectort^\star)^\T.
\end{gather}
Observe that, since $\ChanDChan(\vectort^\star)$ is a positive semidefinite diagonal matrix (see Lemma \ref{lem:path}), we have $\mat{C}(\vectort^\star) \succeq \mat{0}$, which
further implies that $[\mat{C}(\vectort^\star)]_{ii} \geq 0$, for all $i$. Finally, the upper-left $l\times l$
sub-matrix of matrix $\mat{A}$ has been denoted by $[\mat{A}]_{1:l, 1:l}$, and the last transition in (\ref{eq:last_repeated}) is due to the fact that both $\mats{\Sigma}_{\rvec{x}}(\vectort^\star)$ and $\mats{\Sigma}_G(\vectort^\star)$ are diagonal matrices.

In order to proceed with the proof, we require the following lemma.
\begin{lem} \label{lem:amendment}
Let's consider a positive semidefinite matrix $\mat{A}$ and two diagonal
positive semidefinite matrices $\mat{D}_1$ and $\mat{D}_2$ such that $\mat{D}_1
\succeq \mat{D}_2 \succeq \mat{0}$. Then, we have that
\begin{eqnarray}
\mu_{\max} ( \mat{D}_1 \mat{A} \mat{D}_1 - \mat{D}_2 \mat{A} \mat{D}_2 ) \geq
0
\end{eqnarray}
where $\mu_{\max}$ denotes the maximum eigenvalue function.
\end{lem}
\begin{IEEEproof}
See Appendix \ref{app:amendment}.
\end{IEEEproof}

Now, using the fact that $\mat{C}(\vectort^\star)$ is positive semidefinite, and the
first condition in (\ref{eq:lambdaQt_critic_repeat}) that
${\lambda}_i(\vectort^\star, \vectort^\star) < 0$ for $i = 1,\ldots,l$, which
further implies that $\left[ \mats{\Sigma}_{\rvec{x}}(\vectort^\star)\right]_{1:l,
1:l} \succ \left[\mats{\Sigma}_G(\vectort^\star)\right]_{1:l, 1:l} \succeq \mat{0}$
we can use Lemma \ref{lem:amendment} to conclude that,
\begin{gather}
\mu_{\max} \left( \left[ \mats{\Sigma}_{\rvec{x}}(\vectort^\star) \right]_{1:l,
1:l} \left[\mat{C}(\vectort^\star)\right]_{1:l, 1:l}
\left[\mats{\Sigma}_{\rvec{x}}(\vectort^\star)\right]_{1:l, 1:l} -
\left[\mats{\Sigma}_G(\vectort^\star)\right]_{1:l, 1:l}
\left[\mat{C}(\vectort^\star)\right]_{1:l, 1:l} \left[ \mats{\Sigma}_G(\vectort^\star)
\right]_{1:l, 1:l} \right) \geq 0.
\end{gather}
Last result together with (\ref{eq:mui_first})-(\ref{eq:last_repeated}) implies
that there exists some $i \in [1, l]$ such that ${\lambda}_i(\vectort^\star,
\vectort^\star) < 0$ and $\left. \Jacob_{\vectort} {\lambda}_i(\vectort,
\tau=\vectort^\star) \right|_{\vectort = \vectort^\star} \geq 0$, which clearly
contradicts the conditions in (\ref{eq:lambdaQt_critic_repeat}).

Since the contradiction described above holds for any arbitrary values for $k$,
$s$, and $l$ (under the condition $l\leq s \leq k \leq n$), we have thus proved
that no nonnegative-to-negative zero crossing can occur for the eigenvalues of
$\pQ(\vectort, \tau)$ or, equivalently, we have proved that the eigenvalues
of $\pQ(\vectort, \tau)$ have at most a single negative-to-nonnegative zero
crossing of the horizontal axis.

\subsubsection{Single crossing point for the eigenvalues of $\Q( \vectort )$}
\label{proof:second_stage}
The relation between the sign of the eigenvalues of $\Q( \vectort )$
and those of $\pQ(\vectort, \tau)$ is stated in the following lemma.
\begin{lem} \label{lem:conserv_eig}
For all $\tau$ and as a function of $\vectort$, the number of positive, zero, and
negative eigenvalues of $\Q(\vectort )$ and $\pQ(\vectort,
\tau)$
coincide.
\end{lem}
\begin{proof}
The proof follows straightforwardly from the definition of
$\pQ(\vectort, \tau)$, given in equation (\ref{eq:qtilde}), and Sylvester's law of inertia for
congruent matrices \cite[p.~5]{bhatia:07}.
\end{proof}
In the first part of the proof we have shown that $\pQ( \vectort, \tau)$ has,
for each eigenvalue, at most, a single negative-to-nonnegative zero crossing.
From this and Lemma \ref{lem:conserv_eig}, we can conclude that the number of
negative eigenvalues of both functions cannot increase. Now, let's assume that
$\Q(\vectort )$ has an eigenvalue of multiplicity $s$ with
a nonnegative-to-negative zero crossing at $\vectort_0$, \ie,
$\mu_i(\Q(\vectort_0 ) )= 0$ and $\mu_i(\Q(\vectort )) < 0$ for $\vectort \in
(\vectort_0, \vectort_0 + \varepsilon)$, for some positive $\varepsilon$ and for
$i=1,\ldots,s$. In order to refrain from increasing the number of negative
eigenvalues, $s$ negative eigenvalues at $\vectort_0$ must become zero. However,
if we examine the number of eigenvalues at $\vectort_0 + \Delta$ for a
sufficiently small $\Delta$, the eigenvalues that were negative at $\vectort_0$
are still negative at $\vectort_0 + \Delta$, and the total number of negative
eigenvalues has increased. Thus, contradicting the possibility of a
nonnegative-to-negative zero crossing of the multiplicity $s$ eigenvalue of
$\Q(\vectort )$. This is valid for any arbitrary $\vectort_0$, thus concluding
our proof.
\end{IEEEproof}

The following corollary is a simple consequence from Theorem
\ref{thm:TheoremSingleCrossingEig}.
\begin{cor} \label{cor:CorollarySingleCrossingEig}
If for a given $\vectort'$ the function $\Q( \rvec{x}, \Cov{x_g}, \vectort' ) \succeq \mat{0}$ then
for all $\vectort \geq \vectort'$ the function $\Q( \rvec{x}, \Cov{x_g}, \vectort ) \succeq \mat{0}$.
\end{cor}

\subsection{The Conditioned Case} \label{ssec:generalConditioned}

The results of the previous section can be simply extended to the conditioned case. Given an extension of the lower bound on the derivative of $\Q$, the extension of all other results is trivial. Thus, we briefly give the extension to the lower bound with a full proof (given in Appendix \ref{app:LemmaLowerBoundConditioned}) and then for completeness restate the main result of this paper, for the conditioned case, without detailing the proof, which follows identically to the proof given above.



\begin{lem} \label{lem:LemmaLowerBoundConditioned}
The following lower bound holds:
\begin{eqnarray} \label{eq:LemmaLowerBoundConditioned}
\Jacob_{\vectort} \Q(\rvec{x}|\rvec{u}, \Cov{\rvec{x_g}}, \vectort) & \succeq & 2 \left( \MSE{x|u}(\vectort)
\ChanDChan(\vectort) \MSE{x|u}^\T(\vectort)- \EM_G(\vectort)
\ChanDChan(\vectort) \EM_G^\T(\vectort) \right)
\end{eqnarray}
where $\ChanDChan(\vectort)$ was defined in (\ref{eq:definitionB}), and assumed a positive semidefinite diagonal matrix for all $\vectort$ (see Lemma \ref{lem:path}).
\end{lem}
\begin{IEEEproof}
See Appendix \ref{app:LemmaLowerBoundConditioned}.
\end{IEEEproof}
Thus, the following theorem follows:
\begin{thm} \label{thm:TheoremSingleCrossingEig_u}
Each eigenvalue of $\Q( \rvec{x}|\rvec{u}, \Cov{x_g}, \vectort)$ has, \emph{at most}, a single negative-to-nonnegative zero crossing
of the horizontal axis.
\end{thm}
\begin{IEEEproof} The proof follows the same steps as those in the proof of Theorem \ref{thm:TheoremSingleCrossingEig}.
\end{IEEEproof}


\subsection{Properties of the Mutual Information}
\label{ssec:MutualInformationEig}

So far we have seen the ``single crossing point'' property of the matrix $\Q(\rvec{x}, \Cov{x_g},
\vectort)$, or more precisely, of its eigenvalues. As seen, this property also extends naturally to the conditioned case. In this section our goal is to relate this result to the mutual information between the input and the output of a parallel Gaussian channel. As expected, the advantage of this result is in the comparison between the mutual information assuming that the input to the channel has an arbitrary distribution and the mutual information assuming that it has a Gaussian distribution with an arbitrary covariance, $\Cov{x_g}$. Our goal is to make use of this result through the I-MMSE relationship, as given in equations (\ref{eq:lineIntegral})-(\ref{eq:lineIntegral2}) and (\ref{eq:lineIntegral_cond})-(\ref{eq:lineIntegral_cond2}). The results given in this section can be viewed as supporting theorem/lemmas, that make our ``single crossing point'' property applicable through the use of the I-MMSE relationship.

For clarity we will write the results in this section only for, the more general, conditioned case, from which one can easily derive the respective unconditioned theorems.

According to equation (\ref{eq:lineIntegral_cond2}) the difference between the mutual information assuming that the input to the channel has an arbitrary distribution and the mutual information assuming that it has a Gaussian distribution with an arbitrary covariance, $\Cov{x_g}$, is
\begin{align}
\Igen{\rvec{x_g}}{\rvec{y}(\vectort)} - \Icond{\rvec{x}}{\rvec{y}(\vectort)}{\rvec{u}} &= \int_{\tau=0}^{\vectort} \Tr \left(
\ChanDChan(\tau) ( \EM_G(\tau) - \MSE{\rvec{x}|\rvec{u}}(\tau) )
\right) \d \tau \nonumber \\
&= \int_{\tau=0}^{\vectort} \Tr \left(
\ChanDChan(\tau) \Q( \rvec{x}|\rvec{u}, \Cov{x_g}, \tau)
\right) \d \tau.
\end{align}
Thus, we are interested in the properties of
\begin{align}
\Tr \left(
\ChanDChan(\vectort) \Q( \rvec{x}|\rvec{u}, \Cov{x_g}, \vectort)
\right) = \sum_{i=1}^{\dim} \lambda_i( \ChanDChan(\vectort) \Q( \rvec{x}|\rvec{u}, \Cov{x_g}, \vectort) )
\end{align}
where we have used the fact that the trace of a matrix $\mat{A}$ is the sum of its eigenvalues \cite[Th.~1.2.12]{Horn}. The following theorem extends the ``single crossing point'' property of the eigenvalues of $\Q( \rvec{x}|\rvec{u}, \Cov{x_g}, \vectort)$ to the eigenvalues of $\ChanDChan(\vectort) \Q( \rvec{x}|\rvec{u}, \Cov{x_g}, \vectort)$.
\begin{thm} \label{thm:BQ}
Each eigenvalue of $\ChanDChan(\vectort) \Q( \rvec{x}|\rvec{u}, \Cov{x_g}, \vectort)$ has, \emph{at most}, a single negative-to-nonnegative zero crossing of the
horizontal axis. Moreover, the eigenvalues of $\ChanDChan(\vectort) \Q( \rvec{x}|\rvec{u}, \Cov{x_g}, \vectort)$  have the following property:
\begin{eqnarray} \label{eq:eigenvaluesBQ}
\opn{sign} \left\{ \lambda_i ( \ChanDChan(\vectort) \Q( \rvec{x}|\rvec{u}, \Cov{x_g}, \vectort))  \right\} \in \left\{0, \opn{sign} \left\{ \lambda_i ( \Q( \rvec{x}|\rvec{u}, \Cov{x_g}, \vectort) ) \right\}  \right\} \textrm{.}
\end{eqnarray}
\end{thm}
\begin{IEEEproof}
For a non-singular $\ChanDChan(\vectort)$ and due to similarity \cite[Cor.~1.3.4]{Horn} we can write the following,
\begin{eqnarray} \label{eq:lambda_i}
\lambda_i ( \ChanDChan(\vectort) \Q( \rvec{x}|\rvec{u}, \Cov{x_g}, \vectort) )  = \lambda_i( \ChanDChan^{\frac{1}{2}}(\vectort) \Q( \rvec{x}|\rvec{u}, \Cov{x_g}, \vectort) \ChanDChan^{\frac{1}{2}}(\vectort) ).
\end{eqnarray}
Recalling that $\ChanDChan(\vectort)$ is a positive semidefinite diagonal matrix, we have an eigenvalue of a congruent transformation. Thus, the proof follows similarly to the second part of the proof of Theorem \ref{thm:TheoremSingleCrossingEig} (given in Section \ref{proof:second_stage}), concluding the preservation of the signs of the eigenvalues of $\Q( \rvec{x}|\rvec{u}, \Cov{x_g}, \vectort)$ in $\ChanDChan(\vectort) \Q( \rvec{x}|\rvec{u}, \Cov{x_g}, \vectort)$ and, as a result, concluding that all eigenvalues have, \emph{at most} a single, negative-to-nonnegative zero crossing of the horizontal axis.

If $\ChanDChan(\vectort)$ is singular, we can assume without loss of generality that the $i^{th}$ diagonal element is zero. Due to that, the $i^{th}$ row of $\ChanDChan(\vectort) \Q( \rvec{x}|\rvec{u}, \Cov{x_g}, \vectort)$ is all zeros, that is, one of the eigenvalues of $\ChanDChan(\vectort) \Q( \rvec{x}|\rvec{u}, \Cov{x_g}, \vectort)$ is zero (and its sign is also zero). The rest of the eigenvalues can be calculated from the reduced problem, the matrix $\ChanDChan(\vectort) \Q( \rvec{x}|\rvec{u}, \Cov{x_g}, \vectort)$ without the $i^{th}$ row and column. Recalling that $\ChanDChan(\vectort)$ is a diagonal matrix, this is simply the product of $\ChanDChan(\vectort)$ and $\Q( \rvec{x}|\rvec{u}, \Cov{x_g}, \vectort)$ both without the $i^{th}$ row and column. This procedure can be repeated as long as the reduced $\ChanDChan(\vectort)$ matrix is singular. When the reduced matrix is non-singular, we again follow the proof of Theorem \ref{thm:TheoremSingleCrossingEig}.

Thus, we have shown that the eigenvalues preserve the sign of the eigenvalues of $\Q( \rvec{x}|\rvec{u}, \Cov{x_g}, \vectort)$ with the additional possibility of falling to zero when $\ChanDChan(\vectort)$ becomes singular.
\end{IEEEproof}

The next two lemmas provide the link between the above results, regarding the behavior of the eigenvalues of the matrix $\Q( \rvec{x}|\rvec{u}, \Cov{x_g}, \vectort)$ and the matrix $\ChanDChan(\vectort) \Q( \rvec{x}|\rvec{u}, \Cov{x_g}, \vectort)$, and the mutual information. Thus, they facilitate the usage of these results on information theory problems, as will be shown in the sequel. More particularly, so far we discussed the behavior of each and every eigenvalue of the matrix $\Q( \rvec{x}|\rvec{u}, \Cov{x_g}, \vectort)$ and the matrix $\ChanDChan(\vectort) \Q( \rvec{x}|\rvec{u}, \Cov{x_g}, \vectort)$, which holds true for any proper choice of $\Cov{x_g}$ with no regards to the random vector $\rvec{x}$. The next two lemmas identify the existence of specific Gaussian inputs which have unique properties with respect to the given random vector $\rvec{x}$.


\begin{lem} \label{lem:GaussianExistenceConditioned}
Assume $\rvec{x} \in \R^{\dim}$ is an arbitrary distributed random vector.
For any $\vectortEq \in [0, \infty)$, there exists a Gaussian input covariance matrix $\Cov{x_g}$
such that the following hold
\begin{enumerate}
\item $\Cov{x_g} \preceq \Cov{x}$ \label{req:covariance}
\item $\Icond{\rvec{x}}{\rvec{y}(\vectortEq)}{\rvec{u}} = \Igen{\rvec{x_g}}{\rvec{y_g}(\vectortEq)}$ \label{req:MI}
\item $\Q( \rvec{x}|\rvec{u}, \Cov{x_g}, \vectortEq) \succeq \mat{0}$ \label{req:Q}
\end{enumerate}
\end{lem}
\begin{IEEEproof} See Appendix \ref{app:GaussianExistenceConditionedProof}. \end{IEEEproof}

Note that the above claim can be extended to a general non-singular
$\Chan(\vectortEq)$, that is, \emph{not necessarily diagonal}, by defining
$\widetilde{\rvec{x}} \equiv \Chan(\vectortEq) \rvec{x}$.
Due to the non-singularity of $\Chan(\vectortEq)$, the mutual information is unchanged, \ie, $\Icond{\widetilde{\rvec{x}}}{\rvec{y}(\vectortEq)}{\rvec{u}} = \Icond{\rvec{x}}{\rvec{y}(\vectortEq)}{\rvec{u}}$.
Requirements \ref{req:covariance} and \ref{req:Q} are preserved under any congruent
transformation, specifically under the transformation $\Chan^{-1}(\vectortEq)$.

The next lemma is an extension of Lemma \ref{lem:GaussianExistenceConditioned} that will prove useful in the sequel.
\begin{lem} \label{lem:GaussianExistenceExtensionConditioned}
Assume that for a given input distribution on the pair $(\rvec{u}, \rvec{x})$ there exists a Gaussian random vector, $\rvec{x_g}^{\ss{ub}}$, with covariance $\Cov{x_g}^{\ss{ub}}$ such that for some $\vectortEq \in [0,\infty)$ we have that,
\begin{enumerate}
\item $\Icond{\rvec{x}}{\rvec{y}(\vectortEq)}{\rvec{u}} \leq \Igen{\rvec{x_g}^{\ss{ub}}}{\rvec{y_g}^{\ss{ub}}(\vectortEq)}$
\item $\Q( \rvec{x}|\rvec{u}, \Cov{x_g}^{\ss{ub}}, \vectortEq) \succeq \mat{0}$
\end{enumerate}
Thus, there exists a Gaussian random vector, $\rvec{x_g}$, with covariance $\Cov{x_g}$
such that the following holds:
\begin{enumerate}
\item $\Cov{x_g} \preceq \Cov{x_g}^{\ss{ub}}$
\item $\Icond{\rvec{x}}{\rvec{y}(\vectortEq)}{\rvec{u}} = \Igen{\rvec{x_g}}{\rvec{y_g}(\vectortEq)}$
\item $\Q( \rvec{x}|\rvec{u}, \Cov{x_g}, \vectortEq) \succeq \mat{0}$
\end{enumerate}
\end{lem}
\begin{IEEEproof}
The proof follows the proof of Lemma \ref{lem:GaussianExistenceConditioned}, where instead of using $\MSElin{x}(\vectortEq)$ (\ref{eq:MMSE_L}) as a trivial upper bound we use:
\begin{eqnarray} \label{eq:GaussianExistenceExtension_proof}
\EM_{G}^{\ss{ub}}(\vectortEq) = \Identity - (\Cov{x_g}^{\ss{ub}} + \Identity)^{-1}
\end{eqnarray}
and the assumptions stated above.
\end{IEEEproof}

\subsection{Connections to Fisher Information} \label{ssec:fisher}

In addition to the MMSE matrix, another important quantity in estimation theory is the Fisher information matrix \cite{MMSE}. Its connection to information theory has been established in the late 1950's and has been attributed to de Bruijn \cite{STAM}. The de Bruijn identity relates the derivative of the differential entropy to the Fisher information matrix defined as\footnote{For any differentiable function $f: \R^{\dim} \to \R$, its gradient at any $\rvecr{y}$ is a column vector
$\Gradient f(\rvecr{y}) = \left[ \Jacob_{y_1} f(\rvecr{y} ), \ldots,\Jacob_{y_{\dim}} f(\rvecr{y} )   \right]^\T$.}:
\begin{align}
\Fisher( \rvec{y} ) = \Esp[1]{ \left[ \Gradient \log \pdf{\rvec{y}}{\rvec{y}} \right] \left[ \Gradient \log \pdf{\rvec{y}}{\rvec{y}} \right]^\T }
\end{align}
where the expectation is over $\rvec{y}$. Note that this is a special form of the Fisher Information matrix (with respect to a translation parameter) which does not involve an explicit parameter as in its most general definition \cite{MMSE}. In \cite{IMMSE} the authors have shown that the de Bruijn identity is equivalent to the I-MMSE relationship. Using this connection, the de Bruijn identity has been extended to a multivariate version in \cite[Th.~4]{Palomar}. For our purposes we will use the following notation:
\begin{align}
\Fisher_{\rvec{x}}( \Chan ) = \Fisher( \Chan \rvec{x} + \rvec{n} )
\end{align}
when we have some arbitrary input distribution on the random vector $\rvec{x}$. For the case of a Gaussian distribution on $\rvec{x}$ with covariance matrix $\Cov{x_g}$ we will write $\Fisher_{G}( \Cov{x_g}, \Chan ) $. We further note that, as in the case of the MMSE matrix, whenever the channel coefficients depend on other parameters, $\Chan = \Chan( \phi)$, we will write $\Fisher_{\rvec{x}}( \phi)$. We can now extend the idea of the the matrix $\Q$ to the Fisher Information, using the following definition:
\begin{align}
\QFisher ( \rvec{x}, \Cov{x_g}, \phi) = \Fisher_{\rvec{x}} ( \phi ) - \Fisher_G ( \Cov{x_g}, \phi) .
\end{align}
As in the case of the matrix $\Q$, the matrix $\QFisher$ has some distinct properties. Using the relationship between the two matrices we can derive these properties directly from the results of the previous sections.
We first require the following lemma, given by Palomar and Verd$\acute{\textrm{u}}$ in \cite{Palomar}.
\begin{lem}{\cite[App.~E]{Palomar}} \label{lem:FisherMMSE}
Assuming the Gaussian additive noise channel (\ref{eq:generalModel}), the following connection between the Fisher Information matrix and the MMSE matrix holds:
\begin{gather}
 \J{y} = \mat{I}_\dim - \Chan \MSE{x} \Chan^\T
\end{gather}
\end{lem}
\begin{IEEEproof} The result follows directly from equation (106) in \cite{Palomar} by setting $\tilde{\mat{\Sigma}}_n$ equal to the identity matrix and recalling that the MMSE matrix in (106) is the MMSE matrix of $\rvec{z} = \Chan \rvec{x}$, from which it follows that $\MSE{z} = \Chan \MSE{x} \Chan^\T$.
\end{IEEEproof}
We can now state the main result of this section.
\begin{thm} \label{thm:FisherInfroamtionTheorem}
The matrix $\QFisher( \rvec{x}, \Cov{x_g}, \vectort)$ is related to the matrix $\Q(\rvec{x}, \Cov{x_g}, \vectort)$ as follows:
\begin{align} \label{eq:FisherInfroamtionTheorem}
\QFisher( \rvec{x}, \Cov{x_g}, \vectort) = \Chan(\vectort) \Q(\rvec{x}, \Cov{x_g}, \vectort) \Chan(\vectort)^\T .
\end{align}
Moreover, the properties given in Sections \ref{sec:independentGaussian} and \ref{sec:generalGaussian} for the matrix $\Q(\rvec{x}, \Cov{x_g}, \vectort)$, transfer to the matrix $\QFisher( \rvec{x}, \Cov{x_g}, \vectort)$.
\end{thm}
\begin{IEEEproof}
Equation (\ref{eq:FisherInfroamtionTheorem}) is obtained through the use of Lemma \ref{lem:FisherMMSE}. The properties given in Section \ref{sec:independentGaussian} regarding the matrix $\Q(\rvec{x}, \Cov{x_g}, \vectort)$ transfer to the matrix $\QFisher( \rvec{x}, \Cov{x_g}, \vectort)$, due to the fact that $\Chan(\vectort)$ is a diagonal positive semidefinite matrix for all $\vectort$. The properties given in Section \ref{sec:generalGaussian} regarding the matrix $\Q(\rvec{x}, \Cov{x_g}, \vectort)$ transfer to the matrix $\QFisher( \rvec{x}, \Cov{x_g}, \vectort)$, since it is simply a congruent transformation of $\Q(\rvec{x}, \Cov{x_g}, \vectort)$ (this was explained in detail in part two of the proof of Theorem \ref{thm:TheoremSingleCrossingEig}).
\end{IEEEproof}

\subsection{Application: The Degraded Parallel Gaussian BC Capacity
Region under Covariance Constraint} \label{ssec:BCcovarianceConstraint}

In this section we show that the result of Section \ref{ssec:MutualInformationEig} can be used to provide a converse proof for the \emph{degraded} parallel Gaussian BC capacity region under a covariance constraint. We consider the following model:
\begin{align} \label{eq:app:BCcovarianceModel}
\rvec{y}_1[\timeIn] &= \Chan_{1} \rvec{x}[\timeIn] + \rvec{n}_1[\timeIn] \nonumber \\
\rvec{y}_2[\timeIn] &= \Chan_{2} \rvec{x}[\timeIn] + \rvec{n}_2[\timeIn]
\end{align}
where $\rvec{n}_1[\timeIn]$ and $\rvec{n}_2[\timeIn]$ are standard additive Gaussian noise vectors independent for different time indices $\timeIn$, and $\Chan_{1}$ and $\Chan_{2}$ are diagonal positive semidefinite matrices such that $\Chan_{1} \preceq \Chan_{2}$. $\rvec{x} \in \R^{\dim}$ is the random input vector, and it is assumed independent for different time indices $\timeIn$.

We consider a covariance constraint:
\begin{align} \label{eq:app:BCcovarianceConstraint}
\Cov{x} \preceq \mat{S}
\end{align}
where $\mat{S}$ is some positive definite matrix.

Since we have a \emph{degraded} BC, we can use the single-letter expression as given in (\ref{eq:app:perAntenna_singleLetterExpr}). As in Section \ref{ssec:applicationInd}, we will follow the proof given for the scalar Gaussian BC in \cite{PROP,PROP_full}. Using Lemma \ref{lem:path} we can construct a path such that:
\begin{align} \label{eq:app:covariancePath}
\Chan(\vectort_2) & = \Chan_2 \nonumber \\
\Chan(\vectort_1) & = \Chan_1 \nonumber \\
\Chan(0) & = \mat{0}
\end{align}
where $0 \leq \vectort_1 \leq \vectort_2$ and $\Chan(\vectort)$ is diagonal for all $\vectort \in [0, \vectort_2]$.

Now, assume a pair $(\rvec{u}, \rvec{x})$ with covariance $\Cov{x}$ for $\rvec{x}$. According to Lemma \ref{lem:GaussianExistenceConditioned}, there exists a Gaussian random vector with covariance $\Cov{x_g}$ such that the following properties hold:
\begin{enumerate}
\item $\Cov{x_g} \preceq \Cov{x}$. \label{application:covariance:item1}

\item $\Icond{\rvec{x}}{\rvec{y}(\vectort_1)}{\rvec{u}} = \Igen{\rvec{x_g}}{\rvec{y_g}(\vectort_1)}$. \label{application:covariance:item2}

\item $\Q( \rvec{x|u}, \Cov{x_g}, \vectort) \succeq \mat{0}$ for all $\vectort \geq \vectort_1$. \label{application:covariance:item3}

\end{enumerate}
Using the I-MMSE relationship (\ref{eq:lineIntegral_cond2}) we can write,
\begin{align} \label{eq:app:covariance_IMMSE_1}
\Igen{\rvec{x_g}}{\Chan(\vectort) \rvec{x_g} + \rvec{n}} - \Icond{\rvec{x}}{\Chan(\vectort) \rvec{x} + \rvec{n}}{\rvec{u}} &=
\int_{\tau=0}^{\vectort} \Tr \left(
\ChanDChan(\tau) \Q(\rvec{x}|\rvec{u},\Cov{x_g},\tau) \right) \d \tau \\
\label{eq:app:covariance_IMMSE_2}
& = \int_{\tau=0}^{\vectort} \sum_{i=1}^{\dim} \lambda_i \left(
\ChanDChan(\tau) \Q(\rvec{x}|\rvec{u}, \Cov{x_g} ,\tau) \right) \d \tau .
\end{align}
Using the above properties on (\ref{eq:app:covariance_IMMSE_2}) we have that for any $\vectort' > \vectort_1$,
\begin{align} \label{eq:app:covariance_conclusingFromThm}
\Igen{\rvec{x_g}}{\Chan(\vectort') \rvec{x_g} + \rvec{n}} - \Icond{\rvec{x}}{\Chan(\vectort') \rvec{x} + \rvec{n}}{\rvec{u}}
&= \int_{\tau=0}^{\vectort_1} \Tr \left(
\ChanDChan(\tau) \Q(\rvec{x}|\rvec{u},\Cov{x_g} ,\tau) \right)
\d \tau \nonumber \\
& + \int_{\tau=\vectort_1}^{\vectort'} \Tr \left(
\ChanDChan(\tau) \Q(\rvec{x}|\rvec{u},\Cov{x_g} ,\tau) \right)
\d \tau \\ \label{eq:app:covariance_conclusingFromThm2}
&= 0 + \int_{\tau=\vectort_1}^{\vectort'} \sum_{i=1}^{\dim} \lambda_i \left(
\ChanDChan(\tau) \Q(\rvec{x}|\rvec{u},\Cov{x_g},\tau) \right)
\d \tau \geq 0
\end{align}
where (\ref{eq:app:covariance_conclusingFromThm2}) follows from property \ref{application:covariance:item2}, and
the inequality follows from property \ref{application:covariance:item3} and Theorem \ref{thm:BQ}.

Thus, we have shown the existence of a Gaussian random vector, $\rvec{x}_G$, with covariance matrix $\Cov{x_g}$, with the following properties:
\begin{align}
\Icond{\rvec{x}}{\rvec{y}(\vectort_1)}{\rvec{u}} &= \Igen{\rvec{x_g}}{\rvec{y_g}(\vectort_1)} \nonumber \\
\Icond{\rvec{x}}{\rvec{y}(\vectort_2)}{\rvec{u}} &\leq \Igen{\rvec{x_g}}{\rvec{y_g}(\vectort_2)} \nonumber \\
\Cov{x_g} & \preceq \Cov{x}
\end{align}
Using these properties on the single-letter expression (\ref{eq:app:perAntenna_singleLetterExpr}) we obtain the following outer bound,
\begin{align} \label{eq:app:covariance_final}
\rate_1 & \leq \Igen{\rvec{u}}{\rvec{y}_1} = \Igen{\rvec{x}}{\rvec{y}_1} - \Icond{\rvec{x}}{\rvec{y}_1}{\rvec{u}} \nonumber \\
& \leq \frac{1}{2} \log | \Identity + \Chan_1 \mat{S} \Chan_1^\T |
- \frac{1}{2} \log | \Identity + \Chan_1 \Cov{x_g} \Chan_1^\T |
= \frac{1}{2} \log  \frac{| \Identity + \Chan_1 \mat{S} \Chan_1^\T |}{| \Identity + \Chan_1 \Cov{x_g} \Chan_1^\T |} \\
\rate_2 & \leq \Icond{\rvec{x}}{\rvec{y}_2}{\rvec{u}} \leq
\frac{1}{2} \log | \Identity + \Chan_2 \Cov{x_g} \Chan_2^\T |
\end{align}
This outer bound is tight and the
achievability is well-known using superposition coding. 
This approach can be extended to the M-user scenario as shown in Appendix \ref{app:converseProofBCCovarianceMUsers}.

\subsection{Application: The Compound Degraded Parallel Gaussian BC
Capacity Region under Covariance Constraint}
\label{ssec:compBCcovarianceConstraint}

In this section we show that the results of Section \ref{ssec:singleCrossingPointEig} can also be used to provide a converse proof for the compound \emph{degraded} parallel Gaussian BC capacity region under a covariance constraint. We consider the following model,
\begin{align} \label{eq:modelCompundMIMO}
\rvec{y}_{i_j}^j[\timeIn] = \Chan_{i_j}^j \rvec{x}[\timeIn] + \rvec{n}_{i_j}^j[\timeIn], \quad j = 1,\ldots,M, \quad i_j = 1,\ldots,K_j
\end{align}
where $\rvec{n}_{i_j}^j$, $j = 1,..,M, i_j = 1,...,K_j$ are standard additive Gaussian noise vectors independent for different time indices $\timeIn$, and
$\Chan_{i_j}^j$, $j = 1,...,M, i_j = 1,...,K_j$ are diagonal positive definite matrices such that:
\begin{eqnarray} \label{eq:app_compound:degradedCompound_1}
\Chan_{i_j}^j \preceq \Chan_{i_{(j+1)}}^{j+1} \quad \forall j=1,\ldots,M ,i_j \in \{1, \ldots, K_j\} ,i_{j+1} \in \{1, \ldots, K_{j+1}\}.
\end{eqnarray}
Since these matrices are diagonal, there exist matrices $\Chan_{(j+1)j}^\star$ for $j=1,\ldots,M-1$ such that
\begin{eqnarray} \label{eq:app_compound:degradedCompound_2}
\Chan_{i_j}^j \preceq \Chan_{(j+1)j}^\star \preceq \Chan_{i_{(j+1)}}^{j+1} \quad \forall j=1,\ldots,M-1 ,i_j \in \{1, \ldots, K_j\} ,i_{j+1} \in \{1, \ldots, K_{j+1}\}.
\end{eqnarray}
Note that the equivalence between conditions (\ref{eq:app_compound:degradedCompound_1}) and (\ref{eq:app_compound:degradedCompound_2}) is not true in general (for non-diagonal matrices), as explained in \cite{HananCompound}.
$\rvec{x} \in \R^{\dim}$ is the random input vector, and it is assumed independent for different time indices $\timeIn$. We consider a covariance constraint:
%
\begin{eqnarray} \label{eq:app_compound:constraint}
\Cov{x} \preceq \mat{S}
\end{eqnarray}
where $\mat{S}$ is some positive definite matrix.

Before proceeding, we provide the following single-letter expression for the capacity region of this $M$ user memoryless channel. This is a simple extension of \cite[Lem.~4]{HananCompound}.

\begin{lem} \label{lem:LemmaCompoundMemoryless}
Consider a memoryless compound BC with input $\rvec{x}$, $M$ outputs $\rvec{y}_{i_j}^j$, $j=1,\ldots,M, i_j = 1, \ldots, K_j$, and auxiliary random outputs $\rvec{y}_{(j+1)j}^\star$ with $j \in \{1,\ldots,M-1\}$. All outputs are defined by their conditional probability functions: $\pdffun_{\rvec{y}_{i_j}^j|\rvec{x}}$ and $\pdffun_{\rvec{y}_{(j+1)j}^\star|\rvec{x}}$. Furthermore, assume that these outputs are stochastically \emph{degraded} such that there exists some distribution such that $\rvec{x} - \rvec{y}_{i_M}^M - \rvec{y}_{M(M-1)}^\star - \rvec{y}_{i_{M-1}}^{M-1} - \rvec{y}_{(M-1)(M-2)}^\star -\ldots - \rvec{y}_{i_{2}}^{2} - \rvec{y}_{21}^\star - \rvec{y}_{i_1}^1$ form a Markov chain for every choice of $i_1,i_2,\ldots,i_M$. The capacity region of this channel is given by the union of the rate tuples satisfying
\begin{eqnarray} \label{eq:capacityRegionMemoryless}
\rate_j \leq \min_{i_j = 1,\ldots,K_j} \Icond{\rvec{v}_{j}}{\rvec{y}_{i_j}^j}{\rvec{v}_{j-1}}
\end{eqnarray}
where $\rvec{v}_0 \equiv \emptyset$, $\rvec{v}_M \equiv \rvec{x}$ and the union is over all probability distributions satisfying
\begin{eqnarray} \label{eq:capacityRegionMemoryless_MC}
\rvec{v}_0 - \rvec{v}_1 -\ldots - \rvec{v}_{M-1} - \rvec{v}_M - \rvec{x} - \rvec{y}_{i_M}^M - \rvec{y}_{M(M-1)}^\star - \rvec{y}_{i_{M-1}}^{M-1} - \rvec{y}_{(M-1)(M-2)}^\star - \ldots \rvec{y}_{i_{2}}^{2} -\rvec{y}_{21}^\star - \rvec{y}_{i_1}^1 \textrm{.}
\end{eqnarray}
\end{lem}
\begin{proof}
See Appendix \ref{appendix:CompoundMemoryless}.
\end{proof}

Using Lemma \ref{lem:LemmaCompoundMemoryless} we prove the following theorem,

\begin{thm} \label{thm:Compound}
The capacity region of the compound \emph{degraded} parallel Gaussian BC (\ref{eq:modelCompundMIMO}), is given by the following expression:
\begin{align} \label{eq:capacityRegion}
\rate_M & \leq \min_{i_M=1,\ldots,K_M} \frac{1}{2} \log \left| \Chan_{i_M}^M {\Cov{g}}_M \left( \Chan_{i_M}^{M} \right)^\T + \Identity \right| \nonumber \\
\rate_j & \leq \min_{i_j=1,\ldots,K_j} \frac{1}{2} \log \frac{\left| \Chan_{i_j}^j \sum_{l=j}^M {\Cov{g}}_l \left( \Chan_{i_j}^{j}\right)^\T + \Identity \right|}{\left| \Chan_{i_j}^j \sum_{l=j+1}^{M} {\Cov{g}}_l \left( \Chan_{i_j}^{j}\right)^\T + \Identity \right|}, \quad \forall j = 1,\ldots,M-1
\end{align}
where ${\Cov{g}}_j$ are some positive semidefinite matrices such that $\mat{0} \preceq \sum_{l=1}^{M} {\Cov{g}}_l \preceq \mat{S}$.
\end{thm}
\begin{IEEEproof}
According to Lemma \ref{lem:path} (and the remark after this lemma) for any set of $\{ i_1, i_2,\ldots,i_M\}$ where $i_j \in K_j$ we can construct a diagonal path such that
\begin{align} \label{eq:app:CompoundConverseProofPath}
\Chan(\vectort_{i_j}) & = \Chan_{i_j}, \quad j=1,\ldots,M \nonumber \\
\Chan(\vectort_{(j+1)j}) & = \Chan_{(j+1)j}^\star, \quad j=1,\ldots,M-1 \nonumber \\
\Chan(\vectort=0) & = \mat{0}
\end{align}
with $0 \leq \vectort_{i_1} \leq \vectort_{21} \leq \vectort_{i_2} \leq \ldots \leq \vectort_{i_j} \leq \vectort_{(j+1)j} \leq \vectort_{i_{j+1}} \leq \ldots \leq \vectort_{i_M}$.
Now, let's examine a tuple of rates on the boundary of the capacity region: $(\rate_1^{opt}, \rate_2^{opt},\ldots,\rate_M^{opt})$. Assume that this tuple has been attained by the joint distribution $\Parg{V_{1}, \ldots, V_{M-1}, \rvec{x} }$ on the tuple with covariance $\Cov{x} \preceq \mat{S}$ as required by the constraint (\ref{eq:app_compound:constraint}).

We begin by looking at the following partial Markov chain:
\begin{eqnarray} \label{eq:capacityRegionMemoryless_MC_partial}
\rvec{v}_0 - \rvec{v}_1 -\ldots - \rvec{v}_{M-1} - \rvec{v}_M - \rvec{x} - \rvec{y}_{M(M-1)}^\star - \rvec{y}_{(M-1)(M-2)}^\star - \ldots -\rvec{y}_{21}^\star.
\end{eqnarray}
Now, assuming that $\rvec{y}_{(j+1)j}^\star$ are the outputs, we can use Lemma \ref{lem:appendix:covarianceMUserInduction} which states that
there exist $M$ Gaussian inputs $\rvec{x}_{G_j}$, with covariance matrices ${\Cov{x_g}}_j$ such that,
\begin{align} \label{eq:UsingMUserCovarianceResult}
\int_0^{\vectort_{(j+1)j}} \Tr \left(
\ChanDChan(\tau) \Q(\rvec{x}|\rvec{v}_j,{\Cov{x_g}}_j,\tau) \right) \d \tau & = 0  \\ \label{eq:UsingMUserCovarianceResult2}
\int_0^{\vectort_{(j+2)j}} \Tr \left(
\ChanDChan(\tau) \Q(\rvec{x}|\rvec{v}_{j},{\Cov{x_g}}_j,\tau) \right) \d \tau & \geq 0, \quad \quad \forall j=1,\ldots,M-2
\end{align}
and such that $\mat{0} \preceq {\Cov{x_g}}_j \preceq {\Cov{x_g}}_{j-1}$, for $j=2,\ldots,M-1$ and $\mats{0} \preceq  {\Cov{x_g}}_{1} \preceq \mats{S}$. Furthermore,
\begin{align}
\Q( \rvec{x}|\rvec{v}_{j},{\Cov{x_g}}_j,\vectort_{(j+1)j}) \succeq \mat{0}
\end{align}
for all $j = 1,\ldots, M-1$.

Using this result, and according to Corollary \ref{cor:CorollarySingleCrossingEig} we know that $\Q( \rvec{x}|\rvec{v}_{j},{\Cov{x_g}}_j,\vectort) \succeq \mat{0}$ for all $\vectort \geq \vectort_{(j+1)j}$. This holds for any diagonal path, such that $\Chan( \vectort_{(j+1)j} ) = \Chan^\star_{(j+1)j}$. Now, using Theorem \ref{thm:BQ} and (\ref{eq:UsingMUserCovarianceResult}) we can conclude that,
\begin{align} \label{eq:UsingMUserCovarianceResult3}
\int_0^{\vectort} \Tr \left(
\ChanDChan(\tau) \Q(\rvec{x}|\rvec{v}_j,{\Cov{x_g}}_j,\tau) \right) \d \tau & \leq 0, \quad \forall \vectort \leq \vectort_{(j+1)j} \nonumber \\
\int_0^{\vectort} \Tr \left(
\ChanDChan(\tau) \Q(\rvec{x}|\rvec{v}_j,{\Cov{x_g}}_j,\tau) \right) \d \tau & \geq 0, \quad \forall \vectort \geq \vectort_{(j+1)j}.
\end{align}
Due to the Markov chain:
\begin{eqnarray} \label{eq:usingMC}
\rvec{v}_j - \rvec{v}_{j+1} - \rvec{x} - \rvec{y}_{i_{j+1}}^{j+1} - \rvec{y}_{(j+1)j}^\star - \rvec{y}_{i_{j}}^{j}
\end{eqnarray}
(\ref{eq:UsingMUserCovarianceResult3}) is particularly valid for,
\begin{align} \label{eq:UsingMUserCovarianceResultParticularly}
\int_0^{\vectort_{i_j}} \Tr \left(
\ChanDChan(\tau) \Q(\rvec{x}|\rvec{v}_j,{\Cov{x_g}}_j,\tau) \right) \d \tau & \leq 0, \quad \forall i_{j} \in K_{j} \nonumber \\
\int_0^{\vectort_{i_{j+1}}} \Tr \left(
\ChanDChan(\tau) \Q(\rvec{x}|\rvec{v}_j,{\Cov{x_g}}_j,\tau) \right) \d \tau & \geq 0, \quad \forall i_{j+1} \in K_{j+1}
\end{align}
for any $j=1,\ldots,M-1$. Equations (\ref{eq:capacityRegionMemoryless}) can be written explicitly, as follows:
\begin{align} \label{eq:explicitlyCompound}
\rate_M & \leq \min_{i_M = 1,\ldots,K_M}  \Icond{\rvec{x}}{ \rvec{y}_{i_M}^M }{ \rvec{v}_{M-1}} \nonumber \\
\rate_{M-1} & \leq \min_{i_{M-1} = 1,\ldots,K_{M-1}}  \Icond{\rvec{x}}{\rvec{y}_{i_{M-1}}^{M-1} }{\rvec{v}_{M-2}} - \Icond{\rvec{x}}{ \rvec{y}_{i_{M-1}}^{M-1} }{ \rvec{v}_{M-1}} \nonumber \\
\rate_{M-2} & \leq \min_{i_{M-2} = 1,\ldots,K_{M-2}}  \Icond{\rvec{x}}{\rvec{y}_{i_{M-2}}^{M-2} }{\rvec{v}_{M-3}} - \Icond{\rvec{x}}{ \rvec{y}_{i_{M-2}}^{M-2} }{ \rvec{v}_{M-2}} \nonumber \\
\vdots \nonumber \\
\rate_{2} & \leq \min_{i_2 = 1,\ldots,K_2}  \Icond{\rvec{x}}{ \rvec{y}_{i_{2}}^{2} }{ \rvec{v}_{1}} - \Icond{\rvec{x}}{\rvec{y}_{i_{2}}^{2} }{ \rvec{v}_{2}} \nonumber \\
\rate_1 & \leq \min_{i_1 = 1,\ldots,K_1}  \Icond{\rvec{x}}{ \rvec{y}_{i_1}^1 }{ \rvec{v}_0 \equiv \emptyset} - \Icond{\rvec{x}}{ \rvec{y}_{i_1}^1 }{ \rvec{v}_{1}}.
\end{align}
Using (\ref{eq:UsingMUserCovarianceResultParticularly}) and the trivial bound on $\Igen{\rvec{x}}{\rvec{y}_{i_1}^1}$ we can upper bound these expressions as follows:
\begin{align} \label{eq:upperBoundCompound}
\rate_M & \leq \min_{i_M = 1,\ldots,K_M}  \Igen{\rvec{x}_{G_{M-1}}}{ \Chan_{i_M}^M \rvec{x}_{G_{M-1}} + \rvec{n}} \nonumber \\
\rate_{M-1} & \leq  \min_{i_{M-1} = 1,\ldots,K_{M-1}}  \Igen{\rvec{x}_{G_{M-2}}}{ \Chan_{i_{M-1}}^{M-1} \rvec{x}_{G_{M-2}} + \rvec{n}} - \Igen{\rvec{x}_{G_{M-1}}}{ \Chan_{i_{M-1}}^{M-1} \rvec{x}_{G_{M-1}} + \rvec{n}} \nonumber \\
\rate_{M-2} & \leq \min_{i_{M-2} = 1,\ldots,K_{M-2}}  \Igen{\rvec{x}_{G_{M-3}}}{ \Chan_{i_{M-2}}^{M-2} \rvec{x}_{G_{M-3}} + \rvec{n}} - \Igen{\rvec{x}_{G_{M-2}}}{ \Chan_{i_{M-2}}^{M-2} \rvec{x}_{G_{M-2}} + \rvec{n}} \nonumber \\
\vdots \nonumber \\
\rate_{2} & \leq \min_{i_2 = 1,\ldots,K_2}  \Igen{\rvec{x}_{G_{1}}}{ \Chan_{i_{2}}^{2} \rvec{x}_{G_{1}}+ \rvec{n}} - \Igen{\rvec{x}_{G_{2}}}{ \Chan_{i_{2}}^{2} \rvec{x}_{G_{2}}+ \rvec{n}} \nonumber \\
\rate_1 & \leq \min_{i_1 = 1,\ldots,K_1}  \frac{1}{2} \log \left| \Identity  + \Chan_{i_1}^1 \mats{S} \left( \Chan_{i_1}^{1}\right)^\T \right| - \Igen{\rvec{x}_{G_{1}}}{ \Chan_{i_{1}}^{1} \rvec{x}_{G_{1}} + \rvec{n}}.
\end{align}
Defining,
\begin{align} \label{app:compound:applicationCovariance_definingC}
{\Cov{g}}_1 & = \mats{S} - {\Cov{x_g}}_1 \nonumber \\
{\Cov{g}}_j & = {\Cov{x_g}}_{j-1} - {\Cov{x_g}}_{j}, \quad \forall j=2,\ldots,M-1 \nonumber \\
{\Cov{g}}_M & = {\Cov{x_g}}_{M-1}
\end{align}
(\ref{eq:upperBoundCompound}) becomes the following set of upper bound,
\begin{align} \label{eq:tpperBoundCompound}
\rate_M & \leq \min_{i_M=1,\ldots,K_M} \frac{1}{2} \log \left| \Identity + \Chan_{i_M}^M {\Cov{g}}_M \left( \Chan_{i_M}^{M} \right)^T \right| \nonumber \\
\rate_j & \leq \min_{i_j=1,\ldots,K_j} \frac{1}{2} \log \frac{\left| \Identity +  \Chan_{i_j}^j \sum_{l=j}^M {\Cov{g}}_l \left( \Chan_{i_j}^{j}\right)^T \right|}{\left| \Identity + \Chan_{i_j}^j \sum_{l=j+1}^{M} {\Cov{g}}_l \left( \Chan_{i_j}^{j}\right)^T  \right|}, \quad \forall j = 1,\ldots,M-1
\end{align}
where ${\Cov{g}}_j$ are some positive semidefinite matrices such that $\mat{0} \preceq \sum_{l=1}^{M}{\Cov{g}}_l = \mats{S}$.

The above upper bounds can be attained simultaneously using a joint Gaussian distribution on the tuple,
\begin{align}
\left( \rvec{v}_0 \equiv \emptyset , \rvec{v}_{1}, \ldots, \rvec{v}_{M-1}, \rvec{v}_M \equiv \rvec{x} \right)
\end{align}
as follows:
\begin{align} \label{eq:appendix:achievable_Covariance}
\rvec{v}_j = \rvec{v}_{j-1} + \rvec{u}_{j}
\end{align}
where $\rvec{u}_j \sim \mathcal{N} \left( \mat{0}, {\Cov{g}}_j \right)$ for $j = 1,\ldots, M$, independent of each other, and where ${\Cov{g}}_j$ are positive semidefinite matrices such that $\sum_{l=1}^M {\Cov{g}}_l \preceq \mats{S}$. This concludes the proof of the capacity region.
\end{IEEEproof}

\section{Summary} \label{sec:conclusions}
In this work we extended the ``single crossing point'' property from the scalar setting to the parallel MIMO setting. We have shown three different ``single crossing point'' properties, given in three phases of extension from scalar-to-vector. These properties cannot be trivially deduced from each other. All three emphasize the basic optimality of the Gaussian input distribution in the Gaussian regime. The most general of these properties, given in the third phase, shows a ``single crossing point'' property for each of the eigenvalues of the matrix $\Q(\rvec{x}, \Cov{\rvec{x_G}}, \vectort)$, the difference between the MMSE matrix assuming an arbitrary Gaussian input, and the MMSE matrix assuming an arbitrary input distribution. We demonstrate the applicability of these properties on several information theoretic problems: a proof of a special case of Shannon's vector EPI, a converse proof of the capacity region
of the parallel \emph{degraded} MIMO broadcast channel (BC) under per-antenna power constrains and under covariance
constraint, and a converse proof of the capacity region of the compound parallel \emph{degraded} MIMO BC under
covariance constraint.

An open question is: can we extend the ``single crossing point'' property to the general MIMO channel? Note that, although the optimality of the Gaussian input is known for several MIMO Gaussian multi-terminal problems, we cannot necessarily conclude the existence of a ``single crossing point'' property.
However, the implications of a general ``single crossing point'' property go beyond the specific applications shown here, and are also of interest on their own. 



\appendix

\subsection{Proofs of Lemmas} \label{app:prooflemmas}


\subsubsection{Proof of Lemma \ref{lem:AequivI}} \label{app:AequivI}

Since $\mat{A}$ is positive semidefinite we can always write $\mat{A} = \alpha
\bar{\mat{A}}\bar{\mat{A}}^\T$ such that $\Tr \left( \bar{\mat{A}}
\bar{\mat{A}}^\T \right) = \dim$ and $\alpha \geq 0$. Then, it can be checked that
\begin{IEEEeqnarray}{rCl}
\q_{\mat{A}}(\rvec{x}, \sigma^2, \scalart)
 & = & \frac{\sigma^2}{1 + \sigma^2\scalart} \Tr \left( \mat{A} \right) - \Tr
\left( \mat{A} \MSE{x}(\scalart) \right) \\
 & = & \alpha \left( \dim \frac{\sigma^2}{1 + \sigma^2\scalart} - \Tr \left(
\bar{\mat{A}}^\T \MSE{x}(\scalart) \bar{\mat{A}} \right) \right) \\
 & = & \alpha \left( \dim \frac{\sigma^2}{1 + \sigma^2\scalart} - \Tr \left(
\MSE{{\bar{\mat{A}}^\T x}}(\scalart) \right) \right) \\
 & = & \alpha \q_{\mat{I}_\dim}(\hat{\rvec{x}}, \sigma^2, \scalart)
\label{eq:AequivI}
\label{eq:equiv}
\end{IEEEeqnarray}
where we have defined $\hat{\rvec{x}} = \bar{\mat{A}}^\T \rvec{x}$. Now, from
(\ref{eq:AequivI}) and the fact that $\alpha \geq 0$, the desired result follows.
\begin{flushright}
\IEEEQED
\end{flushright}

\subsubsection{Proof of Lemma \ref{lem:GaussianOnTop}} \label{app:GaussianOnTop}
Let us consider the random vector $\rvec{x} \in \R^\dim$, whose covariance is
given by $\Cov{x}$ and denote its eigenvalues by $\egvaliCov{x}{i}$. Recalling
the model in (\ref{eq:modelSimplest}), it is well known that
$\MSE{x}(\scalart) \preceq \Cov{x} - \scalart\Cov{x} ( \scalart
\Cov{x} + \mat{I}_\dim)^{-1} \Cov{x}$, \cite{MMSE}. Thus, we have that
\begin{align}
\Tr\left( \MSE{x}(\scalart) \right)
 &\leq \Tr \left( \Cov{x} - \scalart\Cov{x} ( \scalart \Cov{x} +
\mat{I}_\dim)^{-1} \Cov{x} \right) \label{eq:linearEstimator1} \\
 &= \sum_{i=1}^{\dim}  \left( \egvaliCov{x}{i} - \frac{\scalart
\egvaliCov{x}{i}^2}{ 1 + \scalart \egvaliCov{x}{i}} \right)
\label{eq:linearEstimator2} \\
 &= \sum_{i=1}^{\dim} \frac{\egvaliCov{x}{i}}{1 + \scalart \egvaliCov{x}{i}}.
\label{eq:linearEstimator3}
\end{align}

Now, realizing that the right hand side in (\ref{eq:linearEstimator3}) is a
Schur-concave function (it follows directly from the concavity of
$\frac{\lambda}{1 + \scalart \lambda}$) and that, from the statement of
Lemma \ref{lem:GaussianOnTop}, we have that $\sum_{i=1}^{\dim} \egvaliCov{x}{i}
\leq \dim \sigma^2$, it follows directly from majorization theory
\cite{marshall1979inequalities} that the right hand side in
(\ref{eq:linearEstimator3}) is maximized when $\egvaliCov{x}{i}$ are uniformly
distributed, \ie, $\egvaliCov{x}{i} = \sigma^2$.
\begin{flushright}
\IEEEQED
\end{flushright}

\subsubsection{Proof of Lemma \ref{lem:dqAdgamma}} \label{app:dqAdgamma}
From the definition in (\ref{eq:defqA}), it follows that
\begin{gather} \label{eq:dqAdgamma_prev}
\Jacob_\scalart \q_{\mat{A}}(\rvec{x}, \sigma^2, \scalart) = -
\frac{\sigma^4}{(1 + \sigma^2\scalart)^{2}} \Tr(\mat{A}) - \Jacob_\scalart \Tr
\left( \mat{A} \MSE{x}(\scalart) \right).
\end{gather}
The expression for $\Jacob_\scalart \Tr \left( \mat{A}
\MSE{x}(\scalart) \right)$ can be computed from the
results in \cite{Palomar2} and applying the chain rule as
\begin{IEEEeqnarray}{rCl}
\Jacob_\scalart \Tr \left( \mat{A}
\MSE{x}(\scalart) \right)
& = & \Jacob_{\MSE{x}(\scalart)} \Tr \left( \mat{A} \MSE{x}(\scalart) \right)
\cdot \Jacob_\Chan \MSE{x}(\scalart) \cdot \Jacob_\scalart \Chan \\
& = & \vecop^\T \big( \mat{A}^\T \big) \Dup{\dim} \big(- 2 \Dup{\dim}^\pinv
\Esp{\CMSE{x}{y} \otimes \CMSE{x}{y}} \big(\mat{I}_\dim \otimes \Chan^\T \big)
\big) \frac{1}{2\sqrt{\scalart}} \vecop(\mat{I}_\dim) \\
& = & - \vecop^\T \big( \mat{A}^\T \big) \Sym{\dim} \Esp{\CMSE{x}{y} \otimes
\CMSE{x}{y}} \vecop(\mat{I}_\dim) \\
& = & -\Tr \left( \mat{A} \Esp{\CMSE{x}{y}^2} \right) \label{eq:dqAdgamma_ap}
\end{IEEEeqnarray}
where we have used that $\Chan = \sqrt{\gamma} \mat{I}_\dim$, $\Sym{\dim}
\Esp{\CMSE{x}{y} \otimes \CMSE{x}{y}} = \Esp{\CMSE{x}{y} \otimes \CMSE{x}{y}}
\Sym{\dim}$, and $\Sym{\dim} \vecop(\mat{I}_\dim) = \vecop (\mat{I}_\dim)$ (see
\cite[App.~A]{Palomar2} for the definitions of the matrices $\Dup{\dim}$ and
$\Sym{\dim}$ and some of their properties).

Plugging (\ref{eq:dqAdgamma_ap}) in (\ref{eq:dqAdgamma_prev}), the desired
result follows.
\begin{flushright}
\IEEEQED
\end{flushright}

\subsubsection{Proof of Lemma \ref{lem:indGaussianOnTop}}
\label{app:indGaussianOnTop}

For any arbitrarily distributed random vector $\rvec{x}$, with zero mean
(assumed w.l.o.g.) and covariance matrix given by $\Cov{x}$, it is well known
that $\MSE{x}(\vectort) \preceq \EM_{G}(\Cov{x}, \vectort)$, from which it follows
that \cite[Obs.~7.1.2]{Horn}
\begin{eqnarray} \label{eq:Lemma1_1}
[ \MSE{x}(\vectort) ]_{ii} \leq [ \EM_{G}(\Cov{x}, \vectort) ]_{ii}
\end{eqnarray}
where we recall that $\EM_{G}(\Cov{x}, \vectort)$ is the MMSE matrix attained
assuming a zero mean Gaussian input with covariance matrix equal to $\Cov{x}$.
Observe that equality in (\ref{eq:Lemma1_1}) is attained if and only if
$\rvec{x} \sim \NZ{\Cov{x}}$.

Furthermore, from the fact that dependence among entries can only improve the
MMSE, we have:
\begin{eqnarray} \label{eq:Lemma1_2}
[ \EM_{G}(\Cov{x}, \vectort) ]_{ii} \leq [
\EM_{G}(\mat{I}_\dim \circ \Cov{x}, \vectort) ]_{ii} = \frac{[\Cov{x}]_{ii}}{1 +
[\Chan(\vectort)]_{ii}^2 [\Cov{x}]_{ii}}
\end{eqnarray}
where $\EM_{G}(\mat{I}_\dim \circ \Cov{x}, \vectort)$ represents the MMSE matrix
when the entries of the input vector are independent Gaussian random
variables (thus, with diagonal covariance matrix). Observe that
equality in (\ref{eq:Lemma1_2}) is obtained if and only if the entries of the
Gaussian distribution in the left hand side are independent.

Now, the desired result follows immediately from the fact that the right
hand side in (\ref{eq:Lemma1_2}) is an increasing function of $[\Cov{x}]_{ii}$.
\begin{flushright}
\IEEEQED
\end{flushright}


\subsubsection{Proof of Lemma \ref{lem:LemmaLowerBound}} \label{app:LemmaLowerBound}

We first provide the derivative of the MMSE with respect to the parameter $\vectort$.
Using equation (\ref{eq:proofThmSingleCrossingInd}), we have
\begin{eqnarray} \label{eq:chainRuleDiagGeneralized}
\Jacob_{\vectort} \left[ \MSE{x}(\vectort) \right]_{ij} & = & \sum_{l}
\Jacob_{ \left[ \Chan(\vectort) \right]_{ll} } \left[ \MSE{x}(\vectort) \right]_{ij}
\left[ \Jacob_{\vectort} \Chan(\vectort) \right]_{ll} \textrm{.} 
\end{eqnarray}
Using the result (\cite[eq.~(131)]{Palomar2}),
\begin{eqnarray} \label{eq:resultPalomar2_generalGaussian}
\Jacob_{\left[ \Chan(\vectort) \right]_{ll}} \left[ \MSE{x}(\vectort) \right]_{ij} & = & - \Esp[1]{
\left[ \CMSE{x}{y} \right]_{jl} \left[ \CMSE{x}{y}
\Chan(\vectort)^\T \right]_{il} + \left[ \CMSE{x}{y} \right]_{il} \left[
\CMSE{x}{y} \Chan(\vectort)^\T\right]_{jl} } \nonumber \\
& = & - \Esp[1]{   \left[ \CMSE{x}{y} \right]_{jl}
\left[\CMSE{x}{y}\right]_{il} \left[ \Chan(\vectort) \right]_{ll}
+ \left[ \CMSE{x}{y} \right]_{il} \left[ \CMSE{x}{y} \right]_{jl}
\left[ \Chan(\vectort) \right]_{ll}  } \nonumber \\
& = & -2 \left[ \Chan(\vectort) \right]_{ll} \mathbb{E} \left\{ \left[\CMSE{x}{y} \right]_{jl} \left[\CMSE{x}{y} \right]_{il} \right\}
\end{eqnarray}
where $\CMSEr{x}{y}$ was defined in (\ref{eq:CMMSE_matrix}).
The second
equality in equation (\ref{eq:resultPalomar2_generalGaussian}) is due to the fact
that $\Chan(\vectort)$ is diagonal. Thus, we can write the derivative of $\left[ \MSE{x}(\vectort) \right]_{ij}$ as
\begin{eqnarray} \label{eq:derivativeEij}
\Jacob_{\vectort} \left[ \MSE{x}(\vectort) \right]_{ij} & = & -2 \sum_{l}  \left[ \Chan(\vectort) \right]_{ll} \Esp[1]{  \left[ \CMSE{x}{y} \right]_{jl} \left[ \CMSE{x}{y} \right]_{il} } \left[ \Jacob_{\vectort} \Chan(\vectort) \right]_{ll} \nonumber \\
& = & -2 \sum_{l} \left[\ChanDChan(\vectort) \right]_{ll} \Esp[1]{ \left[ \CMSE{x}{y} \right]_{jl} \left[ \CMSE{x}{y} \right]_{il} }
\end{eqnarray}
since $\left[ \ChanDChan(\vectort) \right]_{ll} =
\left[ \Chan(\vectort)  \right]_{ll} \left[ \Jacob_{\vectort} \Chan( \vectort) \right]_{ll}$ (\ref{eq:definitionB}). We can put this
expression into a matrix form as follows:
\begin{eqnarray} \label{eq:derivative_matrixE}
\Jacob_{\vectort} \MSE{x}(\vectort) = -2 \sum_{l} \left[ \ChanDChan(\vectort) \right]_{ll}
\Esp[1]{  \left[ \CMSE{x}{y} \right]_l \left[ \CMSE{x}{y} \right]_{l}^\T }
\end{eqnarray}
where $\left[ \CMSEr{x}{y} \right]_{l}$ is the $l^{th}$ column of the
matrix $\CMSEr{x}{y}$. Using the fact that for
a Gaussian input distribution $\CMSEr{x}{y}$ does not depend
on $\rvec{y}$ and thus $\CMSEr{x}{y} = \Esp[1]{ \CMSEr{x}{y}  } = \EM_{G}(\vectort)$
\cite{Palomar2}, we can obtain the following lower bound on the derivative of
the matrix $\Q(\rvec{x}, \Cov{x_G},\vectort)$:
\begin{eqnarray} \label{eq:matrixhDerivative}
\Jacob_{\vectort} \Q(\rvec{x}, \Cov{x_G},\vectort)
& = & 2 \sum_{l} \left[ \ChanDChan(\vectort)  \right]_{ll} \left( \Esp[1] { \left[ \CMSE{x}{y} \right]_l \left[ \CMSE{x}{y} \right]_{l}^\T } - \left[ \EM_G \right]_l \left[ \EM_G \right]_l^\T \right) \nonumber \\
& \succeq & 2 \sum_{l} \left[ \ChanDChan(\vectort)  \right]_{ll} \left( \Esp[1] {
\left[ \CMSE{x}{y} \right]_l } \Esp[1] { \left[ \CMSE{x}{y} \right]_{l} }^\T - \left[ \EM_G \right]_l \left[\EM_G\right]_l^\T \right) \nonumber \\
& = & 2 \sum_{l} \left[ \ChanDChan(\vectort)  \right]_{ll} \left( {\MSE{x}}_l {\MSE{x}}_l^\T  - \left[ \EM_G \right]_l {\left[ \EM_G \right]_l}^\T \right) \nonumber \\
& = & 2 \left( \MSE{x}(\vectort) \ChanDChan(\vectort)
\MSE{x}^\T(\vectort) - \EM_G(\vectort) \ChanDChan(\vectort)
\EM_G^\T(\vectort) \right) \nonumber
\end{eqnarray}
where the inequality is due to Jensen. This concludes the proof of the lemma.
\begin{flushright}
\IEEEQED
\end{flushright}

\subsubsection{Proof of Lemma \ref{lem:simultaneouslyD}} \label{app:simultaneouslyD}

Since $\mat{A}$ and $\mat{B}$ are two general positive semidefinite matrices,
the dimension of the intersection of their null spaces, denoted by $N(\cdot)$
fulfills
\begin{gather}
\opn{dim} N(\mat{A}) \cap N(\mat{B}) = k, \quad  0 \leq k \leq n.
\end{gather}
Let $\{ \vec{u}_1, \ldots, \vec{u}_n \}$ be an orthonormal basis of the
$\dim$-dimensional space such that $\{ \vec{u}_1, \ldots, \vec{u}_k \}$ is an
orthonormal basis of $N(\mat{A}) \cap N(\mat{B})$ and define $\mat{U} =
[\vec{u}_1 \ldots \vec{u}_n]$. We thus have
\begin{gather}
\mat{U}^\T \mat{A} \mat{U} = \left(
\begin{array}{cc}
\mat{0} & \mat{0} \\
\mat{0} & \mat{A}'
\end{array}\right), \quad
\mat{U}^\T \mat{B} \mat{U} = \left(
\begin{array}{cc}
\mat{0} & \mat{0} \\
\mat{0} & \mat{B}'
\end{array}\right),
\end{gather}
where $\mat{A}'$ and $\mat{B}'$ are the non-zero $(\dim-k)\times(\dim-k)$ lower right
square sub-matrices of $\mat{U}^\T \mat{A} \mat{U}$ and $\mat{U}^\T \mat{B}
\mat{U}$, respectively. Observe that now we have $N(\mat{A}') \cap N(\mat{B}') =
\{ \emptyset \}$.

Now, from \cite[Sec.~4.5,Prob.~8(e)]{Horn}, we have that $\mat{A}$ and
$\mat{B}$ are simultaneously diagonalizable by an invertible matrix $\mat{S}$
if and only if $\mat{A}'$ and $\mat{B}'$ are also simultaneously
diagonalizable. Consequently, we have reduced our proof to showing the
simultaneous diagonalization of two positive semidefinite matrices such that the
dimension of the intersection of their null spaces is 0.

From this point, we can thus assume the following:
\begin{align}
 \mat{A} = \widetilde{\mat{A}}^\T \widetilde{\mat{A}} &\succeq \mat{0},
\label{eq:ATA} \\
 \mat{B} = \widetilde{\mat{B}}^\T \widetilde{\mat{B}} &\succeq \mat{0}, \\
 \opn{dim} N(\mat{A}) \cap N(\mat{B}) &=  0. \label{eq:NAcapNB}
\end{align}

The next step is to prove that $\mat{A}$ and $\mat{B}$ have no common isotropic
vector, which is defined in \cite[Def.~1.7.14]{horn:91} as a vector $\vec{x}
\neq \vecs{0}$ such that $\vec{x}^\T \mat{A} \vec{x} = 0$ and $\vec{x}^\T
\mat{B} \vec{x} = 0$ are both simultaneously fulfilled.

Using the expression in (\ref{eq:ATA}), we have that
\begin{gather}
\vec{x}^\T \mat{A} \vec{x} = 0 \Leftrightarrow \vec{x}^\T \widetilde{\mat{A}}^\T
\widetilde{\mat{A}} \vec{x} = 0 \Leftrightarrow \widetilde{\mat{A}} \vec{x} =
0,
\end{gather}
which can also be applied to $\vec{x}^\T \mat{B} \vec{x} = 0$. Consequently, if
a vector $\vec{x}$ fulfills $\vec{x}^\T \mat{A} \vec{x} = 0$ and
$\vec{x}^\T \mat{B} \vec{x} = 0$, we have necessarily that $\vec{x} \in
N(\widetilde{\mat{A}}) \cap N(\widetilde{\mat{B}})$. However, since
$N(\widetilde{\mat{A}}) \subseteq N(\mat{A})$ and, similarly,
$N(\widetilde{\mat{B}}) \subseteq N(\mat{B})$, from (\ref{eq:NAcapNB}) we have
that $\opn{dim} N(\widetilde{\mat{A}}) \cap N(\widetilde{\mat{B}}) = 0$, which
implies that $\mat{A}$ and $\mat{B}$ have no common isotropic
vector. Now, from \cite[Th.~1.7.17]{horn:91} we have that $\mat{A}$
and $\mat{B}$ are simultaneously diagonalizable.
\begin{flushright}
\IEEEQED
\end{flushright}

\subsubsection{Proof of Lemma \ref{lem:amendment}} \label{app:amendment}

For this proof we require the following result:
\begin{lem}
Let's consider two positive semidefinie matrices $\mat{A}$ and $\mat{B}$. Then we have
\begin{gather}
\mu_{\max} ( \mat{A} - \mat{B}) \geq \mu_{\max}( \mat{A}) -
\mu_{\max} ( \mat{B} )
\end{gather}
where we recall that $\mu_{\max}(\mat{A})$ denotes the maximum eigenvalue of
matrix $\mat{A}$.
\end{lem}
\begin{IEEEproof}
The proof follows directly from \cite[Th.~4.3.1]{Horn} recalling that
$\mu_{\min}( - \mat{B}) = - \mu_{\max}( \mat{B})$.
\end{IEEEproof}
Now, it is clear that for a positive semidefinie matrix $\mat{A}$ and two
positive semidefinite diagonal matrices $\mat{D}_1$ and $\mat{D}_2$ we have
that $\mat{D}_i \mat{A} \mat{D}_i \succeq \mat{0}$, $i = 1, 2$. Now,
from the above lemma we have,
\begin{eqnarray}
\mu_{\max} (\mat{D}_1 \mat{A} \mat{D}_1 - \mat{D}_2 \mat{A} \mat{D}_2 )
& \geq & \mu_{\max} (\mat{D}_1 \mat{A} \mat{D}_1) - \mu_{\max}
(\mat{D}_2 \mat{A} \mat{D}_2) \\ \nonumber
&=& \mu_{\max} (\mat{A}^{\frac{1}{2}}\mat{D}_1^2 \mat{A}^{\frac{1}{2}} ) -
\mu_{\max} (\mat{A}^{\frac{1}{2}}\mat{D}_2^2 \mat{A}^{\frac{1}{2}})
\end{eqnarray}
where the last equality follows from \cite[Th.~1.3.20]{Horn} and the remark,
for square matrices, in the paragraph preceding it. Finally, since $\mat{D}_1
\succeq \mat{D}_2 \succeq \mat{0}$ and they are both diagonal, we have that
$\mat{D}_1^2 \succeq \mat{D}_2^2 \succeq \mat{0}$ and, using
\cite[Obs.~7.7.2]{Horn} and \cite[Cor.~7.7.4]{Horn} we can write,
\begin{gather}
\mu_{\max} (\mat{A}^{\frac{1}{2}}\mat{D}_1^2 \mat{A}^{\frac{1}{2}} ) \geq
\mu_{\max} (\mat{A}^{\frac{1}{2}}\mat{D}_2^2 \mat{A}^{\frac{1}{2}})
\end{gather}
from which the desired result follows.
\begin{flushright}
\IEEEQED
\end{flushright}

\subsubsection{Proof of Lemma \ref{lem:LemmaLowerBoundConditioned}} \label{app:LemmaLowerBoundConditioned}

We extend the lower bound derived in Lemma
\ref{lem:LemmaLowerBound} to the conditioned case, that is, we assume $\rvec{u}- \rvec{x} - \rvec{y}$.
From (\ref{eq:derivative_matrixE}), for the conditioned case we have the following:
\begin{eqnarray} \label{eq:derivative_matrixE_conditioned}
\Jacob_{\vectort} \MSE{x|u}(\vectort,\rvecr{u}) & = & -2 \sum_{l} \left[ \ChanDChan(\vectort) \right]_{ll}
\Esp[1] { \left[ \CMSE{\rvec{x}_{\rvecr{u}}}{y} \right]_l \left[ \CMSE{\rvec{x}_{\rvecr{u}}}{y} \right]_l^\T } \nonumber \\
& = & -2 \sum_{l} \left[ \ChanDChan(\vectort) \right]_{ll} \Esp[1] { \left[\CMSE{x}{y,\rvec{u}=\rvecr{u}} \right]_l \left[\CMSE{x}{y,\rvec{u}=\rvecr{u}} \right]_l^\T } \textrm{.}
\end{eqnarray}
Taking expectation according to $\rvec{u}$ on both sides we have:
\begin{eqnarray} \label{eq:derivative_matrixE_conditionedExpectation}
\Jacob_{\vectort} \MSE{x|u}(\vectort) = \Esp[1] {  \Jacob_{\vectort} \MSE{x|u}(\vectort,\rvec{u}) }
=  -2 \sum_{l} \left[ \ChanDChan(\vectort) \right]_{ll} \Esp[1]{ \left[ \CMSE{x}{y,u} \right]_l \left[\CMSE{x}{y,u} \right]_{l}^\T } \textrm{.}
\end{eqnarray}
The derivative of $\Q(\rvec{x}|\rvec{y}, \Cov{x_g}, \vectort)$ is then given by:
\begin{eqnarray} \label{eq:derivativeQ_conditioned}
\Jacob_{\vectort} \Q(\rvec{x}|\rvec{y}, \Cov{x_g}, \vectort) & = & 2 \sum_{l} \left[ \ChanDChan(\vectort) \right]_{ll}
\left( \Esp[1] {  \left[ \CMSE{x}{y,u} \right]_l \left[\CMSE{x}{y,u} \right]_{l}^\T } -  \Esp[1]{  \left[ \CMSE{x_g}{y} \right]_l \left[ \CMSE{x_g}{y} \right]_{l}^\T }\right)  \nonumber \\
& = & 2 \sum_{l} \left[\ChanDChan(\vectort) \right]_{ll} \left( \Esp[1] {
\left[ \CMSE{x}{y,u}\right]_l \left[ \CMSE{x}{y,u} \right]_{l}^\T } -  \left[ \EM_G(\vectort) \right]_l \left[ \EM_G(\vectort) \right]_l^\T \right)  \nonumber \\
& \succeq & 2 \sum_{l} \left[ \ChanDChan(\vectort) \right]_{ll} \left( \Esp[1] {
\left[ \CMSE{x}{y,u} \right]_l } \Esp[1] { \left[ \CMSE{x}{y,u} \right]_{l}^\T }-  \left[\EM_G(\vectort) \right]_l \left[ \EM_G(\vectort)\right]_l^\T \right)  \nonumber \\
& = & 2 \sum_{l} \left[ \ChanDChan(\vectort) \right]_{ll} \left( \left[ \MSE{x|u}(\vectort) \right]_l \left[ \MSE{x|u}(\vectort) \right]_l^\T -  \left[ \EM_G(\vectort) \right]_l \left[ \EM_G(\vectort) \right]_l^\T \right)  \nonumber \\
& = & 2 \left( \MSE{x|u}(\vectort) \ChanDChan(\vectort) \MSE{x|u}^\T(\vectort)
- \EM_G(\vectort) \ChanDChan(\vectort) \EM_G^\T(\vectort)   \right)
\end{eqnarray}
where the inequality is due to Jensen. This completes the proof of the lemma.
\begin{flushright}
\IEEEQED
\end{flushright}

\subsubsection{Proof of Lemma \ref{lem:GaussianExistenceConditioned}} 
\label{app:GaussianExistenceConditionedProof}
We first claim that w.l.o.g. we can restrict the proof to $\Chan(\vectortEq) = \Identity$. This is shown by redefining $\widetilde{\rvec{x}} = \Chan(\vectortEq) \rvec{x}$. Now, if $\Chan(\vectortEq)$ is non-singular, then this redefinition does not change the mutual information \ie, $\Icond{\widetilde{\rvec{x}}}{\rvec{y}(\vectortEq)}{\rvec{u}} = \Icond{{\rvec{x}}}{\rvec{y}(\vectortEq)}{\rvec{u}}$, and requirements \ref{req:covariance} and \ref{req:Q} are preserved under any congruent transformation. If $\Chan(\vectortEq)$ is singular, the problem can first be reduced in size, since $\Chan(\vectortEq)$ is diagonal for all $\vectort$. Thus, from this point on, we will assume $\Chan(\vectortEq) = \Identity$.

We provide a constructive proof, and show how one can build a Gaussian input distribution such that all three requirements are fulfilled. We begin by rewriting requirement \ref{req:covariance} as a condition on the matrix $\Q( \rvec{x}|\rvec{u}, \Cov{x_g}, \vectortEq)$ rather then on the covariance matrix $\Cov{x_g}$. We do so by defining a new matrix, which is the distance of the MMSE matrix $\MSE{x|u}(\vectortEq)$ from the linear MSE matrix $\MSElin{x}(\vectortEq)$. We proceed by showing that, there exists a fraction such that, by defining $\Q( \rvec{x}|\rvec{u}, \Cov{x_g}, \vectortEq)$ to be that fraction of the newly defined matrix, we comply also with requirement \ref{req:MI}.

As explained above, we begin by rewriting requirement \ref{req:covariance} in terms of the matrix $\Q( \rvec{x}|\rvec{u}, \Cov{x_g}, \vectortEq)$.
Requirement \ref{req:Q} is already a requirement on the matrix $\Q( \rvec{x}|\rvec{u}, \Cov{x_g}, \vectortEq)$ and is as follows,
\begin{eqnarray} \label{eq:definingJ}
\Q( \rvec{x}|\rvec{u}, \Cov{x_g}, \vectortEq) = \EM_G(\vectortEq) - \MSE{x|u}(\vectortEq) \succeq \mat{0}.
\end{eqnarray}
The MMSE for the Gaussian input is:
\begin{align} \label{eq:MMSE_G}
\EM_G(\vectortEq) &= \Cov{x_g} - \Cov{x_g}(\Cov{x_g} + \Identity)^{-1} \Cov{x_g} \nonumber \\
&= \Cov{x_g} - \Cov{x_g}(\Cov{x_g} + \Identity)^{-1} (\Cov{x_g} + \Identity) + \Cov{x_g}(\Cov{x_g} + \Identity)^{-1}  \nonumber \\
&= \Cov{x_g}(\Cov{x_g} + \Identity)^{-1} \nonumber \\
&= (\Cov{x_g} + \Identity)(\Cov{x_g} + \Identity)^{-1} - (\Cov{x_g} + \Identity)^{-1} \nonumber \\
&= \Identity - (\Cov{x_g} + \Identity)^{-1} .
\end{align}
From equation (\ref{eq:definingJ}) $\Cov{x_g}$ complies with
the following:
\begin{align} \label{eq:Cg}
(\Cov{x_g} + \Identity)^{-1} = \Identity -
\MSE{x|u}(\vectortEq) - \Q( \rvec{x}|\rvec{u}, \Cov{x_g}, \vectortEq).
\end{align}
Note that the above equation connects $\Q( \rvec{x}|\rvec{u}, \Cov{x_g}, \vectortEq)$ with $\Cov{x_g}$. Thus, given a specific substitution of $\Q( \rvec{x}|\rvec{u}, \Cov{x_g}, \vectortEq)$ we have a complete definition of the Gaussian input distribution.
Similarly, the MMSE assuming an optimal linear estimator of $\rvec{x}$ (only from
$\rvec{y}(\vectortEq)$) is given by:
\begin{align} \label{eq:MMSE_L}
\MSElin{x} = \Identity - (\Cov{x} + \Identity)^{-1}
\end{align}
and we have that,
\begin{eqnarray} \label{eq:LinearUpperBound}
\MSE{x|u}(\vectortEq) \preceq \MSElin{x}(\vectortEq) \quad \forall \vectort \textrm{.}
\end{eqnarray}
Thus, we can define:
\begin{align} \label{eq:definingC}
\mat{C} & \equiv \MSElin{x}(\vectortEq)
- \MSE{x|u}(\vectortEq) = \Identity - (\Cov{x} + \Identity)^{-1} - \MSE{x|u}(\vectortEq) \succeq \mat{0} \nonumber \\
\MSE{x|u}(\vectortEq) & = \Identity - (\Cov{x} +
\Identity)^{-1} - \mat{C}, \quad \mat{C} \succeq
\mat{0} \textrm{.}
\end{align}
Note that $\mat{C}$ is completely defined by the input random vector, $\rvec{x}$.
Inserting (\ref{eq:definingC}) into equation (\ref{eq:Cg}) we have:
\begin{align} \label{eq:Cg2}
(\Cov{x_g} + \Identity)^{-1} & = \Identity - \left[\Identity - (\Cov{x} + \Identity)^{-1} - \mat{C}\right] - \Q( \rvec{x}|\rvec{u}, \Cov{x_g}, \vectortEq) \nonumber \\
& =  (\Cov{x} + \Identity)^{-1} + \mat{C} -
\Q( \rvec{x}|\rvec{u}, \Cov{x_g}, \vectortEq), \quad \Q( \rvec{x}|\rvec{u}, \Cov{x_g}, \vectortEq) \succeq \mat{0}, \quad
\mat{C} \succeq \mat{0} \textrm{.}
\end{align}
We now require the following supporting lemma,
\begin{lem} \label{lem:existenceOfR_gMin}
Assume $\rvec{x} \in \R^{\dim}$ is an arbitrary distributed random vector. For any $t' \in [0, \infty)$
there exists a Gaussian random vector, $\rvec{x_g}$, with covariance matrix $\Cov{x_g}$ such that,
\begin{enumerate}
\item $\mat{0} \preceq \Cov{x_g} \preceq \Cov{x}$ \label{proof:req:Rmin}
\item $\Q( \rvec{x}|\rvec{u}, \Cov{x_g}, \vectort') = \mat{0}$ \label{proof:req:Qeq0}
\item $\Igen{\rvec{x_g}}{\rvec{y_g}(\vectort')} \leq \Icond{\rvec{x}}{\rvec{y}(\vectort')}{\rvec{u}}$ \label{proof:req:MI}
\end{enumerate}
\end{lem}
\begin{IEEEproof}
See Appendix \ref{app:existenceOfR_gMin}.
\end{IEEEproof}
Note that according to Lemma \ref{lem:existenceOfR_gMin} we have that for $\Q( \rvec{x}|\rvec{u}, \Cov{x_g}, \vectortEq) = \mat{0}$ there exists a Gaussian random vector, $\rvec{x_g}$, which ensures $\Igen{\rvec{x_g}}{\rvec{y_g}(\vectortEq)} \leq \Icond{\rvec{x}}{\rvec{y}(\vectortEq)}{\rvec{u}}$.
On the other hand, if $\Q( \rvec{x}|\rvec{u}, \Cov{x_g}, \vectortEq) = \mat{C}$ we have, according to (\ref{eq:Cg2}),
that $\Cov{x_g} = \Cov{x}$ in which case 
we have $\Igen{\rvec{x_g}}{\rvec{y_g}(\vectortEq)} \geq \Icond{\rvec{x}}{\rvec{y}(\vectortEq)}{\rvec{u}}$. Moreover, from (\ref{eq:Cg2}) we can observe that instead of requirement \ref{req:covariance} \ie, $\Cov{x_g} \preceq \Cov{x}$, we may simply require
$\Q( \rvec{x}|\rvec{u}, \Cov{x_g}, \vectortEq) \preceq \mat{C}$ (\ref{eq:Cg2}), thus requirements \ref{req:covariance} and \ref{req:Q} can be written as follows,
\begin{align} \label{app:lem:GaussianExistenceConditioned:eq13}
\mat{0} \preceq \Q( \rvec{x}|\rvec{u}, \Cov{x_g}, \vectortEq) \preceq \mat{C}
\end{align}
where $\mat{C}$ is defined in
equation (\ref{eq:definingC}).
The question is whether there exists such a
$\Q( \rvec{x}|\rvec{u}, \Cov{x_g}, \vectortEq) \preceq \mat{C}$ that will also attain requirement \ref{req:MI} \ie, $\Igen{\rvec{x_g}}{\rvec{y_g}(\vectortEq)} = \Icond{\rvec{x}}{\rvec{y}(\vectortEq)}{\rvec{u}} \equiv
\alpha$. From the above mentioned we know that,
\begin{eqnarray} \label{eq:boundingAlpha}
\frac{1}{2} \log | \Identity + \Cov{x_g}^1| \leq \alpha
\leq \frac{1}{2} \log | \Identity + \Cov{x_g}^2|
\end{eqnarray}
where,
\begin{align} \label{eq:defineCg1}
\Identity + \Cov{x_g}^1 & = \left( (\Cov{x}
+ \Identity)^{-1} + \mat{C} \right)^{-1}  \\
\label{eq:defineCg2} \Identity + \Cov{x_g}^2 & =
\Identity + \Cov{x} .
\end{align}
Thus, (\ref{eq:boundingAlpha}) can be rewritten as:
\begin{eqnarray} \label{eq:boundingAlpha_2}
\frac{1}{2} \log | \left( (\Cov{x} + \Identity)^{-1} +
\mat{C} \right)^{-1}| \leq & \alpha &
\leq \frac{1}{2} \log | \Identity + \Cov{x}| \nonumber \\
\frac{1}{2} \log \frac{ |\Identity|}{|  (\Cov{x} +
\Identity)^{-1} + \mat{C} |} \leq & \alpha & \leq
\frac{1}{2} \log \frac{|\Identity|}{| (\Identity +
\Cov{x})^{-1}|} .
\end{eqnarray}
We now need the following result,
\begin{lem} \label{lem:r_of_nu}
Let's define the function:
\begin{eqnarray}
r(\nu) = \frac{1}{2} \log
\frac{|\mat{A}|}{|  \mat{B} + \mat{\Delta} \nu
|}.
\end{eqnarray}
For $\mat{A} \succ \mat{0}$,
$\mat{B} \succ \mat{0}$ and $\mat{\Delta}
\succeq \mat{0}$, the function, $r(\nu)$ is continuous and monotonically decreasing
in $\nu$ for $0 \leq \nu \leq 1$.
\end{lem}
\begin{IEEEproof}
The proof is similar to the proof
of Lemma 10 in \cite{EkremUlukusMIMO}.
\end{IEEEproof}
In our case we have:
\begin{align}
\label{eq:ALemma10} \mat{A} & = \Identity \succ
\mat{0} \\ \label{eq:BLemma10} \mat{B} & =
(\Identity + \Cov{x})^{-1} \succ \mat{0} \\
\label{eq:DeltaLemma10} \mat{\Delta} & = \mat{C}
\succeq \mat{0}
\end{align}
and,
\begin{eqnarray} \label{eq:boundingAlpha2}
\left.\frac{1}{2} \log \frac{| \mat{A}|}{|\mat{B} + \nu
\mat{\Delta}|} \right|_{\nu = 1} \leq \left. \alpha \leq \frac{1}{2}
\log \frac{| \mat{A}|}{|\mat{B} + \nu
\mat{\Delta}|} \right|_{\nu = 0} .
\end{eqnarray}
Thus, according to Lemma \ref{lem:r_of_nu}, there exists a $\nu^\star$ such that, $r(\nu^\star) = \alpha$. That
is,
\begin{align} \label{eq:alphaEquals}
\alpha & = \frac{1}{2} \log \frac{|\mat{A}|}{|\mat{B} +
\nu^\star \mat{\Delta}|} = \frac{1}{2} \log
\frac{|\Identity|}{| (\Cov{x} + \Identity)^{-1} +
\nu^\star \mat{C} |} \nonumber \\
& = \frac{1}{2} \log | \Identity +
\Cov{x_g^\star} | =  \frac{1}{2} \log \frac{|\Identity|}{
|(\Cov{x} + \Identity)^{-1} + \mat{C} -
\Q( \rvec{x}|\rvec{u}, \Cov{x_g}, \vectortEq)^\star |}
\end{align}
where the last equality is due to equation (\ref{eq:Cg2}).
That is,
\begin{eqnarray} \label{eq:settingJ_star}
\Q( \rvec{x}|\rvec{u}, \Cov{x_g}, \vectortEq)^\star = (1 - \nu^\star) \mat{C}
\end{eqnarray}
and since $0 \leq \nu^\star \leq 1$ we have that $\mat{0} \preceq
\Q( \rvec{x}|\rvec{u}, \Cov{x_g}, \vectortEq)^\star \preceq \mat{C}$, as required. To conclude, we can construct a Gaussian input distribution, complying with all three requirements, as follows,
\begin{eqnarray} \label{eq:concludingCgStar}
\Identity + \Cov{x_g^\star} = \left( (\Cov{x} +
\Identity)^{-1} + \mat{C} - \Q( \rvec{x}|\rvec{u}, \Cov{x_g}, \vectortEq)^\star
\right)^{-1} = \left( (\Cov{x} +
\Identity)^{-1} + \nu^\star \mat{C} \right)^{-1}
\end{eqnarray}
where $\nu^\star$ is derived from the equality in (\ref{eq:alphaEquals}).
This completes the proof of the lemma.
\begin{flushright}
\IEEEQED
\end{flushright}

\subsubsection{Proof of Lemma \ref{lem:existenceOfR_gMin}} \label{app:existenceOfR_gMin}

We first show that there exists a covariance matrix $\Cov{x_g}$ such that requirements \ref{proof:req:Rmin} and \ref{proof:req:Qeq0} are fulfilled. Then, we will show, using contradiction, that requirement \ref{proof:req:MI} is also fulfilled.

First note that, requirement \ref{proof:req:Qeq0} \ie, $\Q( \rvec{x}|\rvec{u}, \Cov{x_g}, \vectort') = \mat{0}$ completely defines $\Cov{x_g}$:
\begin{align} \label{eq:C_g_min}
\Identity - (\Cov{x_g} + \Identity)^{-1} & =  \MSE{x|u}(\vectort') \nonumber \\
(\Identity - \MSE{x|u}(\vectort'))^{-1} - \Identity & =  \Cov{x_g}
\end{align}
where we have used the expression in (\ref{eq:MMSE_G}), and using the expression in (\ref{eq:MMSE_L}) we can show that $\Identity - \MSE{x|u}(\vectort')$ is an invertible matrix since,
\begin{align} \label{eq:invertible}
\MSE{x|u}(\vectort') \preceq \MSElin{x}(\vectort') = \Identity - (\Cov{x} + \Identity)^{-1} \prec \Identity .
\end{align}
We now need to check that the first requirement holds:
\begin{align} \label{eq:regionTop}
\Cov{x_g} = (\Identity - \MSE{x|u}(\vectort'))^{-1} - \Identity & \succeq \mat{0} \nonumber \\
\Identity - \MSE{x|u}(\vectort') & \preceq \Identity \nonumber \\
\mat{0} & \preceq \MSE{x|u}(\vectort')
\end{align}
and,
\begin{align} \label{eq:regionBottom}
\MSE{x|u}(\vectort')  & \preceq \MSElin{x}(\vectort') = \Identity - (\Cov{x} + \Identity)^{-1} \nonumber \\
(\Cov{x} + \Identity)^{-1} & \preceq \Identity - \MSE{x|u}(\vectort')  \nonumber \\
\Cov{x} + \Identity & \succeq (\Identity - \MSE{x|u}(\vectort'))^{-1} \nonumber \\
\Cov{x} - (\Identity - \MSE{x|u}(\vectort'))^{-1} +  \Identity & \succeq \mat{0} \nonumber \\
\Cov{x} & \succeq \Cov{x_g} .
\end{align}
Thus, we have shown that given an arbitrary input we can find the required $\Cov{x_g}$.

We now want to show that $\Igen{\rvec{x_g}}{\rvec{y_g}(\vectort')} \leq \Icond{\rvec{x}}{\rvec{y}(\vectort')}{\rvec{u}}$.
For any Gaussian random vector $\rvec{x_g}^\star$ with covariance $\Cov{x_g}^\star$ such that $\Cov{x_g}^\star \prec \Cov{x_g}$ we have $\EM_G^\star(\vectort') \prec \EM_G(\vectort')$ and thus, $\Q(\rvec{x}|\rvec{u}, \Cov{x_g}^\star, \vectort') \prec \mat{0}$. Using Theorem \ref{thm:BQ}, assuming that we do not have $\ChanDChan(\vectort) = \mat{0}$ for all $0 \leq \vectort \leq \vectort'$\footnote{$\ChanDChan(\vectort) = \mat{0}$ for all $0 \leq \vectort \leq \vectort'$ then all mutual informations equal to zero regardless of the input distribution and the lemma holds trivially.}, and the I-MMSE relationship (\ref{eq:lineIntegral_cond2}) we have,
\begin{eqnarray} \label{eq:initialStatment}
\Igen{\rvec{x_g}^\star}{\rvec{y_g}^\star(\vectort')} < \Icond{\rvec{x}}{\rvec{y}(\vectort')}{\rvec{u}} .
\end{eqnarray}
Now let's assume that,
\begin{eqnarray} \label{eq:toContradict}
\Igen{\rvec{x_g}}{\rvec{y_g}(\vectort')} > \Icond{\rvec{x}}{\rvec{y}(\vectort')}{\rvec{u}} .
\end{eqnarray}
The function $\Igen{\rvec{x_g}}{\rvec{y_g}(\vectort')}$ is continuous in the value of its eigenvalues, since,
\begin{eqnarray} \label{eq:continuityEigenvalues}
\Igen{\rvec{x_g}}{\rvec{y_g}(\vectort')} = \frac{1}{2} \sum_{i=1}^{n} \log \left( 1 + \lambda_i( \Cov{x_g} ) \right) .
\end{eqnarray}
We can construct $\Cov{x_g}^\star$ by reducing by $\epsilon$ the value of all eigenvalues of $\Cov{x_g}$. According to (\ref{eq:toContradict}) we can find a small enough $\epsilon$, such that the following inequality still holds:
\begin{eqnarray} \label{eq:toContradict2}
\Igen{\rvec{x_g}}{\rvec{y_g}(\vectort')} > \Igen{\rvec{x_g}^\star}{\rvec{y}(\vectort')} \geq \Icond{\rvec{x}}{\rvec{y}(\vectort')}{\rvec{u}}
\end{eqnarray}
but this contradicts (\ref{eq:initialStatment}) and by that proves that,
\begin{eqnarray} \label{eq:final}
\Igen{\rvec{x_g}}{\rvec{y_g}(\vectort')} \leq \Icond{\rvec{x}}{\rvec{y}(\vectort')}{\rvec{u}} .
\end{eqnarray}
This concludes the proof of the lemma.
\begin{flushright}
\IEEEQED
\end{flushright}


\subsection{Converse Proof of BC Capacity Under Per-Antenna Constraints for
M-Users} \label{app:converseProofBCPerAntennaMUsers}

We consider the \emph{degraded} parallel Gaussian BC channel:
\begin{eqnarray} \label{eq:appendix:model3}
\rvec{y}_j[\timeIn] & = & \Chan_j \rvec{x}[\timeIn] + \rvec{n}_j[\timeIn] \quad j=1, \ldots,M
\end{eqnarray}
where $\rvec{n}_j[\timeIn]$, $j=1,..,M$ are standard
additive Gaussian noise vectors independent for different time indices $\timeIn$ (and can be considered independent of each other), and $\Chan_j$, $j=1,..,M$ are diagonal positive semidefinite matrices such that
$\Chan_j \preceq \Chan_{j+1}$, for all $j=1,\ldots, M-1$. $\rvec{x} \in \R^{\dim}$ is the random input vector and it is assumed independent for different time indices $\timeIn$.

We consider a per-antenna power constraint:
\begin{eqnarray} \label{eq:appendix:constraintPerAntenna}
\left[ \Esp[1]{\rvec{x}\rvec{x}^\T } \right]_{ii} \leq \Power_i \quad \forall i, 1 \leq i \leq \dim  \text{.}
\end{eqnarray}

Since we have a \emph{degraded} BC, we can use the single-letter expression given in \cite{Comments}:
\begin{align} \label{eq:appendix:degradedCapacityRegion}
\rate_j \leq \Icond{\rvec{v}_j}{\rvec{y}_j}{\rvec{v}_{j-1}} \quad j=1,..,M
\end{align}
where $\rvec{v}_j$ are auxiliary random variables, $\rvec{v}_M \equiv \rvec{x}$, $\rvec{v}_0 \equiv \emptyset$, and the union is over all probability distributions satisfying
\begin{eqnarray} \label{eq:appendix:capacityRegionGaussian_MC}
\rvec{v}_0 - ... - \rvec{v}_{M-1} - \rvec{v}_M - \rvec{x} - \rvec{y}_M - \rvec{y}_{M-1} - ... - \rvec{y}_2 - \rvec{y}_1 \textrm{.}
\end{eqnarray}
This is an extension of the proof given for the two user case. We begin by rewriting the single-letter expression (\ref{eq:appendix:degradedCapacityRegion}) as follows:
\begin{eqnarray} \label{eq:appendix:degradedCapacityRegion_rewrite}
\rate_j \leq \Icond{\rvec{x}}{\rvec{y}_j}{\rvec{v}_{j-1}} - \Icond{\rvec{x}}{\rvec{y}_j}{\rvec{v}_j} \quad j=1,..,M
\end{eqnarray}
and more explicitly:
\begin{eqnarray} \label{eq:appendix:degradedCapacityRegion_explicitly}
\rate_1 & \leq & \Igen{\rvec{x}}{\rvec{y}_1} - \Icond{\rvec{x}}{\rvec{y}_1}{\rvec{v}_1} \nonumber \\
\rate_2 & \leq & \Icond{\rvec{x}}{\rvec{y}_2}{\rvec{v}_1} - \Icond{\rvec{x}}{\rvec{y}_2}{\rvec{v}_2} \nonumber \\
\vdots \nonumber \\
\rate_{M-2} & \leq & \Icond{\rvec{x}}{\rvec{y}_{M-2}}{\rvec{v}_{M-3}} - \Icond{\rvec{x}}{\rvec{y}_{M-2}}{\rvec{v}_{M-2}} \nonumber \\
\rate_{M-1} & \leq & \Icond{\rvec{x}}{\rvec{y}_{M-1}}{\rvec{v}_{M-2}} - \Icond{\rvec{x}}{\rvec{y}_{M-1}}{\rvec{v}_{M-1}} \nonumber \\
\rate_M & \leq & \Icond{\rvec{x}}{\rvec{y}_{M}}{\rvec{v}_{M-1}} \textrm{.}
\end{eqnarray}
According to Lemma \ref{lem:path} (and the remark after this lemma) we can construct a diagonal path such that
\begin{eqnarray} \label{eq:appendix:converseProofPath}
\Chan(\vectort_j) & = & \Chan_{j} \quad j=1,...,M \nonumber \\
\Chan(0) & = & \mat{0}
\end{eqnarray}
with $0 \leq \vectort_1 \leq \vectort_2 \leq \ldots \leq \vectort_M$. Now, assume a distribution $\pdffun_{\{\rvec{v}_0 \equiv \emptyset, \rvec{v}_1,\ldots ,\rvec{v}_{M-1}, \rvec{v}_M \equiv \rvec{x}\}}$ on the tuple $(\rvec{v}_0 \equiv \emptyset, \rvec{v}_1,\ldots ,\rvec{v}_{M-1}, \rvec{v}_M \equiv \rvec{x})$ with covariance matrix $\Cov{\rvec{x}}$. We begin by proving the following lemma:

\begin{lem} \label{lem:perAntennaMusers}
There exist $M$ independent Gaussian inputs $\rvec{x}_{G_j}$, with covariance matrices $\mat{\Lambda}_{G_j}$ such that,
\begin{eqnarray} \label{eq:appendix:goal}
\int_0^{\vectort_j} \BQ( \rvec{x}|\rvec{v}_{j}, \mat{\Lambda}_{G_j}, \tau) \d \tau & = & 0 \nonumber \\
\int_0^{\vectort_{j+1}} \BQ( \rvec{x}|\rvec{v}_{j}, \mat{\Lambda}_{G_j}, \tau) \d \tau & \geq & 0, \quad \forall i \quad \forall j=1,...,M-1
\end{eqnarray}
and such that $\left[\mat{\Lambda}_{G_{j}} \right]_{ii} \leq \left[\mat{\Lambda}_{G_{j-1}} \right]_{ii}$, for $j=2,..,M-1$ and $\left[\mat{\Lambda}_{G_1} \right]_{ii} \leq \left[\Cov{\rvec{x}} \right]_{ii}$.
\end{lem}
\begin{IEEEproof}
We will prove the above using induction.

\emph{The case of $j=1$:}
This is identical to the proof given in Section \ref{ssec:applicationInd}.

\emph{For a general $j$:} We assume the above holds for $j$ and prove for $j+1$.
Due to the Markov relation (\ref{eq:appendix:capacityRegionGaussian_MC}) we have that,
\begin{eqnarray}
\BQ( \rvec{x}|\rvec{v}_{j+1}, \mat{\Lambda}_{G_j}, \vectort) = \BQ( \rvec{x}|\rvec{v}_{j} \rvec{v}_{j+1}, \mat{\Lambda}_{G_j}, \vectort), \quad \forall \vectort
\end{eqnarray}
and thus,
\begin{eqnarray} \label{eq:appendix:proofConverse4}
\BQ( \rvec{x}|\rvec{v}_{j+1}, \mat{\Lambda}_{G_j}, \vectort) = \BQ( \rvec{x}|\rvec{v}_{j} \rvec{v}_{j+1}, \mat{\Lambda}_{G_j}, \vectort) \geq \BQ( \rvec{x}|\rvec{v}_{j}, \mat{\Lambda}_{G_j}, \vectort) \quad \forall \vectort
\end{eqnarray}
where the inequality is, again, due to the Markov relation (\ref{eq:appendix:capacityRegionGaussian_MC}) and the definition of the function $\BQ( \rvec{x}|\rvec{v}_{j}, \mat{\Lambda}_{G_j}, \vectort)$, given in equation (\ref{eq:BQ_cond}).
This provides us with the following inequality,
\begin{eqnarray} \label{eq:appendix:proofConverse5}
\int_0^{\vectort_{j+1}} \BQ( \rvec{x}|\rvec{v}_{j+1}, \mat{\Lambda}_{G_j}, \tau) \d \tau \geq \int_0^{\vectort_{j+1}} \BQ( \rvec{x}|\rvec{v}_{j}, \mat{\Lambda}_{G_j}, \tau) \d \tau \geq 0
\end{eqnarray}
where the first inequality is due to (\ref{eq:appendix:proofConverse4}) and the second is due to the induction assumption on $j$ (\ref{eq:appendix:goal}).
Again, following the same derivation as in the proof in Section \ref{ssec:applicationInd}, we know that there exists an independent Gaussian input with covariance $\mat{\Lambda}_{G_{j+1}}$ such that,
\begin{eqnarray} \label{eq:appendix:proofConverse6}
\int_0^{\vectort_{j+1}} \BQ( \rvec{x}|\rvec{v}_{j+1}, \mat{\Lambda}_{G_{j+1}}, \tau) \d \tau & =  & 0 \\ \label{eq:appendix:proofConverse7}
\int_0^{\vectort'} \BQ( \rvec{x}|\rvec{v}_{j+1}, \mat{\Lambda}_{G_{j+1}}, \tau) \d \tau & \geq  & 0
\quad \forall \vectort' > \vectort_{j+1}, \quad \forall i
\end{eqnarray}
where (\ref{eq:appendix:proofConverse7}) is true specifically for $\vectort' = \vectort_{j+2}$. Finally, from (\ref{eq:appendix:proofConverse6}) and (\ref{eq:appendix:proofConverse5}) and the monotonically increasing property of $\BQ( \rvec{x}|\rvec{v}_{j+1}, \mat{\Lambda}_{G}, \tau)$ in $\left[\mat{\Lambda}_{G} \right]_{ii}$ (fourth property of Corollary \ref{cor:BQ_cond}), and the fact that it is independent of all other entries in $\mat{\Lambda}_{G}$, we can conclude that
$\left[\mat{\Lambda}_{G_{j+1}} \right]_{ii} \leq \left[\mat{\Lambda}_{G_j}\right]_{ii}$.
This concludes the proof of the induction.
\end{IEEEproof}
Now, inserting the above bounds (\ref{eq:appendix:goal}) (with the addition of the trivial bound on $\Igen{\rvec{x}}{\rvec{y}_1}$, under the per-antenna constraint (\ref{eq:appendix:constraintPerAntenna})) into the single-letter expression in (\ref{eq:appendix:degradedCapacityRegion_explicitly}) we obtain the following outer bound
\begin{eqnarray} \label{eq:appendix:explicitly_outerBound}
\rate_1 & \leq & \frac{1}{2} \log | \Identity + \Chan_1 \mats{P} \Chan_1^T | - \frac{1}{2} \log | \Identity + \Chan_1 \mat{\Lambda}_{G_1} \Chan_1^T | \nonumber \\
\rate_2 & \leq & \frac{1}{2} \log | \Identity + \Chan_2 \mat{\Lambda}_{G_1} \Chan_2^T | - \frac{1}{2} \log | \Identity + \Chan_2 \mat{\Lambda}_{G_2} \Chan_2^T | \nonumber \\
\vdots \nonumber \\
\rate_{M-2} & \leq & \frac{1}{2} \log | \Identity + \Chan_{M-2} \mat{\Lambda}_{G_{M-3}} \Chan_{M-2}^T | - \frac{1}{2} \log | \Identity + \Chan_{M-2} \mat{\Lambda}_{G_{M-2}} \Chan_{M-2}^T | \nonumber \\
\rate_{M-1} & \leq & \frac{1}{2} \log | \Identity + \Chan_{M-1} \mat{\Lambda}_{G_{M-2}} \Chan_{M-1}^T | - \frac{1}{2} \log | \Identity + \Chan_{M-1} \mat{\Lambda}_{G_{M-1}} \Chan_{M-1}^T | \nonumber \\
\rate_M & \leq & \frac{1}{2} \log | \Identity + \Chan_{M} \mat{\Lambda}_{G_{M-1}} \Chan_{M}^T | \textrm{.}
\end{eqnarray}
where $\mats{P}$ is a diagonal matrix with $\left[ \mats{P}\right]_{ii} = P_i$, and $\mat{\Lambda}_{G_j}$ are positive semidefinite diagonal matrices such that $\mat{0} \preceq \mat{\Lambda}_{G_M} \preceq \mat{\Lambda}_{G_{M-1}} \preceq \ldots \preceq \mat{\Lambda}_{G_2} \preceq \mat{\Lambda}_{G_1} \preceq \boldsymbol{P}$. The
achievability of this outer bound is well-known using superposition coding.
\begin{flushright}
\IEEEQED
\end{flushright}

\subsection{Converse Proof of BC Capacity Under Covariance Constraints for
M-Users}
\label{app:converseProofBCCovarianceMUsers}

We consider the same setting as in Appendix \ref{app:converseProofBCPerAntennaMUsers}, given in (\ref{eq:appendix:model3}), but now with a covariance constraint,
\begin{align} \label{appendix:BCcovarianceConstraint}
\Cov{x} \preceq \mat{S}
\end{align}
where $\mat{S}$ is some positive definite matrix.

As in Appendix \ref{app:converseProofBCPerAntennaMUsers}, since we have a \emph{degraded} BC, we can use the single-letter expression given explicitly in equation (\ref{eq:appendix:degradedCapacityRegion_explicitly}), with auxiliary random variables complying with the Markov chain as detailed in (\ref{eq:appendix:capacityRegionGaussian_MC}). Furthermore, we construct a path as was done in equation (\ref{eq:appendix:converseProofPath}). Now, assume distribution $\pdffun_{\{\rvec{v}_0 \equiv \emptyset, \rvec{v}_1,\ldots ,\rvec{v}_{M-1}, \rvec{v}_M \equiv \rvec{x}\}}$ on the tuple $(\rvec{v}_0 \equiv \emptyset, \rvec{v}_1,\ldots ,\rvec{v}_{M-1}, \rvec{v}_M \equiv \rvec{x})$ with covariance matrix $\Cov{\rvec{x}}$. We begin by proving the following lemma,

\begin{lem} \label{lem:appendix:covarianceMUserInduction}
There exist $M$ Gaussian inputs $\rvec{x}_{G_j}$, with covariance matrices ${\Cov{x_g}}_j$ such that,
\begin{align} \label{eq:appendix:goal_covariance}
\int_0^{\vectort_j} \Tr \left(
\ChanDChan(\tau) \Q(\rvec{x}|\rvec{v}_j,{\Cov{x_g}}_j,\tau) \right) \d \tau & = 0 \nonumber \\
\int_0^{\vectort_{j+1}} \Tr \left(
\ChanDChan(\tau) \Q(\rvec{x}|\rvec{v}_{j},{\Cov{x_g}}_j,\tau) \right) \d \tau & \geq 0, \quad \quad \forall j=1,\ldots,M-1
\end{align}
and such that $\mat{0} \preceq {\Cov{x_g}}_j \preceq {\Cov{x_g}}_{j-1}$, for $j=2,\ldots,M$ and $\mat{0} \preceq  {\Cov{x_g}}_{1} \preceq \mats{S}$. Furthermore, $\Q( \rvec{x}|\rvec{v}_{j},{\Cov{x_g}}_j,\vectort_j) \succeq \mat{0}$, for all $j = 1,\ldots, M$.
\end{lem}
\begin{IEEEproof}
We will prove the above using induction.

\emph{The case of $j=1$:}
This is identical to the proof given in Section \ref{ssec:BCcovarianceConstraint}.

\emph{For a general $j$:} We assume the above holds for $j$ and prove for $j+1$.
Due to the Markov relation (\ref{eq:appendix:capacityRegionGaussian_MC}) we have that,
\begin{align}
\EM_{ \rvec{x} | \rvec{v}_{j+1}}(\vectort) = \EM_{ \rvec{x} | \rvec{v}_{j+1},\rvec{v}_{j}}(\vectort) \preceq \EM_{ \rvec{x} | \rvec{v}_{j}}(\vectort), \quad \forall \vectort
\end{align}
from which we can conclude that,
\begin{eqnarray}
\Q( \rvec{x}|\rvec{v}_{j+1},{\Cov{x_g}}_j,\vectort) = \Q( \rvec{x}|\rvec{v}_{j} \rvec{v}_{j+1},{\Cov{x_g}}_j,\vectort), \quad \forall \vectort
\end{eqnarray}
and thus,
\begin{eqnarray}  \label{eq:appendix:proofConverse3_covariance}
\Q( \rvec{x}|\rvec{v}_{j+1},{\Cov{x_g}}_j,\vectort) = \Q( \rvec{x}|\rvec{v}_{j} \rvec{v}_{j+1},{\Cov{x_g}}_j,\vectort) \succeq \Q( \rvec{x}|\rvec{v}_{j} ,{\Cov{x_g}}_j,\vectort), \quad \forall \vectort .
\end{eqnarray}
Since $\ChanDChan( \vectort)$ is a diagonal positive semidefinite matrix for all $\vectort$, this leads to,
\begin{eqnarray} \label{eq:appendix:proofConverse4_covariance}
\Tr \left( \ChanDChan( \vectort) \Q( \rvec{x}|\rvec{v}_{j+1},{\Cov{x_g}}_j,\vectort) \right) \geq \Tr \left( \ChanDChan( \vectort) \Q( \rvec{x}|\rvec{v}_{j} ,{\Cov{x_g}}_j,\vectort) \right), \quad \forall \vectort .
\end{eqnarray}
Now, taking into account the induction assumptions on $j$, together with (\ref{eq:appendix:proofConverse3_covariance}) and (\ref{eq:appendix:proofConverse4_covariance}) we have,
\begin{align} \label{eq:appendix:proofConverse5_covariance}
\Q( \rvec{x}|\rvec{v}_{j+1},{\Cov{x_g}}_j,\vectort_{j+1}) & \succeq \Q( \rvec{x}|\rvec{v}_{j} ,{\Cov{x_g}}_j,\vectort_{j+1}) \succeq \mat{0} \nonumber \\
\int_0^{\vectort_{j+1}} \Tr \left( \ChanDChan( \tau) \Q( \rvec{x}|\rvec{v}_{j+1},{\Cov{x_g}}_j,\tau) \right) \d \tau & \geq \int_0^{\vectort_{j+1}} \Tr \left( \ChanDChan( \tau) \Q( \rvec{x}|\rvec{v}_{j} ,{\Cov{x_g}}_j,\tau) \right) \d \tau \geq 0
\end{align}
which can also be written as:
\begin{align} \label{eq:appendix:proofConverse5_covariance}
\Q( \rvec{x}|\rvec{v}_{j+1},{\Cov{x_g}}_j,\vectort_{j+1}) & \succeq \mat{0} \nonumber \\
\Igen{ \rvec{x_g}_j}{\rvec{y_g}_j(\vectort_{j+1})} & \geq \Icond{\rvec{x}}{\rvec{y}(\vectort_{j+1})}{\rvec{v}_{j+1}}.
\end{align}
These are the two conditions required for Lemma \ref{lem:GaussianExistenceExtensionConditioned}, with $\rvec{x_g}^{ub} \equiv \rvec{x_g}_j$. Thus, according to Lemma \ref{lem:GaussianExistenceExtensionConditioned} there exists a Gaussian random vector with covariance ${\Cov{x_g}}_{j+1}$ such that,
\begin{enumerate}

\item ${\Cov{x_g}}_{j+1} \preceq {\Cov{x_g}}_{j}$ \label{apendix:applicationCovarianceMuser_prop1}

\item $\Igen{ \rvec{x_g}_{j+1}}{\rvec{y_g}_{j+1}(\vectort_{j+1})} = \Icond{\rvec{x}}{\rvec{y}(\vectort_{j+1})}{\rvec{v}_{j+1}}$ \label{apendix:applicationCovarianceMuser_prop2}

\item $\Q( \rvec{x}|\rvec{v}_{j+1},{\Cov{x_g}}_{j+1},\vectort_{j+1}) \succeq \mat{0}$ \label{apendix:applicationCovarianceMuser_prop3}

\end{enumerate}
Property \ref{apendix:applicationCovarianceMuser_prop2} is equivalent to,
\begin{align}
\int_0^{\vectort_{j+1}} \Tr \left(
\ChanDChan(\tau) \Q(\rvec{x}|\rvec{v}_{j+1},{\Cov{x_g}}_{j+1},\tau) \right) \d \tau & = 0
\end{align}
and from property \ref{apendix:applicationCovarianceMuser_prop3}, Corollary \ref{cor:CorollarySingleCrossingEig}, and Theorem \ref{thm:BQ} we can conclude the following:
\begin{align}
\int_0^{\vectort_{j+2}} \Tr \left(
\ChanDChan(\tau) \Q(\rvec{x}|\rvec{v}_{j+1},{\Cov{x_g}}_{j+1},\tau) \right) \d \tau & = \int_0^{\vectort_{j+1}} \Tr \left(
\ChanDChan(\tau) \Q(\rvec{x}|\rvec{v}_{j+1},{\Cov{x_g}}_{j+1},\tau) \right) \d \tau \nonumber \\
& + \int_{\vectort_{j+1}}^{\vectort_{j+2}} \Tr \left(
\ChanDChan(\tau) \Q(\rvec{x}|\rvec{v}_{j+1},{\Cov{x_g}}_{j+1},\tau) \right) \d \tau \nonumber \\
& = 0 + \int_{\vectort_{j+1}}^{\vectort_{j+2}} \Tr \left(
\ChanDChan(\tau) \Q(\rvec{x}|\rvec{v}_{j+1},{\Cov{x_g}}_{j+1},\tau) \right) \d \tau  \geq 0.
\end{align}
Together with property \ref{apendix:applicationCovarianceMuser_prop1}, this concludes the proof of the induction.
\end{IEEEproof}

Lemma \ref{lem:appendix:covarianceMUserInduction} provides us with Gaussian random vectors $\rvec{x_g}_j$ with covariance matrices ${\Cov{x_g}}_j$ with the following properties:
\begin{enumerate}

\item $\mat{0} \preceq {\Cov{x_g}}_j \preceq {\Cov{x_g}}_{j-1}$, for $j=2,\ldots,M-1$ and $\mat{0} \preceq  {\Cov{x_g}}_{1} \preceq \mats{S}$.

\item $\Icond{ \rvec{x}}{\rvec{y}_j}{ \rvec{v}_j} = \frac{1}{2} \log \left| \Identity + \Chan_j {\Cov{x_g}}_j \Chan_j^\T  \right|$, for $j=1,\ldots,M-1$.

\item $\Icond{ \rvec{x}}{\rvec{y}_{j+1}}{ \rvec{v}_j} \leq \frac{1}{2} \log \left| \Identity + \Chan_{j+1} {\Cov{x_g}}_j \Chan_{j+1}^\T  \right|$, for $j=1,\ldots,M-1$

\end{enumerate}
Substituting these results into the single-letter expression (\ref{eq:appendix:degradedCapacityRegion_explicitly}), and defining,
\begin{align} \label{appendix:applicationCovariance_definingC}
{\Cov{g}}_1 & = \mats{S} - {\Cov{x_g}}_1 \nonumber \\
{\Cov{g}}_j & = {\Cov{x_g}}_{j-1} - {\Cov{x_g}}_{j}, \quad \forall j=2,\ldots,M-1 \nonumber \\
{\Cov{g}}_M & = {\Cov{x_g}}_{M-1}
\end{align}
provides the following upper bound,
\begin{eqnarray} \label{eq:appendix:capacityRegion_conclude_Covariance}
\rate_M & \leq & \frac{1}{2} \log \left| \Chan_M {\Cov{g}}_M \left(\Chan_{M}\right)^\T + \Identity \right| \nonumber \\
\rate_j & \leq & \frac{1}{2} \log \frac{\left| \Chan_j \sum_{l=j}^M {\Cov{g}}_l \left(\Chan_{j}\right)^\T + \Identity \right|}{\left| \Chan_j \sum_{l=j+1}^{M}{\Cov{g}}_l \left(\Chan_{j}\right)^\T + \Identity \right|}, \quad \forall j = 1,\ldots,M-1
\end{eqnarray}
where ${\Cov{g}}_j$ are some positive semidefinite matrices such that $\mat{0} \preceq \sum_{l=1}^{M}{\Cov{g}}_l = \mats{S}$. This completes the converse proof.

The above upper bounds can be attained simultaneously using a joint Gaussian distribution on the tuple,
\begin{align}
\left( \rvec{v}_0 \equiv \emptyset , \rvec{v}_{1}, \ldots, \rvec{v}_{M-1}, \rvec{v}_M \equiv \rvec{x} \right)
\end{align}
as follows:
\begin{eqnarray} \label{eq:appendix:achievable_Covariance}
\rvec{v}_j = \rvec{v}_{j-1} + \rvec{u}_{j}
\end{eqnarray}
where $\rvec{u}_j \sim \mathcal{N} \left( \mat{0}, {\Cov{g}}_j \right)$ for $j = 1, ..., M$, independent of each other, and where ${\Cov{g}}_j$ are positive semidefinite matrices such that $\sum_{l=1}^M {\Cov{g}}_l \preceq \mat{S}$.
Thus, we attain the upper bounds for $j=1,...,M-1$ as follows:
\begin{eqnarray} \label{eq:appendix:achievable2_Covariance}
\rate_j & \leq & \Icond{\rvec{v}_{j}}{\rvec{y}_j}{\rvec{v}_{j-1}} \nonumber \\
& = & \Icond{ \rvec{x}}{\rvec{y}_j}{\rvec{v}_{j-1}} - \Icond{ \rvec{x}}{ \rvec{y}_j}{ \rvec{v}_{j}} \nonumber \\
& = & \frac{1}{2} \log \left| \Chan_j \sum_{l=j}^{M} {\Cov{g}}_l \left(\Chan_j \right)^T  + \Identity\right| - \frac{1}{2} \log \left| \Chan_j \sum_{l=j+1}^{M} {\Cov{g}}_l \left(\Chan_j\right)^T  + \Identity \right| \textrm{.}
\end{eqnarray}
For $j=M$ we obtain the following:
\begin{eqnarray} \label{eq:appendix:achievable3_Covariance}
\rate_M & \leq & \Icond{\rvec{x}}{ \rvec{y}_M}{ \rvec{v}_{M-1}} \nonumber \\
& = & \frac{1}{2} \log \left| \Chan_M {\Cov{g}}_M \left(\Chan_M\right)^T + \Identity \right|  \textrm{.}
\end{eqnarray}
Thus, we have shown that (\ref{eq:appendix:capacityRegion_conclude_Covariance}) is the capacity region under the covariance constraint.
\begin{flushright}
\IEEEQED
\end{flushright}

\subsection{Proof of Lemma \ref{lem:LemmaCompoundMemoryless}} \label{appendix:CompoundMemoryless}

The proof of this lemma follows the proof of \cite[Lem.~4]{HananCompound}, which is very similar to the well known proof for the capacity region of a \emph{degraded} BC in \cite{CoverThomas}. The proof of the direct part relies on successive decoding at the stronger user and is practically identical to that found in \cite{CoverThomas}. We will detail the converse proof only.

Let $\bar{\rvec{y}}_{i_j}^j$ denote a sequence of $\dim$ channel outputs of the $i_j$'th realization of user $j$ and let $W_j$ for $j= 1,\ldots,M$ denote the message indices. Furthermore, let $\rvec{y}_{i_j}^j(l)$ be the $l$'th sample of $\bar{\rvec{y}}_{i_j}^j$  and $\rvec{y}_{i_j}^j(1,\ldots,l-1)$ be the set of all samples up to $l-1$ (including). We use similar notation for all other random variables.
As the capacity region depends only on the marginals $\pdffun_{\rvec{y}_{i_j}^j|\rvec{x}}$ we may assume without loss of generality that indeed the mutual distribution is such that
\begin{eqnarray} \label{eq:appendix:compoundMemorylessMC}
(W_1,\ldots,W_M) - \rvec{x} - \rvec{y}_{i_M}^M - \rvec{y}_{M(M-1)}^\star - \rvec{y}_{i_{M-1}}^{M-1} - \rvec{y}_{(M-1)(M-2)}^\star -\ldots - \rvec{y}_{i_{2}}^{2} -\rvec{y}_{21}^\star - \rvec{y}_{i_1}^1
\end{eqnarray}
form a Markov chain for every choice of $i_1,i_2,\ldots,i_M$.

Using Fano's inequality and the fact that $W_j$ are independent messages we can write an upper bound of $\rate_j$ for any $j = 1,\ldots,M$ which holds for every $i_j \in \{1,\ldots, K_j \}$:
\begin{align} \label{eq:capacityRegionMemoryless_proof}
\rate_j & \leq  \frac{1}{n} \Icond{W_j}{\bar{\rvec{y}}_{i_j}^j }{ W_{j-1},\ldots,W_1} + \delta(\dim)  \nonumber \\
& =  \frac{1}{n} \sum_{l=1}^{\dim} \Icond{W_j}{ \rvec{y}_{i_j}^j(l)}{ W_{j-1},\ldots,W_1, \rvec{y}_{i_j}^j(1,\ldots,l-1)} + \delta(\dim) \\
& =  \frac{1}{n} \sum_{l=1}^{\dim} \left( \HdC{ \rvec{y}_{i_j}^j(l)}{ W_{j-1},\ldots,W_1, \rvec{y}_{i_j}^j(1,\ldots,l-1)} - \HdC{ \rvec{y}_{i_j}^j(l) }{ W_{j},\ldots,W_1, \rvec{y}_{i_j}^j(1,\ldots,l-1)}  \right) + \delta(\dim) \nonumber \\ \label{eq:capacityRegionMemoryless_proof_b}
& =  \frac{1}{n} \sum_{l=1}^{\dim} \left( \HdC{ \rvec{y}_{i_j}^j(l)}{ W_{j-1},\ldots,W_1, \rvec{y}_{i_j}^j(1,\ldots,l-1), \rvec{y}_{j(j-1)}^\star(1,\ldots,l-1)}  \right. \nonumber \\
&  \left.- \HdC{ \rvec{y}_{i_j}^j(l) }{ W_{j},\ldots,W_1, \rvec{y}_{i_j}^j(1,\ldots,l-1)} \right) + \delta(\dim) \\ \label{eq:capacityRegionMemoryless_proof_c}
& \leq \frac{1}{n} \sum_{l=1}^{\dim} \left( \HdC{ \rvec{y}_{i_j}^j(l)}{ W_{j-1},\ldots,W_1, \rvec{y}_{i_j}^j(1,\ldots,l-1), \rvec{y}_{j(j-1)}^\star(1,\ldots,l-1)} \right.
\nonumber \\
& \left. - \HdC{ \rvec{y}_{i_j}^j(l)}{ W_{j},\ldots,W_1,\rvec{y}_{(j+1)j}^\star(1,\ldots,l-1), \rvec{y}_{i_j}^j(1,\ldots,l-1)} \right) + \delta(\dim) \\
\label{eq:capacityRegionMemoryless_proof_d}
& = \frac{1}{n} \sum_{l=1}^{\dim} \left( \HdC{ \rvec{y}_{i_j}^j(l)}{ W_{j-1},\ldots,W_1, \rvec{y}_{i_j}^j(1,\ldots,l-1), \rvec{y}_{j(j-1)}^\star(1,\ldots,l-1)} \right.
\nonumber \\
& \left. - \HdC{ \rvec{y}_{i_j}^j(l) }{ W_{j},\ldots,W_1,\rvec{y}_{(j+1)j}^\star(1,\ldots,l-1)} \right) + \delta(\dim) \\
\label{eq:capacityRegionMemoryless_proof_e}
& =  \frac{1}{n} \sum_{l=1}^{\dim} \left( \HdC{ \rvec{y}_{i_j}^j(l)}{ W_{j-1},\ldots,W_1, \rvec{y}_{i_j}^j(1,\ldots,l-1), \rvec{y}_{j(j-1)}^\star(1,\ldots,l-1)} \right.
\nonumber \\
& \left. - \HdC{ \rvec{y}_{i_j}^j(l) }{ W_{j},\ldots,W_1,\rvec{y}_{(j+1)j}^\star(1,\ldots,l-1), \rvec{y}_{j(j-1)}^\star(1,\ldots,l-1)} \right) + \delta(\dim) \\
\label{eq:capacityRegionMemoryless_proof_f}
& \leq \frac{1}{n} \sum_{l=1}^{\dim} \left( \HdC{ \rvec{y}_{i_j}^j(l) }{ W_{j-1},\ldots,W_1, \rvec{y}_{j(j-1)}^\star(1,\ldots,l-1)} \right.
\nonumber \\
&  \left. - \HdC{ \rvec{y}_{i_j}^j(l)}{ W_{j},\ldots,W_1,\rvec{y}_{(j+1)j}^\star(1,\ldots,l-1), \rvec{y}_{j(j-1)}^\star(1,\ldots,l-1)} \right) + \delta(\dim) \\
\label{eq:capacityRegionMemoryless_proof_g}
& = \frac{1}{n} \sum_{l=1}^{\dim} \left( \HdC{ \rvec{y}_{i_j}^j(l)}{ \rvec{v}_{j-1}(l) } - \HdC{ \rvec{y}_{i_j}^j(l) }{ \rvec{v}_{j-1}(l), \rvec{v}_j(l)} \right) + \delta(\dim) \\
& = \frac{1}{n} \sum_{l=1}^{\dim} \Icond{ \rvec{v}_{j}(l)}{ \rvec{y}_{i_j}^j(l) }{ \rvec{v}_{j-1}(l) } + \delta(\dim)
\end{align}
where $\delta(\dim) \to 0$ as $\dim \to \infty$. The equality in (\ref{eq:capacityRegionMemoryless_proof}) is due to the chain rule of mutual information. The equality in (\ref{eq:capacityRegionMemoryless_proof_b}) is due the the Markov chain $(W_1,...,W_M) - \rvec{x} - \rvec{y}_{i_j}^j - \rvec{y}_{j(j-1)}^\star$ and the memoryless nature of the channel, as can be seen in the following identity
\begin{align} \label{eq:capacityRegionMemoryless_proof_addition}
& \pdffun \left\{ \rvec{y}_{i_j}^j(l) | W_{j-1},...,W_1, \rvec{y}_{i_j}^j(1,...,l-1), \rvec{y}_{j(j-1)}^\star(1,...,l-1)  \right\}  \nonumber \\
& = \frac{\pdffun \left\{ \rvec{y}_{i_j}^j(l), \rvec{y}_{j(j-1)}^\star(1,...,l-1) | W_{j-1},...,W_1, \rvec{y}_{i_j}^j(1,...,l-1)  \right\} }{\pdffun \left\{ \rvec{y}_{j(j-1)}^\star(1,...,l-1) | W_{j-1},...,W_1, \rvec{y}_{i_j}^j(1,...,l-1) \right\} } \nonumber \\
& = \frac{\pdffun \left\{\rvec{y}_{i_j}^j(l) | W_{j-1},...,W_1, \rvec{y}_{i_j}^j(1,...,l-1)  \right\} \pdffun \left\{ \rvec{y}_{j(j-1)}^\star(1,...,l-1) | W_{j-1},...,W_1, \rvec{y}_{i_j}^j(1,...,l-1), \rvec{y}_{i_j}^j(l) \right\} }{\pdffun \left\{ \rvec{y}_{j(j-1)}^\star(1,...,l-1) | W_{j-1},...,W_1, \rvec{y}_{i_j}^j(1,...,l-1) \right\} } \nonumber \\
& = \frac{\pdffun \left\{\rvec{y}_{i_j}^j(l) | W_{j-1},...,W_1, \rvec{y}_{i_j}^j(1,...,l-1)  \right\} \pdffun \left\{ \rvec{y}_{j(j+1)}^\star(1,...,l-1) | \rvec{y}_{i_j}^j(1,...,l-1) \right\} }{\pdffun \left\{ \rvec{y}_{j(j+1)}^\star(1,...,l-1) | \rvec{y}_{i_j}^j(1,...,l-1) \right\} } \nonumber \\
& = \pdffun \left\{ \rvec{y}_{i_j}^j(l) | W_{j-1},...,W_1, \rvec{y}_{i_j}^j(1,...,l-1) \right\} \textrm{.} \nonumber
\end{align}
The inequality in (\ref{eq:capacityRegionMemoryless_proof_c}) follows from the fact that conditioning decreases entropy. (\ref{eq:capacityRegionMemoryless_proof_d}) and (\ref{eq:capacityRegionMemoryless_proof_e}) follow, again, form the Markov chain $(W_1,...,W_M) - \rvec{x} - \rvec{y}_{(j+1)j}^\star - \rvec{y}_{i_j}^j - \rvec{y}_{j(j-1)}^\star$ and the memoryless nature of the channel. (\ref{eq:capacityRegionMemoryless_proof_f}) follows from the fact that conditioning decreases entropy. In (\ref{eq:capacityRegionMemoryless_proof_g}) we used the following definition of auxiliary random variables:
\begin{eqnarray} \label{eq:capacityRegionMemoryless_auxiliary}
\rvec{v}_{j}(l) = \left( W_{j},...,W_1, \rvec{y}_{(j+1)j}^\star(1,...,l-1) \right).
\end{eqnarray}
Next we replace the index $l$ with a random variable $I$ which is uniformly distributed over the integers $1,\ldots,\dim$ and define $\rvec{v}_j = \left( \rvec{v}_j(I), I \right), \rvec{x} = \rvec{x}(I), \rvec{y}_{i_j}^j = \rvec{y}_{i_j}^j(I)$. As the channel is memoryless, we get
\begin{eqnarray}
\rate_j \leq \Icond{\rvec{v}_{j}}{ \rvec{y}_{i_j}^j}{ \rvec{v}_{j-1}} + \delta(\dim)
\end{eqnarray}
for all $j$ and for all $i_j$. Note that as the channel is memoryless these auxiliary random variables satisfy the Markov chain defined in (\ref{eq:capacityRegionMemoryless_MC}). Moreover, from this definition one can easily see that $\rvec{v}_0 \equiv \emptyset$ and the largest region will be attained when $\rvec{v}_M \equiv \rvec{x}$. Finally, as the above inequalities hold for every $j = 1, \ldots,M$ and every $i_j = 1,\ldots,K_j$, we complete the proof by taking $\dim$ to infinity.
\begin{flushright}
\IEEEQED
\end{flushright}

\bibliographystyle{IEEEtran}
\bibliography{bib}

\end{document}